\documentclass[a4paper]{article}

\usepackage[utf8]{inputenc}
\usepackage[T1]{fontenc}
\usepackage{textcomp}
\usepackage{amsmath, amssymb, amsthm, physics}
\usepackage{mdframed}
\usepackage[backend=bibtex]{biblatex}
\usepackage{hyperref}

\theoremstyle{plain}
\newtheorem{definition}{Definition}[section]
\newtheorem{thm}{Theorem}[section]
\newtheorem{fact}{Fact}[section]

\newtheorem{cor}{Corollary}[section]
\newtheorem{prop}{Proposition}[section]
\newtheorem{lem}{Lemma}[section]

\usepackage{graphicx}
\graphicspath{ {diagrams/} }
\begin{document}

\title{Gauging the Spacetime Code}

\author{Gideon Lee$^1$}
\date{%
    $^1$Pritzker School of Molecular Engineering, The University of Chicago, Chicago, Illinois 60637, USA%
}


\maketitle

\begin{abstract}
    In recent years, the spacetime code has arisen as a candidate for a unifying view of fault tolerance in space and time. On the other hand, the recent study of dynamical phases has increasingly turned its attention to fault tolerance as a notion of a dynamically stable process. In this work, I explore one pathway between the two, achieved by gauging the spacetime code. This gives rise to a lattice gauge theory that inherits the elements of fault tolerance associated with a circuit, with Gauss laws corresponding to equivalence relations between configurations of spacetime errors and Wilson loops corresponding to detectors. The obtained gauge theory finds a surprisingly wide array of applications, from quantum error correction to condensed matter physics, and even learning theory: (1) It contains in its description foliated computation, and hence gives rise to one version of a gauge theory for measurement-based quantum computation. (2) For a class of topologically ordered mixed states, it gives us a gauge-theoretic language to describe the classical memory associated with the state. (3) The gauge-invariant observables of the theory which describe detectors also coincide with the learnable degrees of freedom of circuit Pauli noise. 
\end{abstract}

\tableofcontents

\section{Introduction}

This work is motivated by an attempt to understand the analogies between the theory of fault tolerance \cite{Terhal_RMP_2015_QEC_review, NC_2019, Gottesman_PRA_1998_FTQC} and dynamical phases of matter \cite{Rakovszky_PRX_2024_stable_open, Sang_PRL_2025_markov_length, Coser_Quantum_2019_mixed_state_phase_equivalence}. The drawing of such analogies has a long history, starting from early works like Ref.~\cite{Dennis_JMP_2002_topological_quantum_memory}, which map fault tolerant circuits for a toric code to statistical mechanical models, and measurement-based quantum computation \cite{Raussendorf_AnnPhys_2006_one_way_QC, Raussendorf_NJP_2007_topological_cluster_states, Raussendorf_PRA_2005_cluster_states, Raussendorf_PRL_2007_FTQC_2D}, which maps fault tolerant computations to resource states. On a high level, the analogy is clear and compelling -- fault tolerant processes should be sufficiently stable against noise that they can preserve quantum information in the thermodynamic limit. Conversely, a dynamical phase of matter should be stable against perturbations (noise) in parameter space such that it preserves crucial features, such as possibly a degeneracy (quantum information) in steady states. 

In recent years, various frameworks have emerged that can be said to formalize this analogy, all of which involve some form of rewriting of circuits, from methods based on ZX calculus \cite{Bombin_quantum_2024_unifying_FT, Rodatz_arXiv_2026_FT_by_construction, vandewetering_arXiv_2020_zx_tutorial}, to topological fixed-point path integrals \cite{Davydova_arXiv_2025_universal_FT, Bauer_quantum_2024_fixed_point_PI}, to spacetime codes \cite{Bacon_STOC_2015_sparse_codes_circuits,Delfosse_arXiv_2023_spacetime_code_clifford,Pesah_arXiv_2025_FT_transformations,Gottesman_arXiv_2022_opportunities_FT}. Each has advantages and disadvantages in the drawing of analogies to phases of matter. The ZX calculus provides a complete description of linear algebra -- naturally it also appears to provide a fairly complete description of fault tolerance \cite{Bombin_quantum_2024_unifying_FT}. However, the language and flexibility of spiders and Pauli flows do not immediately or naturally describe a phase of matter (although very recent work \cite{Zhang_arXiv_2026_quantum_process_LDPC} has managed to do exactly that). In contrast, topological fixed-point path integrals are naturally associated with topological phases of matter. However, they are primarily used to describe the operation of topological codes rather than general circuits. Finally, spacetime codes lie somewhere in between -- they describe a large range of circuits (although restricted to Clifford operations) without assuming explicit topological structure, but at the same time, given their explicit description of a circuit as a quantum error correction code, are also naturally associated with a phase of matter \cite{Zeng_2019_QIMQM}. However, in contrast to the other two formalisms, here the associated code does not necessarily inherit all fault-tolerant properties of the circuit (though it often does) -- in particular, whether the code-space is degenerate depends on the choice of discretization of the circuit, and the distance of the code upper bounds but is not in strict correspondence to the number of spacetime errors the circuit can detect.

The goals and contributions of this work then is two-fold. First, starting from spacetime codes, we will explore one systematic method to construct (loosely speaking) a phase of matter, in particular (strictly speaking) a gauge theory, from a Clifford circuit, which preserves certain quantities associated with fault tolerance. Combining a mild modification of spacetime codes \cite{Bacon_STOC_2015_sparse_codes_circuits,Delfosse_arXiv_2023_spacetime_code_clifford,Pesah_arXiv_2025_FT_transformations,Gottesman_arXiv_2022_opportunities_FT} with the gauging procedures outlined in Refs.~\cite{kubica_2018_arXiv_ungauging_QEC,Rakovszky_2023_arxiv_physics_ldpc_I}, we obtain a circuit-to-code mapping which may be interpreted as a gauge theory. Note that the spacetime code is already a circuit-to-code mapping -- if one wishes, one may associate these codes with Hamiltonians, which may then be given an interpretation as a phase of matter \cite{Aitchison_arXiv_2026_spacetime_spins}. The code we will arrive at, while distinct from the spacetime code, is constructed from the same essential elements, and may be considered dual to the spacetime code\footnote{In the sense that both codes can be regarded as different segments of the same chain complex, see Eq.~\eqref{eq:gauge_complex_2} and the surrounding discussion.}. In contrast to the usual approach, constructing this dual code requires conceptualizing the spacetime code as a classical rather than a quantum code.

The second goal of this work is to clarify how the fault tolerance properties of the circuit are inherited by the gauge theory. Along the way, this will also allow us to review and clarify the mechanisms by which these properties of the circuit are inherited by the spacetime code itself. While couched in the language of gauge theory, we point out that this should be of interest to readers whose primary interest is the spacetime code. Our approach to the spacetime code puts measurement and data errors on equal footing, and explore its consequences. Operationally, this means we allow for multi-qubit mid-circuit measurements, but retain a single ancilla error location for the book-keeping of readout errors\footnote{Note that this is just a choice of discretization of the circuit, and one can recover both circuit-level and phenomenological errors via other choices of discretization.}. Previous approaches have treated measurement errors as higher-weight `sandwiching' errors, i.e. pairs of data errors occurring before and after mid-circuit measurements. At first glance, the approach we take in Sec.~\ref{sec:circuits} may look like nothing more than a change in basis that allows us to properly account for the weights of measurement errors\footnote{In the case where there is only one measurement per time step, this looks like a small difference -- we can always choose a weight-two sandwiching error to represent a measurement error. 
However, for phenomenological models, one may want to characterize multiple measurements in a single time step (see the discussion of circuit discretization in Sec.~\ref{sec:circuits}). In that case, it may take sandwiching Pauli errors of arbitrarily high weight to create a single readout error.}. However, on the fault tolerance side, this also allows to conceptually account for the fact that syndrome outcomes are themselves classical bits being transmitted through a circuit.\footnote{As an example, consider the undetectable error marked in yellow in Fig.~\ref{fig:rep_code_gauge_fixed}.} 

On the gauge theory side, this reveals spatial structure that is obscured by the sandwiching approach. The gauge theory frames the fault tolerance properties of the circuit in the following sense: the matter fields are the spacetime locations in the original circuit, the gauge fields are the propagators and measurements that characterize the action of the circuit, and the (local) Wilson loops are the detectors of the circuit. Importantly, the Wilson loops arise from local redundancies in gauge fields. In contrast to stabilizers of the spacetime code formulated in Ref.~\cite{Pesah_arXiv_2025_FT_transformations}\footnote{As argued in Ref.~\cite{Pesah_arXiv_2025_FT_transformations}, all detectors may be mapped to stabilizers for the subsystem spacetime code, however, there can exist stabilizers that do not map back to detectors. We note that in the earlier formulation of the purely stabilizer spacetime code of Ref.~\cite{Delfosse_arXiv_2023_spacetime_code_clifford}, stabilizers were explicitly constructed only from detectors. However, unlike the subsystem version, the purely stabilizer spacetime code does not try to fully capture other elements of fault tolerance, such as fault equivalence.}, these are always in one-to-one correspondence with the detectors. As we will see, this formally extends the various constructions of the spacetime code regarded as a chain complex. The above are all bulk (temporally speaking) properties. The boundary conditions of the gauge theory will also play a critical role when relating it to quantum error correction (QEC), and we will also show how boundaries may be systematically constructed in correspondence with various practical QEC experiments. We will refer to the obtained gauge theory as the gauged subsystem spacetime code (SSC). 



The rest of this work proceeds as follows: In Sec.~\ref{sec:clifford_circuits_and_gauge_generators}, we begin by introducing a description of Clifford circuits using \textit{elementary circuit operators}, and prove some structural claims about this description, which apply to spacetime codes. In particular, we show that these operators retain all information needed to understand the fault tolerance of a circuit, namely, how operators propagate (Cor.~\ref{cor:gauge_op_ISG}), the effect and equivalence of spacetime faults (Prop.~\ref{prop:fault_equiv}), and detectors and their relationship to spacetime errors (Prop.~\ref{prop:detectors_and_errors_I}). In Sec.~\ref{sec:introducing_SSC}, we discuss how and when these elements show up in the spacetime code. The material in this section is a mix of review of basic ideas in fault tolerance and spacetime codes, but also proves some simple results that serve to concretely ground the elementary circuit operators in the language of fault tolerance. Along the way, we try to unify some basic notions and language appearing in Refs.~\cite{Delfosse_arXiv_2023_spacetime_code_clifford, Blackwell_arXiv_2025_distance_floquet, Pesah_arXiv_2025_FT_transformations}. In Sec.~\ref{sec:gauging_SSC}, we introduce the gauging procedure, which we carry out on the elementary circuit operators to obtain a formal description of the gauge theory associated with a Clifford circuit. In Sec.~\ref{ssec:elements_FT_gauged_SSC}, we show that the same elements of fault tolerance are retained, and correspond to various elements of gauge theory. Moreover, we prove that detectors are in one-to-one correspondence with the gauge invariant observables of this theory, which are inherited from the local redundancies of the elementary circuit operators (see Prop.~\ref{prop:detectors_are_redundancies}). This gauge theory is given a topological interpretation in Sec.~\ref{ssec:topology_gauged_SSC}, which allows us to define and prove a notion of an internal distance, which describes the error correction of spacetime errors using only detectors. In Sec.~\ref{ssec:boundaries}, we discuss temporal boundary conditions, and their correspondence to various operational QEC experiments, and in Sec.~\ref{ssec:examples}, we discuss some examples of gauged SSCs. Finally, in Sec.~\ref{sec:applications}, we take a step away from formalism to look at various applications of the gauged SSC. While we do not propose new circuits or schemes, we use the gauged SSC to provide reinterpretations of circuits and schemes going beyond just error correction, from measurement-based quantum computation in Sec.~\ref{ssec:SSC_foliated_computation}, to mixed state order in Sec.~\ref{ssec:mixed_state_order}, to Pauli noise learning in Sec.~\ref{ssec:learning_Pauli_noise}.






\section{Clifford circuits, Elementary Circuit Operators, and the Spacetime Code}\label{sec:clifford_circuits_and_gauge_generators}

We begin our work by identifying a Clifford circuit $\mathcal{C}$  with a set of spacetime locations and operators, which we will refer to as the \textit{elementary circuit operators} (ECOs), or  $G_{\mathcal{C}}$. We will set up notation for a Clifford circuit in Sec.~\ref{sec:circuits}, and define the ECOs in Sec.~\ref{sec:ECOs}. Note that a similar set of operators are referred to as gauge operators in the usual spacetime code literature \cite{Bacon_STOC_2015_sparse_codes_circuits, Pesah_arXiv_2025_FT_transformations, Fu_unpub_subsystem_spacetime_code}, or as benign errors\footnote{We discuss the relationship between benign errors and gauge operators in Sec.~\ref{sssec:floquet_codes}} in Ref.~\cite{Blackwell_arXiv_2025_distance_floquet}. We refer to them as `circuit operators' for two reasons: Primarily, the word `gauge' is already at risk of becoming dangerously overloaded by the end of this work. Conceptually, we also insist that for now, we do not view $G_{\mathcal{C}}$ as the gauge group of anything, but rather a collection of operators. In particular, one is not allowed to refactor these operators by say replacing operators with products of operators.

To that end, we view these operators as simply an alternative expression of the circuit. All quantities of interest related to the fault tolerance and error correction properties of $\mathcal{C}$ may be derived from $G_{\mathcal{C}}$. These quantities include the action of the circuit, as described in Sec.~\ref{ssec:ISGs}, as well as the effect (Def.~\ref{def:bare_effect}) and detectors (Def.~\ref{def:detector}) of a circuit, as described in Sec.~\ref{ssec:elts_of_FT}. These elements together suffice to define (Def.~\ref{def:fault_distance}) and determine the spacetime fault distance of the circuit. In other words, for our purposes, there is no loss of information when going from $\mathcal{C}$ to $G_{\mathcal{C}}$.

\subsection{Discretizing a Clifford circuit}\label{sec:circuits}

We recall the definition of a Clifford operation and circuit, using conventions developed in Refs.~\cite{Bacon_STOC_2015_sparse_codes_circuits, Delfosse_arXiv_2023_spacetime_code_clifford} to standardize notation and setup. The following definitions are schematically illustrated in Fig.~\ref{fig:circuit_convention}, with examples.

\begin{figure}[ht!]
\centering 
\includegraphics[width=0.95\linewidth]{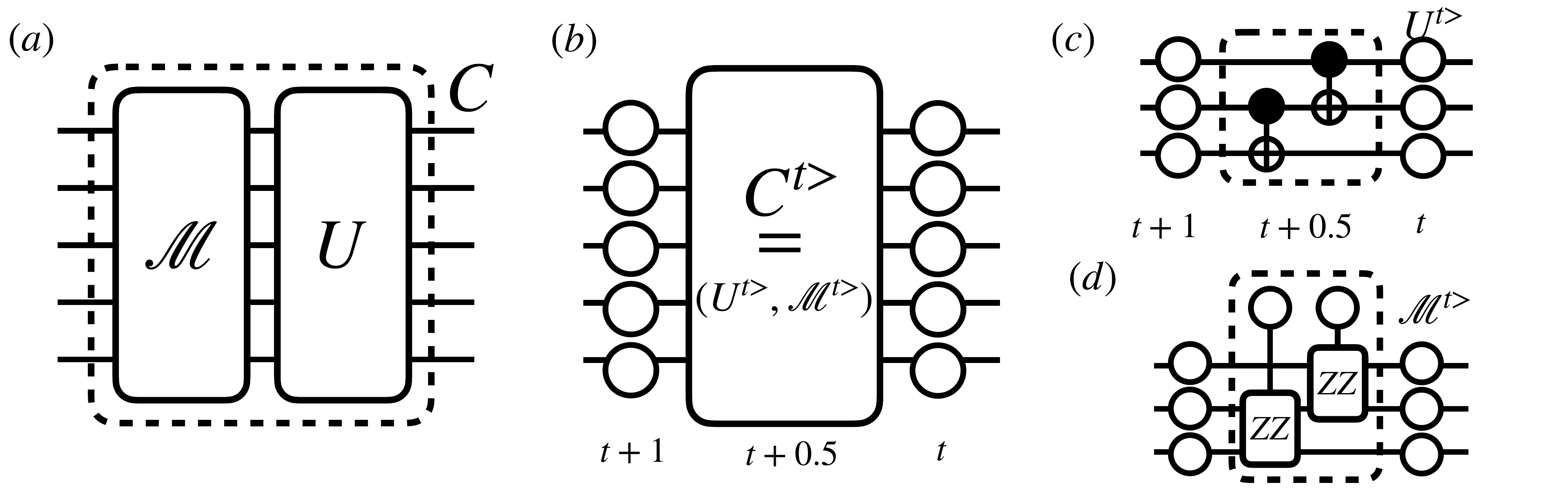} \caption{Schematic illustrating the conventions used in this work. In all figures, time goes from right to left. \textbf{(a)} A Clifford operation $C$ comprises a Clifford unitary $U$ and a set of measurements $\mathcal{M}$. \textbf{(b)} We discretize spacetime locations in a Clifford circuit by grouping its elements up into Clifford operations, and placing these at half-integer time-steps. We insert spacetime locations, depicted by circles, at integer time-steps. \textbf{(c)} An example of a Clifford unitary: We allow the unitaries at each time-step to be arbitrarily complicated -- in this case the unitary comprises two cNOT gates. We note that this circuit could also be differently discretized by separating the two cNOTs into different time-steps. We allow for these different possibilities in order to be able to model the full spectrum from fully phenomenological to circuit-level noise. \textbf{(d)} An example of a measurement. In this case, the measurement comprises two $ZZ$ measurements. Our only requirement for grouping measurements into time-steps is that the measurements in each time-step should commute. Each measurement comes with an additional spacetime location which is placed at half-integer time-steps. Spacetime locations at integer time-steps are expected to be susceptible to both bit and phase flip errors. In contrast, the spacetime locations at half-integer time-steps are regarded as classical, in the sense that they are only susceptible to one kind, say phase flip, of error.}\label{fig:circuit_convention} 
\end{figure}

\textbf{Clifford operations -- } A \textit{Clifford operation} on $n$ qubits $C = (U, \mathcal{M})$ is a tuple comprising a Clifford unitary $U$ and a set of $m$ Pauli measurements $\mathcal{M} = \{q_1, ..., q_m\}$, with $q_i \in \mathcal{P}_N$. When referring to a particular application of a Clifford operation, we may write $(-1)^{o(p_i)} = \pm 1$ for the outcome of measuring $p_i$. We will not restrict the measurements to be disjoint. However, we will require them to be independent in most of this work, although we briefly loosen this second requirement in Sec.~\ref{ssec:examples}. We will also require that they commute (otherwise, we would have to order them). We will take $U$ to happen before any of the measurements in $M$. While the circuit itself may comprise many smaller unitaries, for our purposes we always take $U$ to be a single global $n$-qubit unitary (which may have support on any number of qubits). Since errors only occur before or after Clifford operations, the breakdown of $U$ in individual timesteps into smaller unitaries does not affect fault propagation or detector structure -- this should be regarded as a choice of discretization. Such a choice allows us to account for different error models, ranging from fully phenomenological to circuit level noise.


\textbf{Clifford circuit -- } A \textit{Clifford circuit}  $\mathcal{C}$ of length $T$ on $N$ qubits is an ordered set of $T$ Clifford operations, together with $n(T+1)$ spacetime locations.  The spacetime locations are labelled with a tuple $(i,t)$, with $i = 1,...,n; t = 0, ..., T$. Each spacetime location is associated with a spacetime qubit, and we denote the associated Paulis by $\eta^t(X_n), \eta^t(Z_n)$. In other words, given a Pauli operator $p \in \mathcal{P}_N$, we denote by $\eta^t(p) := I_n^{\otimes T - t + 1} \otimes p \otimes I_n^{\otimes t}\in \mathcal{P}_{n(T+1)}$, i.e. the same Pauli but placed along the $t$-th slice of spacetime qubits. Throughout this work, we will reserve upper indices for time. To go in the other direction, if $q \in \mathcal{P}_{n(T+1)}$ is a Pauli supported on spacetime locations, we denote by $\pi^t(q)$ the $t$-th slice of Paulis in $q$. Spacetime qubits are where we will place Pauli errors later. As short-hand, we also denote $\zeta^t(p) := \pi^t ( \eta^t(p))$ the in-place projection of a spacetime Pauli onto a particular time-slice.

The Clifford operations in $\mathcal{C}$ are regarded to occur at half-integer time-steps and as such are labelled $C^{t+0.5}$, or interchangably $C^{t>}$ for short, and $C^{t+0.5}$ is considered to occur between time steps $t, t+1$. We similarly notate the constituent elements, $C^{t>} = (U^{t>}, \mathcal{M}^{t>})$, and $\mathcal{M}^{t>} = \{q_1^{t>}, ..., q_{m^{t>}}^{t>}\}$. For each measurement $q_i^{t>}$, we introduce an ancillary spacetime location, labelled $(i, t+0.5)$, see Fig.~\ref{fig:circuit_convention}d. Although we will often write Paulis associated with this spacetime location,  this should be regarded as a purely classical degree of freedom. This means that we will only allow one kind of error, namely $Z$  errors, on this location. This is where we will place measurement errors later. For a Pauli $p$ supported across all spacetime locations, we denote by $\Pi^-(p)$ the support of $p$ on half-integer locations only, and we denote by $\Pi^+(p)$ the support on integer locations only. We denote the composition of unitaries in the circuit via,
\begin{equation}\label{eq:unitary_composition}
    U^{t \rightarrow t'} := U^{t' - 0.5} ... U^{t + 1.5} U^{t + 0.5}.
\end{equation}

We will distinguish between \textit{initialization} and \textit{input}. For us, initialization refers to something that happens within the circuit $\mathcal{C}$, hence for instance, inserting a complete set of measurements in the $0.5$-th time step may be regarded as initializing the circuit in some state. Conversely, specifying an input requires that some operations have already been performed prior to the circuit. We will avoid that complication for now, and will treat it only in Sec.~\ref{ssec:boundaries} when dealing with boundary conditions.

Henceforth, whenever we say operation or circuit, we mean the Clifford kind, unless stated otherwise.

\begin{figure}[ht!]
\centering 
\includegraphics[width=0.6\linewidth]{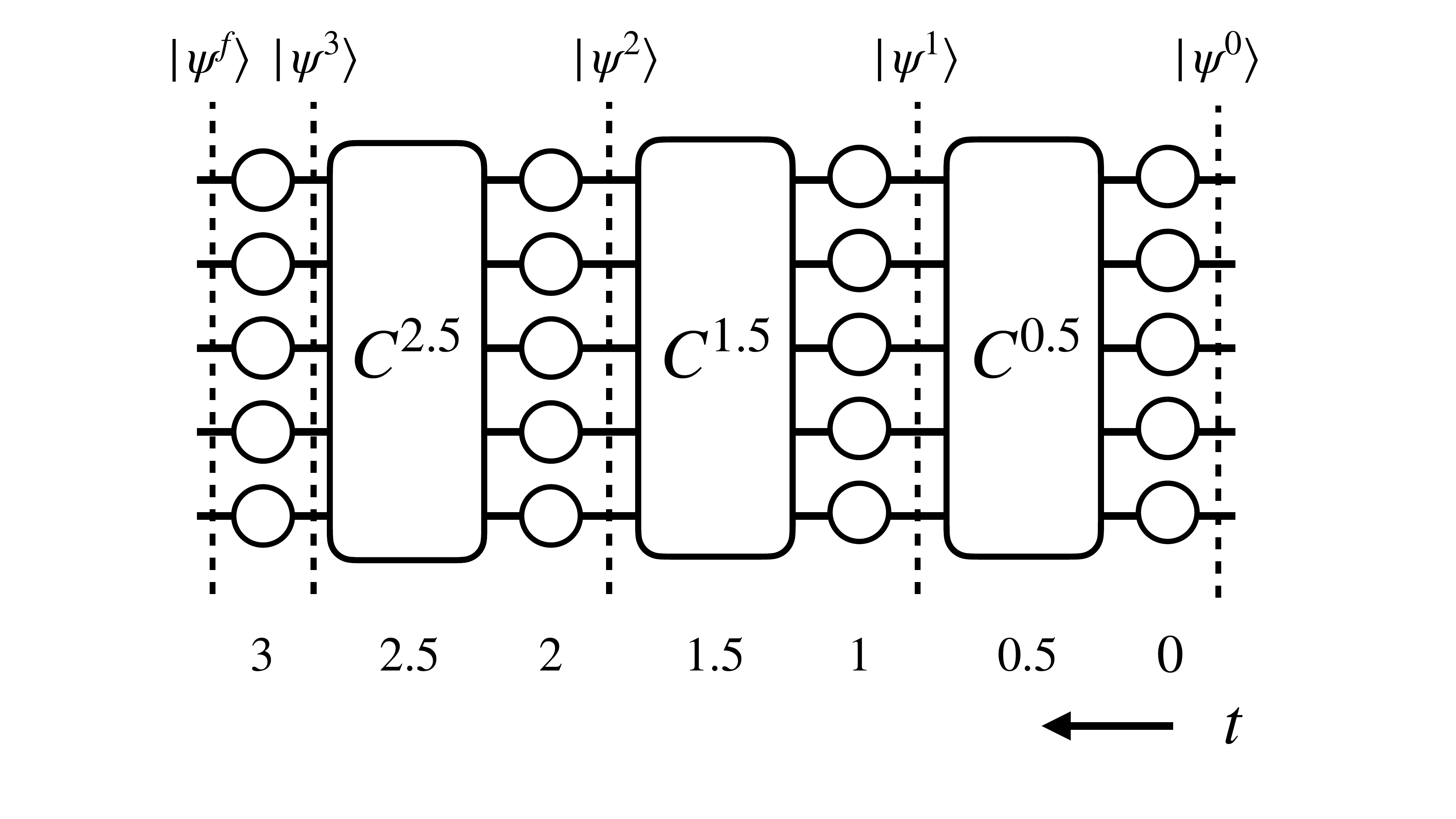} \caption{Schematic depicting our conventions for where we evaluate states and their ISGs. The state $|\psi^t \rangle$ describes the state at just before the $t$-th layer of spacetime locations, and ${\rm ISG}(t)$ describes the ISG of $|\psi^t \rangle$. The final output state of the circuit is labelled $| \psi^f \rangle$.}\label{fig:state_convention} 
\end{figure}
\textbf{States and error identification\footnote{While this is the standard way to identify errors in the literature, in App.~\ref{app_sec:spackle_domain_wall} we will explain an alternate error identification scheme that provides a physically useful picture for the gauged SSC.} -- } Finally, we establish notation for the states associated with the circuits, in order to be clear about how the circuit and errors act on them. Consider a slice of the circuit contains spacetime locations $(i, t)$, and a Clifford operation $C^{t>} = (U^{t>}, M^{t>})$. Say $M^{t>}$ contains $m$ measurements $q_j$. Then, we also have spacetime locations $(j, t+0.5)$ associated with each measurement. We say this slice of circuit acts on an input state $| \psi^t \rangle$, and produces two pieces of data: a measurement record, which we may denote $o(q_1), ..., o(q_m)$, and an output state $|\psi^{t+1} \rangle$. In a run of the circuit, we may encounter a \textit{data error} $p \in \mathcal{P}_n$, this is identified with the operator $\eta^{t}(p)$ on the integer-time spacetime locations, as well as \textit{measurement errors} $Z[\mathbf{v}] = \otimes_{j=1}^m Z_j^{v_j}$, where $\mathbf{v}$ is a vector with $m$ entries, where the $j$-th entry indicates a readout error in the $j$-th measurement, identified with $\eta^{t+0.5}(Z_j)$\footnote{Choosing to identify measurement errors with $Z$ or $X$ is an arbitrary choice. Choosing $Z$ is consistent with preparing an ancilla in the $| + \rangle$ state and using a controlled-Pauli to carry out the measurement. In that case, an ancilla $Z$ error before or after the controlled-Pauli has the same effect on the circuit.}. In that case, the output state is 
\begin{equation}\label{eq:error_identification_I}
    | \psi^{t+1} \rangle \propto \left( \prod_{j=1}^m ( I + (-1)^{o(q_j) + v_j} q_j ) \right) U^{t>} p | \psi^t \rangle,
\end{equation}
up to normalization. We take the outcomes to be $0, 1$, which are randomly distributed according to 
\begin{equation}\label{eq:error_identification_outcomes_I}
    {\rm Prob}(o(q_j) = +1)= \frac{1}{2} + (-1)^{\mathbf{v}_j} \frac{1}{2} \langle \psi^t | p (U^{t>})^{\dag}q_j U^{t>} p | \psi^t \rangle,
\end{equation}
where we see that $\mathbf{v}_j$ indeed acts as a readout error by flipping the sign of the obtained outcome from what one might expect. Eq.~\ref{eq:error_identification_I} allows us to identify \textit{spacetime errors} occurring on spacetime locations, with \textit{circuit errors}, which occur during a run on the circuit.

This convention is depicted in Fig.~\ref{fig:state_convention}. In this manner every set of spacetime locations and Clifford operations in the circuit is associated with an input and output state, $| \psi^0 \rangle, ..., | \psi^T \rangle$. The final layer of the circuit we take to be a final set of spacetime locations. This takes as an input $| \psi^T \rangle$ and outputs $| \psi^f \rangle = p | \psi^T \rangle$, where $p$ is some Pauli error occurring at the end of the circuit. In the absence of errors, we may denote the uncorrupted states $| \psi^t_{(0)} \rangle$ and outcomes $o^{(0)}$. Note that the uncorrupted data is still in general random. 

\subsection{Elementary circuit operators}\label{sec:ECOs}

Next, we will associate a set of elementary circuit operators, $G_{\mathcal{C}}$, associated with the circuit. We emphasize that at this stage we want to think of $G_{\mathcal{C}}$ as a set, not a group -- gauging the spacetime code is a basis-dependent endeavor. It is sufficient to consider two adjacent time-steps $t, t+1$, and the measurements at $t+0.5$. Recall we use $\eta^{t}(p)$ to place a Pauli $p$ on the $t$-th time-slice. We include three sets of elementary circuit operators in $G_{\mathcal{C}}$. One big caveat is that this only fully specifies elementary circuit operators in the `bulk', i.e. away from the $0$ and $T$-th time-steps. Consistent treatment of boundaries will take some care later.

\begin{figure}[ht!]
\centering 
\includegraphics[width=0.9\linewidth]{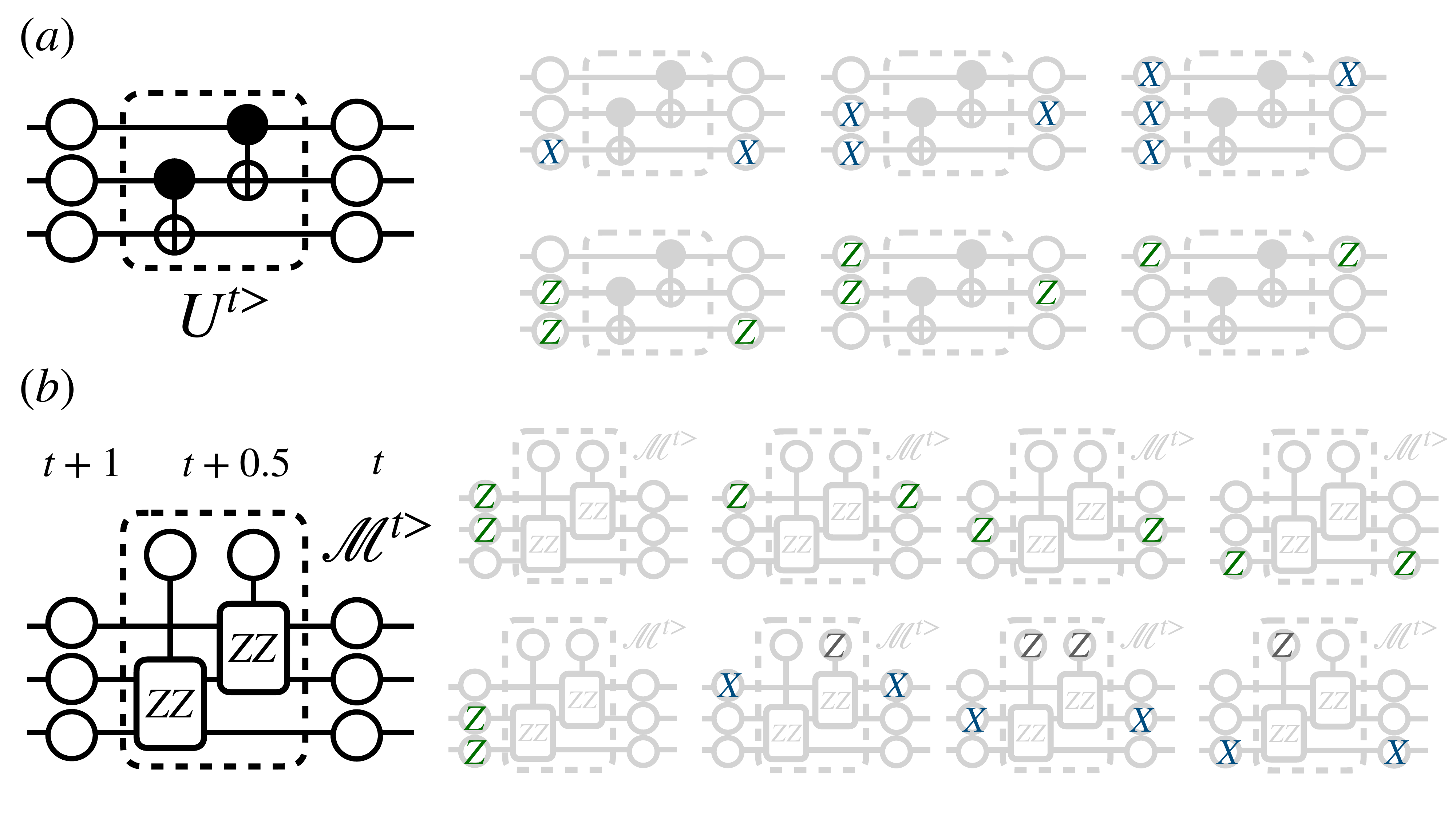} \caption{Schematic illustrating examples of elementary circuit operators associated with circuit elements. For integer time-steps, we will use blue for $X$-type operators and green for $Z$-type operators. Although the operators on half-integer time-steps are formally written as Paulis, they are in principle associated with purely classical variables. To distinguish them, we use gray. \textbf{(a)} Propagators associated with the example depicted in Fig.~\ref{fig:circuit_convention}c. Since the Clifford operation is purely unitary, there are no measurement slices or half-integer spacetime locations. \textbf{(b)} Propagators and measurement slices associated with the example depicted in Fig.~\ref{fig:circuit_convention}d. We omit the measurement dephasers, as they are unimportant at this level. }\label{fig:gauge_op_examples} 
\end{figure}

\begin{enumerate}
    \item \textbf{Measurement slices -- } For each measurement, $q_i^{t>} \in M^{t>}$, we put the operator 
    \begin{equation}\label{eq:meas_slice}
        (g_{\rm ms})_i^{t>} := \eta^{t+1}(q_i^{t>})
    \end{equation} 
    into $G_{\mathcal{C}}$.

    \item \textbf{Propagators -- } Consider the $n$-th qubit. Let $A \subseteq \{1, ..., m^{t>}\}$ denote the set of indices such that $[U^{t>} X_n (U^{t>})^{\dag}, q_i^{t>}] \neq 0$ for all $i \in A$. Then, we put 
    \begin{equation}\label{eq:X_prop}
        (g_{{\rm prop}}^X)_n^{t>} := \eta^{t+1}(U^{t>} X_n (U^{t>})^{\dag}) \eta^{t+0.5} \left( \prod_{i \in A} Z_i \right) \eta^{t}(X_n)
    \end{equation}
    into $G_{\rm st}$. Similarly, let $B \subseteq \{1, ..., m^{t>}\}$ denote the set of indices such that $[U^{t>} Z_n (U^{t>})^{\dag}, q_j^{t>}] \neq 0$ for all $j \in B$. Then, we put 
    \begin{equation}\label{eq:Z_prop}
        (g_{{\rm prop}}^Z)_n^{t>} := \eta^{t+1}(U^{t>} Z_n (U^{t>})^{\dag}) \eta^{t+0.5} \left( \prod_{j \in B} Z_j \right) \eta^{t}(Z_n).
    \end{equation}
    We refer to the operators in Eq.~\eqref{eq:X_prop}, \eqref{eq:Z_prop} as \textit{elementary $X/Z$ propagators}.

    \item \textbf{Measurement dephasers -- } This last set of elementary circuit operators exist primarily to make formal contact with the spacetime code in Sec.~\ref{sec:introducing_SSC}.  For each ancilla qubit in the half-integer time steps, we introduce 
    \begin{equation}
        (g_r)^{t>}_j := \eta^{t+0.5}( X_j).
    \end{equation}
    We will not make much use of these.
\end{enumerate}

It will also be useful to consider a larger set of \textit{propagators}, which are operators that can be generated by elementary propagators. One can regard the propagators as a set of linear maps, indexed by half-integer time-steps, on $\mathbb{Z}_2^{2n}$, writing,
\begin{equation}
    g_{\rm prop}^{t>}(X_n/Z_n) \mapsto (g_{\rm prop}^{X/Z})_n^{t>},
\end{equation}
and extending linearly to arbitrary phaseless Paulis, which will ease notation later. Note that $g_{\rm prop}^{t>}(p)$ for any $p \in \mathcal{P}_n$ has a unique decomposition into propagators (up to a possible minus sign). Later, we will frequently talk about the propagator of a general Pauli, and refer to the set of elementary propagators that generate it.

We illustrate with some examples of elementary circuit operators in Fig.~\ref{fig:gauge_op_examples}. An alternative graphical representation will come in useful later. This is illustrated in Fig.~\ref{fig:check_representation_examples}. We associate each gauge operator with a `check', which is a square vertex. Each spacetime location on an integer time-slice is associated with two vertices, one blue and one green, associated with $X$ and $Z$ operators respectively, whereas the spacetime locations on a half-integer time-slice are associated with a single gray vertex. This distinguishes them as classical degrees of freedom. We connect checks to spacetime locations exactly according to the support of the relevant elementary circuit operators. We will refer to the ECOs generated by operations in $t+0.5$ as the ECOs associated with time $t+0.5$.

\begin{figure}[ht!]
\centering 
\includegraphics[width=0.95\linewidth]{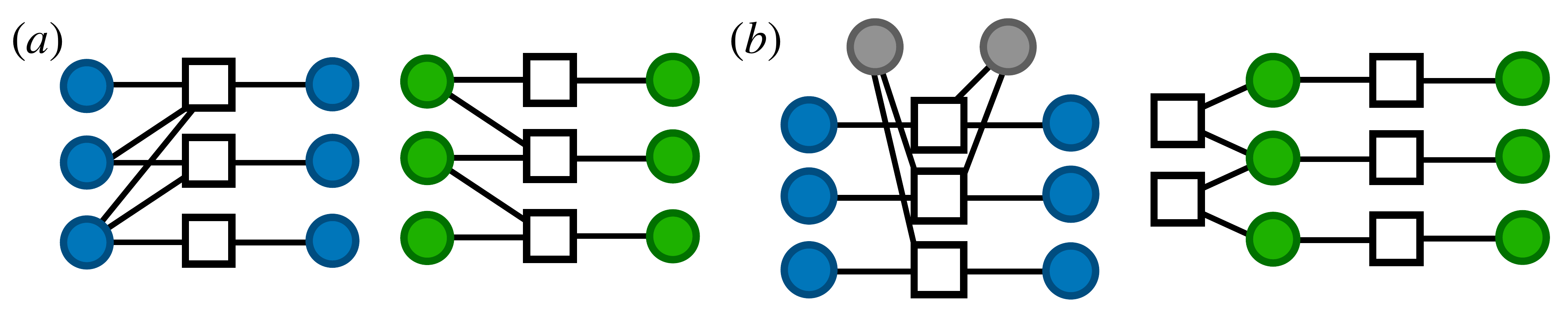} \caption{Schematic illustrating the graphical representation of the elementary circuit operators depicted in Fig.~\ref{fig:gauge_op_examples}. Spacetime locations are associated with circular vertices, and elementary circuit operators are associated with square vertices. Each spacetime location in an integer-valued time-slice is associated with two vertices, one blue and one green, associated with $X$ and $Z$ operators respectively. Each spacetime location in a half-integer time-slice is associated with a single gray vertex. We connect the square vertices to spacetime locations based on their support. We do not distinguish types of square vertices. Each row in the above should be considered part of the same graph. While not depicted here, elementary circuit operators may connect green and blue vertices, eg. when we have a $cZ$ gate, or an $XZ$ measurement. }\label{fig:check_representation_examples} 
\end{figure}

We record an obvious but extremely crucial lemma that we will later refer to again and again. 

\begin{lem}[Emergence of centralizer]\label{lem:centralizer_from_propagator}
    Let $p \in \mathcal{P}_n$ be a Pauli operator. Then, $\pi^{t>}(g^{t>}_{\rm prop}(p)) = I$ if and only if $U^{t>} p (U^{t>})^{\dag}$ commutes with all measurements in $M^{t>}$.  
\end{lem}

This makes the relationship to the elementary propagation and measurement operators of Ref.~\cite{Pesah_arXiv_2025_FT_transformations} clearer (see Sec.~\ref{sec:introducing_SSC}). In light of this lemma, we may think about the half-integer spacetime locations as a tool to help us keep track of which Paulis are able to propagate through the measurements in each layer. 

At this point, we could choose (without motivation) to contextualize the ECOs in a chain complex (which will continue to grow throughout the remainder of this work),
\begin{equation}\label{eq:gauge_complex_1}
    V_2 \xrightarrow{G_{\mathcal{C}}} V_1,
\end{equation} 
where the rows of $G_{\mathcal{C}}$ correspond to ECOs, written in the symplectic representation \cite{Aaronson_2004_PRA_stab_sim}\footnote{For unfamiliar readers, we provide some exposition to chain complexes and the symplectic representation of Paulis in App.~\ref{app_sec:some_useful_concepts}}.

\subsection{ISGs}\label{ssec:ISGs}

The concept of an Instantaneous Stabilizer Group (ISG) was introduced in Ref.~\cite{Hastings_Quantum_2021_floquet_codes} to describe Floquet codes. This will be a useful characterization of a circuit for our purposes. Thinking about circuits as consecutive sets of ISGs is a useful way to anchor our discussion later on. In this section, we prove several statements relating ISGs to elementary circuit operators, and in particular the group $\langle G_{\mathcal{C}} \rangle$ they generate. 

This is used later to define the subsystem spacetime code, but we do this now as the following statements will be convenient for proving things. One immediate piece of information that $G_{\mathcal{C}}$ gives us is that it fully characterizes the ISGs of a circuit, and hence the noiseless action of the circuit. The notion of an ISG is not heavily utilized in the spacetime code literature. Nevertheless, these ideas are largely contained in the exposition of Ref.~\cite{Blackwell_arXiv_2025_distance_floquet}; however, we must adapt them to the particular basis we use.

\begin{definition}[ISG]\label{def:ISG}
    Given a circuit, the ISG at time $t$, denoted ${\rm ISG}(t)$, is the set of stabilizers that the input state will have after the noiseless implementation of all the operations from $0.5, 1.5, ..., t-0.5$. In other words, ${\rm ISG(t)}$ stabilizes $| \psi^t_{(0)} \rangle$. 
\end{definition}

Recall that in the absence of noise $| \psi^f_{(0)} \rangle = | \psi^T_{(0)} \rangle$, so there is no need to specify ${\rm ISG}(f)$. The signs of the operators ${\rm ISG}(t)$ may depend on random measurement outcomes (in a deterministic way). However, the phaseless ISG, denoted $\overline{\rm ISG}(t)$ only depends on the circuit.

It will be useful to describe how the ISG is updated, given how we have chosen to describe a circuit. This follows directly from stabilizer update rules \cite{NC_2019}.

\begin{prop}[ISG update rules]\label{prop:ISG_update}
    Let $\overline{\rm ISG}(t)$ be the ISG at time $t$, and $U^{t>}, M^{t>}$ be the circuit elements at time $t+0.5$, and let $s_1, ..., s_k$ be a generating set of Paulis for $\overline{\rm ISG}(t)$. To obtain $\overline{\rm ISG}(t+1)$, we do the following:
    \begin{enumerate}
        \item Replace each $s_i$ by $s_i' = U^{t>} s_i (U^{t>})^{\dag}$. Initialize the set $S = \{s_1', ..., s_k'\}$.
        \item For each $m_j^{t>} \in M^{t>}$, do the following: 
        \begin{enumerate}
            \item If $m_j^{t>}$ commutes with every $s_i'$: 
            \begin{enumerate}
                \item If $m_j^{t>} \notin S$, append $m_j^{t>}$ to $S$.
                \item If $m_j^{t>} \in S$, do nothing. Note $o(m_j^{t>})$ will be deterministic in this case.
            \end{enumerate}
            \item Suppose $m_j^{t>}$ anticommutes with some set of $K = \{s_i'\}$. Pick one $s_{i*}' \in K$, and replace every other $s_i' \in K$ in $S$ with $s_i' s_{i*}'$. Finally, replace $s_{i*}'$ in $S$ with $m_j^{t>}$. 
        \end{enumerate}
    \end{enumerate}
\end{prop}

To see how $G_{\mathcal{C}}$ yields this information, we note that the content of Prop.~\ref{prop:ISG_update} can be easily restated in terms of propagators.

\begin{lem}[ISG update rules from propagators]\label{lem:ISG_update_prop}
    Let $\overline{\rm ISG}(t)$ be the ISG at time $t$, and $U^{t>}, M^{t>}$ be the circuit elements at time $t+0.5$, and let $s_1, ..., s_k$ be a generating set of Paulis for $\overline{\rm ISG}(t)$. Then a generating set $\eta^{t+1}(v_1), ..., \eta^{t+1}(v_q)$ of $\overline{\rm ISG}(t+1)$ can be written as products of $\eta^t(s_1), ..., \eta^t(s_k)$ together with ECOs associated with time $t+0.5$.
\end{lem}

\begin{cor}[ECOs generate ISGs]\label{cor:gauge_op_ISG}
     Consider a circuit $\mathcal{C}$ with an input state described by ${\rm ISG}(0)$. Then $\eta^t({\rm ISG}(t)) \subseteq \langle G_{\mathcal{C}} \cup \eta^0({\rm ISG}(0)) \rangle$. Furthermore, elements of ${\rm ISG}(t)$ can be generated `causally', i.e. without using any elementary circuit operators associated with any time $t' -0.5 > t -0.5$.
\end{cor}

Hence, all ISGs are obtainable from $G_{\mathcal{C}}$ as well as a specification of initial conditions. From the above, we see that each element of each ISG shows up in in $\langle G_{\mathcal{C}} \rangle$ as an operator on spacetime locations supported on single integer-valued time-slices. We note that the converse is not true -- there are operators supported on single integer-valued time-slices that are not in the ISG (although in that case, they would be, up to a unitary, in the ISG of the next time-slice). We only need a limited version of this statement. 

\begin{lem}\label{lem:alt_ISG_prop}
    Consider a circuit $\mathcal{C}$, and suppose ${\rm ISG}(0) = \langle I \rangle$. Let $p \in \langle G_{\mathcal{C}} \rangle$ be an operator that (i) is supported only on time-slice $t_p$, (ii) can be generated by elementary circuit operators associated with time-steps $t < t_p$ only. Then, $\pi^t(p) \in \overline{\rm ISG}(t)$. 
\end{lem}

We end with a simple corollary.

\begin{cor}
    Consider a circuit $\mathcal{C}$. Let $p \in \langle G_{\mathcal{C}} \cup {\rm ISG}(0)\rangle $ be an operator supported only on the final time-slice $T$. Then, $p \in \overline{\rm ISG}(T)$.  
\end{cor}

\subsection{Elements of fault tolerance}\label{ssec:elts_of_FT}

The goal of this section is to examine two crucial elements of fault tolerance. The first element is spacetime faults, which form the set of errors we want to protect ourselves against. Of course not all spacetime faults are bad, and not all spacetime faults are equivalent. One major goal of this section is to understand the equivalence between spacetime faults. The definition, given in Def.~\ref{def:fault_equiv}, is straightforward, and reduces to a similar definition in Ref.~\cite{Blackwell_arXiv_2025_distance_floquet} when pure measurement errors are removed from the picture and infinitely long circuits are considered. This will set the stage for when we explore equivalence in the gauged version of the spacetime code. 


The other element is detectors \cite{Gidney_2021_quantum_STIM, McEwen_quantum_2023_time_dynamics, Derks_Quantum_2025_designing_detectors}, which form a cornerstone of modern fault tolerance. The main idea is that in fault tolerance, we must treat not just the qubits but also the measurements as unreliable. Decoding algorithms are then not built directly on measurement outcomes, but rather a new set of variables called detectors, which are invariant relations between measurement outcomes. In the presence of measurement error, decoders are built on detectors rather than measurements \cite{Derks_Quantum_2025_designing_detectors}. Pre-empting our formulation of the gauge theory in Sec.~\ref{sec:gauging_SSC}, we point out that in fault tolerance, detectors, not measurements, are the actual observables. When moving to a gauge theory, we will identify detectors with Wilson loops, which are the actual gauge-invariant observables of a gauge theory.

To illustrate the concepts we want to capture, we start with an example. A simple example of a detector is depicted in Fig.~\ref{fig:detector_example}. Measurements may not yield syndrome data directly, as the distribution of outcomes of any single measurement may depend on the input state to the circuit. Instead, the circuits used in fault-tolerance work with \textit{detectors}, which are sets of measurements whose product yield a deterministic outcome. In this circuit we start with a $ZZ$ measurement across qubits $1, 2$, then perform a cNOT gate from qubit $2$ to $1$, and finally measure $Z$ on qubit $1$. The two measurements, which we will refer to as $q^{t+1>}, q^{t>}$ yield the same outcome in the noiseless case, in other words, the product of their outcomes, 
\begin{equation}
   (-1)^{o(q^{t+1>}) + o(q^{t>})} = 1,
\end{equation} 
hence they form a detector. We define a detector more generally in Def.~\ref{def:detector}.  Figs.~\ref{fig:detector_example}a-c show some ways in which this detector may be violated by spacetime \textit{faults}, which comprise Pauli operators placed at various spacetime locations. In Fig.~\ref{fig:detector_example}a, we have a data error $\eta^{t+1.5}(X_2)$, which after the cNOT gate will flip the outcome of the second measurement. In Fig.~\ref{fig:detector_example}b, we have a measurement error, which flips the outcome of the first measurement without any effect on the circuit, which is identified with $\eta^{t+0.5}(Z_a)$. In Fig.~\ref{fig:detector_example}c, we have a pair of data errors that would cause one to readout $q^{t>}$ wrongly. 

All three errors \textit{trigger}, or \textit{violate}, the same detector, i.e. they cause $(-1)^{o(q^{t+1>}) + o(q^{t>})} = -1$, and so cannot be distinguished by the classical measurement outcomes of the circuit. However, the error of Fig.~\ref{fig:detector_example}a should not be regarded as equivalent to the errors of Fig.~\ref{fig:detector_example}b, c, as they lead to different Pauli errors on the output state of the circuit. Conversely, the errors Fig.~\ref{fig:detector_example}b, c are truly equivalent, in the sense that they flip exactly the same measurement outcomes and they lead to exactly the same output state. The distinction is that errors that trigger different sets of detectors may be distinguished by the measurement outcomes of the present circuit alone, whereas errors that are inequivalent could in some cases still be distinguished by further processing on the output state.

The main goal of this section is formalize these notions, and show how they emerge from $G_{\mathcal{C}}$. 

\begin{figure}[ht!]
\centering 
\includegraphics[width=0.95\linewidth]{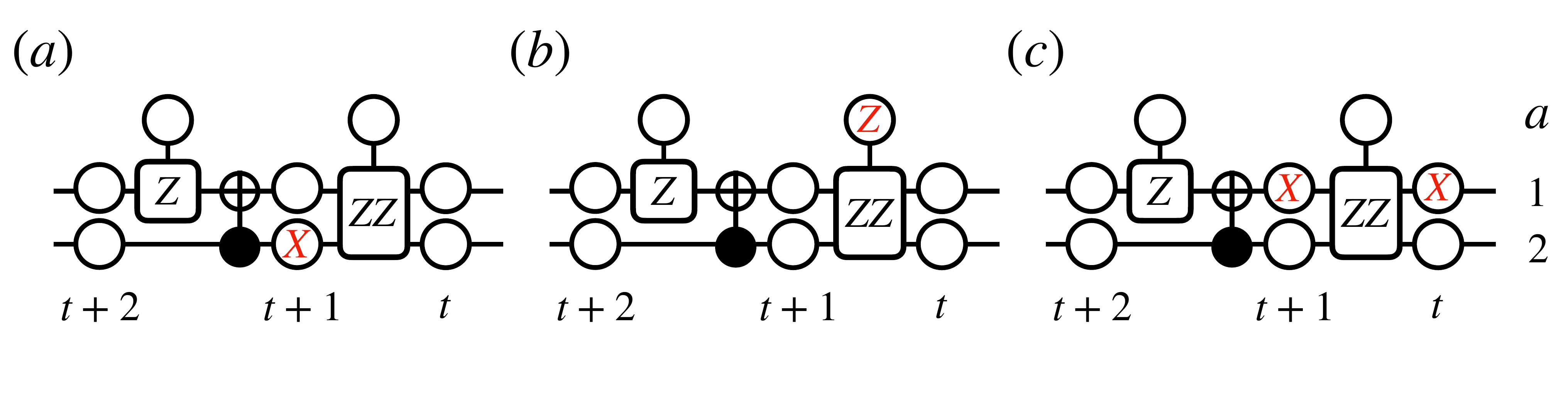} \caption{A simple example of a detector formed by a $ZZ$ measurement followed by a cNOT gate followed by a single qubit $Z$ measurement, and several faults which violate the detector. \textbf{(a)} An $X$ error between the two measurements. \textbf{(b)} A readout error in the first measurement. \textbf{(c)} A pair of $X$ errors just before and just after the first measurement. All three examples trigger the detector. However, while the errors in (b) and (c) are equivalent with each other, they are not equivalent to the error in (a).}\label{fig:detector_example} 
\end{figure}

\subsubsection{Fault equivalence and effect}

We start by characterizing the effects of faults in a circuit. A \textit{fault} $F$ is a tuple $F = (e, f)$, where $e \in \overline{\mathcal{P}}_{n(T+1)}$ is a Pauli supported on spacetime locations, and $f \subseteq \cup_{t=0}^{T-1} M^{t>}$ is a set of measurements which are flipped. This second object represents classical measurement errors, which may occur at readout or otherwise, which have no effect on the state of the circuit. They may also be represented by Pauli $Z$'s in the half-integer locations. Put together, we may alternatively identify a fault as an element $F \in V_1$, see Eq.~\ref{eq:gauge_complex_1}.  In other words, we can regard $F$ as a Pauli supported on both integer and half-integer spacetime locations, as long as we restrict the support on half-integer locations to be $Z$-type. 

We will concern ourselves with the effect of a circuit. We follow the lead of Ref.~\cite{Delfosse_arXiv_2023_spacetime_code_clifford} in formalizing this notion, with two small differences: (1) We include spacetime locations for measurement errors in half-integer time-slices, (2) In their work, they define the effect only via the Pauli error on the output, obtained by propagating the fault to the end of the circuit. However, refining this notion a little will allow us to prove a powerful correspondence to ECOs. 

The key idea is that executing a circuit produces two kinds of data: (1) an outcome bit-string $o \in \mathbb{Z}_2^{\sum_{t} m^{t>}}$, with components corresponding to the outcomes of measurements in the circuit, (2) an output state $\rho_o$ of the qubits after running the circuit. A fault modifies each of these outputs in a systematic way.

Regarding (1), the outcome bit-string of a circuit is generally non-deterministic even in the noiseless case. We may describe it instead by a probability distribution $P_{\mathcal{C}, \rho}(o)$ as the distribution of the outcome bit-string $o$, for a circuit $\mathcal{C}$, given the input state $\rho$.

\begin{definition}[Bare effect]\label{def:bare_effect}
    Consider a circuit $\mathcal{C}$ subject to a fault $f$, and suppose the outcome distribution is modified to be $P_{\mathcal{C}, \rho}^{(f)}(o) = P_{\mathcal{C}, \rho}(o + o_f)$, and the output state is modified as $\rho_f = p_f \rho_o p_f^{\dag}$ for some Pauli $p_f$. We then refer to $(o_f, p_f)$ as a bare effect of the fault. 
\end{definition}

Suppose $\rho_o$ is stabilized by some Pauli $p'$. Then, $(o_f, p_f p')$ will also be a bare effect of $f$ -- in other words, the bare effect is non-unique. However, before dealing with that, we first point out that one can obtain a bare effect using $G_{\mathcal{C}}$.

\begin{lem}[Propagating to effect]\label{lem:prop_to_effect}
    Let $F$ be a fault, and let us denote $F^t$ the support of $F$ on time-slice $t$, so that $F = \eta^T(F^T) \prod_{t=0}^{T-1} \eta^{t+0.5}(F^{t+0.5}) \eta^t(F^t)$. The bare effect of a fault can be obtained as 
    \begin{equation*}
        F^{\rm eff} = F \times \prod_{t = 0}^{T-1} g_{\rm prop}^{t+0.5}( F^t U^{t-0.5}...F^1 U^{0.5} F^0 (U^{0 \rightarrow t})^{\dag} ),
    \end{equation*}
    where we recall the definition of $U^{0 \rightarrow t}$ from Eq.~\eqref{eq:unitary_composition}. By construction, this is a product of $F$ with a set of ECOs. Furthermore, the bare effect can be obtained from this operator: (1) $p_F = \pi^T(F^{\rm eff})$ is the Pauli modifying the output state, and (2) The non-trivial support of $\Pi^-(F^{\rm eff})$ on measurement locations tells us $o_F$, i.e. which outcomes are flipped. 
\end{lem}

The above lemma can be verified either by sequentially propagating the operators in each integer time-slice through to the next one, or by comparison with Prop.~1 of Ref.~\cite{Delfosse_arXiv_2023_spacetime_code_clifford}. The exact form of $F^{\rm eff}$ will not be important, only that it can be obtained via multiplication with ECOs. The main idea is that applying $g^{t>}_{\rm prop}$ to the fault for each $t$ yields a Pauli operator localized on the $T$-th time-slice, which is exactly equal to $p_f$, and leaves behind a bunch of $Z$ operators in half-integer time-slices, which tell us exactly what $o_f$ is.

The mentioned non-uniqueness in Def.~\ref{def:bare_effect} is resolved by taking ISGs into account. 

\begin{definition}[Fault equivalence]\label{def:fault_equiv}
    Let $(o_f, p_f)$, $(o_{f'}, p_{f'})$ be the bare effects of $f, f'$. We say that they are equivalent if $o_f = o_{f'}$ and $p_f =  p_{f'} \cdot s$ (up to a sign) for $s \in {\rm ISG}(T)$. 
\end{definition}

In other words, we say two faults $f, f'$ are equivalent if they flip the same measurements and modify the output state via the same Pauli modulo ${\rm ISG}(T)$. Finally, to relate this to the elementary circuit operators, we have the following proposition: 

\begin{prop}[Fault equivalence]\label{prop:fault_equiv}
    Two faults $F, F'$ are equivalent if and only if they can be obtained from each other via a product of elementary circuit operators.
\end{prop}

The proof is straightforward from application of Lem.~\ref{lem:prop_to_effect}, and is given in App.~\ref{app_sec:deferred_proofs}. We provide a couple comments on the relation of Prop.~\ref{prop:fault_equiv} to prior work. In contrast to Ref.~\cite{Blackwell_arXiv_2025_distance_floquet}, our notion of equivalence does not have to posit a large number of future times. The reason for this is that our notion of equivalence is weaker than theirs, and is adapted for finite depth circuits. To that end, we note that Prop.~\ref{prop:fault_equiv} breaks if we do not include a final layer of spacetime locations at the end of the circuit which are not operated on. Our conceptualization of measurement errors lying on half-integer locations also reveals that we must treat input stabilizers on a different footing than measurements in the circuit. This is why, in contrast to Ref.~\cite{Pesah_arXiv_2025_FT_transformations}, we choose, at this stage, to avoid including elements of $\eta^0({\rm ISG}(0))$. The key reason for this is that these input stabilizers, when regarded as input data, can never suffer from a fault (insofar as faults arise from errors internal to the circuit), and as such play a different role from all other measurement slices in the circuit -- one more akin to boundary conditions, as we will see in Sec.~\ref{ssec:boundaries}.

\subsubsection{Detectors and errors}

For this discussion, we will restrict ourselves to circuits where ${\rm ISG}(0) = I$, i.e. no restrictions on the input state. The reason for this is we want to leave the issue of input states to later on when we deal with boundaries.

\begin{definition}[Detector]\label{def:detector}
    Given a circuit $\mathcal{C}$, a detector is a set of measurements such that the product of measurement outcomes has a definite parity in any given noiseless run of the circuit.
\end{definition}

To characterize detectors in the usual treatment of the spacetime code, we must define the rather adorably named \textit{spackle}. 


\begin{definition}[Spackle \cite{Pesah_arXiv_2025_FT_transformations}]\label{def:spackle}
    The spackle of a measurement $q_j^{t+0.5}$ in time-slice $t -0.5$ is defined,
    \begin{equation*}
        {\rm spackle}(q_j^{t+0.5}) = \eta^{t+0.5}(X_j) \prod_{s = t+1}^T \eta^s \left(U^{t \rightarrow s} q_j (U^{t \rightarrow s})^{\dag} \right),
    \end{equation*}
    and the spackle of multiple measurements $m_1, m_2, ..., m_k$ is simply their product,
    \begin{equation*}
        {\rm spackle}(m_1, ..., m_k) = {\rm spackle}(m_1) \times ... \times {\rm spackle}(m_k).
    \end{equation*}
\end{definition}

\begin{prop}[Detectors are stabilizers]\label{prop:detectors_are_stabilizers}
    The spackle of a detector is in the center of $\mathcal{G}_{\mathcal{C}}$. 
\end{prop}

We will not make much use of spackles. Nevertheless, Prop.~\ref{prop:detectors_are_stabilizers} is crucial in motivating the spacetime code. A version of Prop.~\ref{prop:detectors_are_stabilizers} is proven in Ref.~\cite{Pesah_arXiv_2025_FT_transformations}, which we adapt to our setting in App.~\ref{app_sec:deferred_proofs}. 

\begin{prop}[Detectors and spacetime errors \cite{Delfosse_arXiv_2023_spacetime_code_clifford}]\label{prop:detectors_and_errors_I}
    A fault that anticommutes with the spackle of a detector flips the total parity of the measurements in the detector. 
\end{prop}

\begin{lem}[Detector and ISG]\label{lem:detector_ISG}
    Let $D$ be a set of measurements forming a detector. Let $t$ be the latest time in this set, i.e. there exists a set of measurements $q_1, ..., q_m \in D$ supported on time $t$, and no measurement in $D$ supported on any $t' > t$. Then the product $q_1 \cdot ... \cdot q_m \in U^{t-0.5}\overline{\rm ISG}(t-1) (U^{t-0.5})^{\dag}$.
\end{lem}

\subsubsection{Circuit distance}\label{sssec:circuit_distance}

Taking into account the notion of fault tolerance together with the observability of detectors leads to the following natural definition of an undetectable error.
\begin{definition}[Undetectable error]\label{def:undetectable_error}
    Given a discretization of a Clifford circuit, an undetectable error is a circuit fault that (1) does not trigger any detector, and (2) is not equivalent to the null fault.
\end{definition}

One is tempted to define the fault distance of a circuit as the lowest weight undetectable error. However, this is tricky -- any circuit that does not begin and end with a fully specified state will have weight-one undetectable errors on the $0$-th or $T$-th time-step. The reason for this is that detectors cannot catch errors that happen before or after the circuit. We will instead use the following notion of an \textit{internal distance}.

\begin{definition}[Internal distance of a circuit]\label{def:fault_distance}
    The internal distance (interchangeably, fault distance) of a circuit is the weight of the lowest weight undetectable error with no support on the $0$-th or $T$-th time-slices. 
\end{definition}

We note that another common definition of a circuit distance is the lowest weight error that leads to a logical error on the output code. This is the definition used for instance in Ref.~\cite{Pesah_arXiv_2025_FT_transformations} and similar in spirit to the definition used in Ref.~\cite{Blackwell_arXiv_2025_distance_floquet}. In this work, we avoid this definition of distance for three reasons: first, we do not want to commit to an output code. Second, our circuits do not include any feedback and correction, and so do not take into account a subsequent decoding step that determines what finally ends up being a logical error or not. Finally, as pointed out in Ref.~\cite{Blackwell_arXiv_2025_distance_floquet}, there is some subtlety in how long it takes a problematic error to become a logical error, which may not be captured in our setting of finite-length quantum circuits.

The notion of a fault distance crucially depends on how one discretizes the Clifford circuit. This is to allow us to take into account circuit-level and phenomenological noise models, as well as models in-between. Any given circuit discretization may fail to capture subtleties of the compilation of the gates wrapped up in each time-step. We further note that the above is not the only way to define the fault distance of a circuit, and may not even be the most reasonable way in many circumstances. For instance, in proofs of fault tolerance, one usually requires that a locally stochastic error channel remains locally stochastic after the circuit \cite{Kubica_natcomm_2022_single_shot_3D} -- this notion cannot be captured by Def.~\ref{def:fault_distance}, as any undetected error is regarded as ruinous, and is perhaps a rather unforgiving notion of a distance.

\subsection{The subsystem spacetime code}\label{sec:introducing_SSC}

In this section, we review the subsystem spacetime code \cite{Pesah_arXiv_2025_FT_transformations, Delfosse_arXiv_2023_spacetime_code_clifford, Bacon_STOC_2015_sparse_codes_circuits}. While the construction of the gauged SSC later may be motivated directly from the results of the previous section, this allows us to make contact with a prevailing framework in the literature. 


The \textit{spacetime subsystem code} (SSC) is the subsystem code obtained by regarding $G_{\mathcal{C}} \cup \eta^0({\rm ISG}(0))$ as the generating gauge group (in the rest of this discussion we assume ${\rm ISG}(0)$ is included, and simply say $G_{\mathcal{C}}$). Recall that the code then has stabilizers $Z[\mathcal{G}_{\mathcal{C}}]$, the center of the $\mathcal{G}_{\mathcal{C}}$, and logical operators $C_{\mathcal{P}_n} [\mathcal{G}_{\mathcal{C}}] - Z[\mathcal{G}_{\mathcal{C}}]$, the centralizer of $\mathcal{G}_{\mathcal{C}}$ in the Pauli group, with stabilizers removed) \cite{Terhal_RMP_2015_QEC_review}. Denoting a generating set of stabilizers by $S_{\rm st}$ (regarded as a matrix over $\mathbb{F}^{2n}$), we can associate the SSC with a \textit{gauge complex} \cite{Pesah_arXiv_2025_FT_transformations}.
\begin{equation}\label{eq:gauge_complex_1p5}
    V_3 \xrightarrow{ \Omega S_{\rm SSC}^T} V_2 \xrightarrow{G_{\mathcal{C}}} V_1,
\end{equation} 
with $\Omega$ the symplectic $2$-form encoding the commutation between Paulis. The gauge complex, in the above form, is not very crucial to our results, but it is a nice framing to keep in mind. The gauge theory we construct later may be viewed as an extension of this complex \cite{kubica_2018_arXiv_ungauging_QEC}.

We note that the standard construction of the SSC does not include the half-integer time steps we use, nor propagators that anti-commute with measurement slices. However, one can arrive at the standard construction from ours by applying the fault-tolerant maps developed in Ref.~\cite{Pesah_arXiv_2025_FT_transformations} or simply tracing out half-integer time-steps from our construction. To a large extent, either construction will yield more or less identical properties when regarded as a subsystem code, modulo measurement errors. 

\subsubsection{The gauge group and detectors}

As a construction, the SSC may be motivated directly by Prop.~\ref{prop:detectors_are_stabilizers}. This is the observation that the detectors may be identified with certain elements in the center of $\mathcal{G}_{\mathcal{C}} = \langle G_{\mathcal{C}} \rangle$. In terms of a subsystem code, these are simply stabilizers of the code. In fact this identification is an injective homomorphism, so the identified stabilizers do act and combine the way the detectors do\footnote{For concreteness, see App.~\ref{app_sec:detectors_are_redundancies} for instance, which establishes the way in which detectors form a group.}. Furthermore, due to Prop.~\ref{prop:detectors_and_errors_I}, these also play well with the error identification described in Sec.~\ref{sec:clifford_circuits_and_gauge_generators}. This leads to the following fact:

\begin{fact}\label{fact:detectability_SSC}
    Any error that is detectable by the SSC is also detectable by the circuit. \cite{Gottesman_arXiv_2022_opportunities_FT}
\end{fact}

One may hope that the distance of the SSC captures the fault distance of the circuit (Def.~\ref{def:fault_distance}). Due to the Fact.~\ref{fact:detectability_SSC}, the SSC certainly captures to a large degree the detectability of circuit errors. However, the set of stabilizers of the SSC is potentially larger (but never smaller) than the set of detectors in the circuit. Recall that the distance of a subsystem code is the lowest-weight dressed (that is, up to multiplication by operators in the gauge group) logical operator, and logical operators are operators that have no syndrome, i.e. are undetectable. Since there exists the possibility of errors with non-trivial syndrome in the SSC that cannot be detected by the circuit, leading us to the next fact: 
\begin{fact}\label{fact:distance_Bound_SSC}
    When the SSC has a non-trivial code-space, the distance of the SSC upper bounds the fault distance of the circuit. \cite{Gottesman_arXiv_2022_opportunities_FT}
\end{fact}

The above fact comes with a caveat of a non-trivial code space, which we will deal with in the next section. The difference between the distance of the SSC and the fault distance of the circuit can be attributed to existence of stabilizers which are not detectors. These are usually regarded as `open' detectors \cite{Pesah_arXiv_2025_FT_transformations}. We note that Fact.~\ref{fact:distance_Bound_SSC} is perhaps inelegant but not a big issue. One advantage is that non-detector stabilizers leave open the possibility of composition with open detectors in other circuits -- in a sense, the distance of the SSC describes the potential of the circuit. Note that this does not happen in the construction of Ref.~\cite{Delfosse_arXiv_2023_spacetime_code_clifford}, which constructs a stabilizer code out of the detectors only. However, the stabilizer spacetime code of Ref.~\cite{Delfosse_arXiv_2023_spacetime_code_clifford} is much less forgiving than the SSC when it comes to the detectability of errors. One can think about this as converting many of the gauge qubits of the SSC into logical qubits \cite{Fu_unpub_subsystem_spacetime_code}, as such the stabilizer spacetime code may have a distance smaller than the circuit. 

There is a large plethora of circuits with guarantees in terms of their fault distance. In terms of phases of matter, we then expect the SSC to describe a phase of matter that is in a sense at least as stable as we expect from the circuit.

\subsubsection{Logical operators of the SSC}

How important are logical qubits to the SSC? Whether or not we even need logical qubits in such a framework is questionable. While the notion of a code distance depends on having logical qubits, we can always appropriately redefine it. However, one reason one might care is that we might really like to think about the dynamical process of a circuit as a static code; this can aid us in decoding for instance \cite{Delfosse_arXiv_2023_spacetime_code_clifford}. Another reason to care is that when trying to think about a code as associated with a phase of matter, one minimal expectation is that the basic form of the theory should have a degenerate ground state space \cite{Zeng_2019_QIMQM}.

When does the SSC have a non-trivial code-space? A necessary but insufficient condition is that $T$ is even (i.e there are an odd number of integer time-slices) \cite{Bacon_STOC_2015_sparse_codes_circuits, Pesah_arXiv_2025_FT_transformations}. In this case, the logicals are usually products of logicals of the ISGs across all time-steps. For instance, in a repeated syndrome extraction circuit on a standard stabilizer subspace code, these logicals are of the form $\prod_{t=0}^T \eta^t(L)$, where $L$ is a logical operator of the input code. In a Floquet code \cite{Hastings_Quantum_2021_floquet_codes}, we might have something of the form $\prod_{t=0}^T \eta^t(L^t)$, where $L^t$ is different for different time-steps \cite{Fu_Quantum_2025_Dynamical_Codes, Blackwell_arXiv_2025_distance_floquet}. 

The fine-tuning between even and odd discretizations is inelegant, but this does not seem to be a problem in practice \cite{Aitchison_arXiv_2026_spacetime_spins}. After all, one can always add an additional time-step by appending identity gates to the circuit. Still, one must pay close attention to the caveat of Fact.~\ref{fact:distance_Bound_SSC}. The reason for this is that, even when fixing the parity of time-steps, there can be practical circuits of interest that defy the existence of logical qubits. One pertinent example is circuits describing memory experiments, which we discuss at length in Sec.~\ref{ssec:boundaries}. 



\subsection{What then, and what next?}

In this long preamble, we have made somewhat of a case for the SSC. We've argued that the elements of fault tolerance primarily require three things to work out: circuit action, error equivalence, and detectors. On the level of the gauge group, the SSC gets error equivalence right, as we showed in Prop.~\ref{prop:fault_equiv}. However, this is a separate issue from whether we treat it as a code or not. The ECOs alone give this property, as evidenced by the treatment of the problem by Ref.~\cite{Blackwell_arXiv_2025_distance_floquet} via the benign error formalism with no recourse whatsoever to a code. It can be the case that, without some fine-tuning, the SSC fails to admit information about the logic of the circuit -- in particular, situations where the SSC has no logical qubits. The SSC captures detectors as stabilizers, but not as tightly as one might desire. None of these are major problems, but still we are driven to soften some of them by developing a complementary view of the subject.

Our next step will be to gauge the SSC (rather, we will gauge the ECOs). In doing so, our goal is to construct a lattice gauge theory which retains the idea of gauge equivalence, but always has a distance (i.e. logical degrees of freedom). As a bonus, we will obtain a set of symmetries which cleanly correspond to detectors.

\section{Gauging the spacetime code}\label{sec:gauging_SSC}

In this section, I will describe a complementary approach to the spacetime code. We will not immediately attempt to generate a quantum code with $G_{\mathcal{C}}$. Instead, we will gauge $G_{\mathcal{C}}$ to obtain a lattice gauge theory that serves as a dual to the SSC\footnote{We note that the gauge theory we construct technically does not live on a lattice but a hypergraph, as the gauge fields can neighbor more than two matter fields.}. We comment that gauging has become something of the Swiss Army knife of QEC, and has found widespread use in code construction \cite{Cuiper_quantum_2025_systematic_construction_via_gauging}, state preparation \cite{Tantivasadakarn_PRX_2024_LRE_meas}, measuring logical operators \cite{Williamson_arXiv_2024_low_overhead_gauging}, and carrying out non-Clifford gates \cite{Davydova_arXiv_2025_universal_FT}, just to name a few examples\footnote{A big part of this work was motivated by simply seeing the word `gauging' everywhere and trying to understand what that meant.}. My goal is much more modest. I will show that this dual theory preserves the elements of fault tolerance outlined in the previous section, but with several additional perks, such as symmetries that correspond one-to-one with detectors (Sec.~\ref{sssec:detectors_redundancies}) and a clear distinction between logical and non-logical degrees of freedom in terms of locality (Sec.~\ref{ssec:boundaries}). In principle, all this information is contained in the SSC (what we are getting is but a dual theory after all). However, the gauge theory formulation will also clarify some connections to other constructions and applications in the literature (Sec.~\ref{sec:applications}). 

We will now provide a high-level discussion of the gauged theory as a guide to the rest of the section. Henceforth, we will refer to the gauged theory as simply the \textit{gauged SSC}. In the following discussion, we will treat $G_{\mathcal{C}}$ as absent of measurement dephasers. We originally introduced those as a means to make the measurement locations classical, but will now simply take that as given\footnote{One could keep the measurement dephasers around and go through the whole procedure. The result is that the measurement dephasers form a disconnected component of the resulting gauge theory and so can be removed from the picture without consequence.}.

The procedure we will carry corresponds to treating $G_{\mathcal{C}}$ as the parity check matrix of a classical code by splitting $X, Z$ and measurement locations, essentially using the symplectic representation over $\mathbb{Z}_2^{2n}$ (see App.~\ref{app_sec:some_useful_concepts}), and then gauging the classical code. This corresponds to the general procedure introduced in Refs.~\cite{kubica_2018_arXiv_ungauging_QEC, Rakovszky_2023_arxiv_physics_ldpc_I}, with only minimal modifications at the temporal boundaries. This is the `minimal coupling' prescription used to introduce dynamical $\mathbb{Z}_2$ gauge fields, which is developed at length in the context of QEC in Ref.~\cite{Rakovszky_2023_arxiv_physics_ldpc_I}. In Sec.~\ref{ssec:construction}, we will explain the steps of the gauging procedure as they apply to our work, but refer the reader to the aforementioned references, as well as Refs.~\cite{Haegeman_PRX_2015_gauging_states, Prem_scipost_2019_gauging_swaps}, for a full review.

In broad strokes this will involve:
\begin{enumerate}
    \item Introducing ancilla variables, or \textit{gauge fields}, for each ECO in $G_{\mathcal{C}}$.
    \item Identifying \textit{redundancies} in $G_{\mathcal{C}}$, and enforce them as symmetries acting on the gauge fields.
\end{enumerate}
Arranging the redundancies in a matrix $R_{\mathcal{C}}$, this may be interpreted as an extension of the chain complex of Eq.~\ref{eq:gauge_complex_1p5},
\begin{equation}\label{eq:gauge_complex_2}
    V_3 \xrightarrow{ \Omega S_{\rm SSC}^T} V_2 \xrightarrow{G_{\mathcal{C}}} V_1 \xrightarrow{R_{\mathcal{C}}} V_0.
\end{equation} 
with $V_1$ identified with the gauge fields. The resultant sub-complex from $V_2$ to $V_0$ may be identified as a $\mathbb{Z}_2$ lattice gauge theory \cite{Kogut_RMP_1979_lattice_gauge_theory}, albeit with gauge fields living on hyper-edges rather than edges.

Gauging allows us to capture the elements of fault tolerance expressed by $G_{\mathcal{C}}$. However, it will give us something a little more compelling than just an equivalent theory on a different set of variables. To understand why, we consider the three important elements that identify this as a gauge theory. 
    \begin{enumerate}
        \item The first element is the gauge fields, which are elements of $V_1$, identified with the ECOs of $G_{\mathcal{C}}$. In a model permitting single-qubit $X, Z$ errors at any location\footnote{In this error model, a $Y$ error is considered weight-two.}, as well as in any measurement outcome, these gauge fields can be identified with circuit faults in a way that is completely local, up to boundary conditions. In other words, a weight-one circuit fault corresponds to a weight-one gauge field\footnote{It is also possible to identify gauge qubits in the SSC, and identify these with circuit errors \cite{Fu_unpub_subsystem_spacetime_code}. However, for small weight errors, the gauge qubits operators will in general always have a higher weight than the corresponding circuit errors.}. This is an unsurprising result, and just takes a little care to establish; stemming fundamentally from the duality between the classical code specified by check matrix $G_{\mathcal{C}}$ and its transpose code specified by $G_{\mathcal{C}}^T$. However, we note that this is not a given -- writing the ECOs in the usual basis used for the spacetime code spoils this strict locality property \cite{Aitchison_arXiv_2026_spacetime_spins}.
        \item The second element is the \textit{local Gauss laws}, which are derived from $G_{\mathcal{C}}$. These are a set of operators acting on $V_1$ that map configurations of gauge fields (errors) to other equivalent configurations, in the sense of Prop.~\ref{def:fault_equiv}. We will discuss these physically in Sec.~\ref{sssec:gauss_laws}, deferring the technical (and less intuitive) proof to App.~\ref{app_sec:gauss_laws}.
        \item The third element is the \textit{Wilson loops}, which are the gauge-invariant observables of a gauge theory \cite{Kogut_RMP_1979_lattice_gauge_theory}. These are identified with the redundancies $R_{\mathcal{C}}$ in the chain complex, Eq.~\ref{eq:gauge_complex_2}. Why do we care about redundancies? It turns out one can prove a somewhat more powerful version of Prop.~\ref{prop:detectors_are_stabilizers} -- redundancies of $G_{\mathcal{C}}$ are in fact exactly isomorphic to the detectors of the circuit (Prop.~\ref{prop:detectors_are_redundancies}). Again, we argue this heuristically in Sec.~\ref{sssec:detectors_redundancies}, but defer the technical proof to App.~\ref{app_sec:detectors_are_redundancies}. Given the modern detector-forward approach of fault tolerance, it is particularly apt that we identify detectors with gauge-invariant observables, since detectors should be invariant over different equivalent errors configurations, and detector are really the only real observables we obtain from a circuit. 
    \end{enumerate}

In Sec.~\ref{ssec:topology_gauged_SSC}, we study the the established gauge theory through the lens of the topology of Eq.~\ref{eq:gauge_complex_2}. As argued in Sec.~\ref{sssec:circuit_distance}, for us the relevant distance is actually the internal distance of the circuit, Def.~\ref{def:fault_distance}. In Sec.~\ref{sssec:gauge_fix_internal}, we show that this arises by modifying Eq.~\ref{eq:gauge_complex_2} through a simple gauge fixing procedure.

While the gauging procedure is standard, one has to take especial care with the temporal boundaries. In this construction, we will have fixed the boundary conditions in a particular way. Gauge fixing to obtain the internal distance can be considered a different kind of boundary condition. Up to this point we have not included input stabilizers in $G_{\mathcal{C}}$. In Sec.~\ref{ssec:boundaries}, we will include ${\rm ISG}(0)$ in the gauging procedure, which will necessitate the introduction of yet another set of boundary conditions, and require us to distinguish local and global redundancies. This will allow us to characterize different standard QEC experiments via boundary conditions on the gauged SSC.

Throughout, we will use repeated measurement of a repetition code as a leading example (see Figs.~\ref{fig:repeated_meas_rep_code_circs}, \ref{fig:rep_code_std_and_mem}). We will close in Sec.~\ref{ssec:examples} by looking at some further examples.

\subsection{Construction}\label{ssec:construction}

The idea of fault equivalence in Prop.~\ref{prop:fault_equiv} tells us that there are only a much smaller set of inequivalent error configurations, which are separated by the fact that they cannot reach each other via applications of $G_{\mathcal{C}}$. One would like to construct a gauge theory where the states are only labelled by these inequivalent error configurations \cite{Wen_2017_book}. Consider $g \in G_{\mathcal{C}}$. Two configurations related by the application of $g$ are equivalent as faults, and we want them to be equivalent as states in our theory as well. The natural way to do this is to project onto the $+1$ eigenspace of $g$; unfortunately $G_{\mathcal{C}}$ is non-abelian, so this set of operators don't commute. 

There are two ways to fix this. We could double the space, since doubled Paulis always commute. We will not consider this strategy here. Alternatively, we simply recognize that in the error correction and decoding process for Clifford circuits, the phases of the errors don't matter. As such, we can simply think of each Pauli in $\mathcal{P}_{m}$ instead as a vector in $\mathbb{F}^{2m}$. This is of course the space $V_2$ in Eq.~\eqref{eq:gauge_complex_1}, as augmented by measurement errors, which each contribute only one degree of freedom rather than two.

\begin{figure}[ht!]
\centering 
\includegraphics[width=0.7\linewidth]{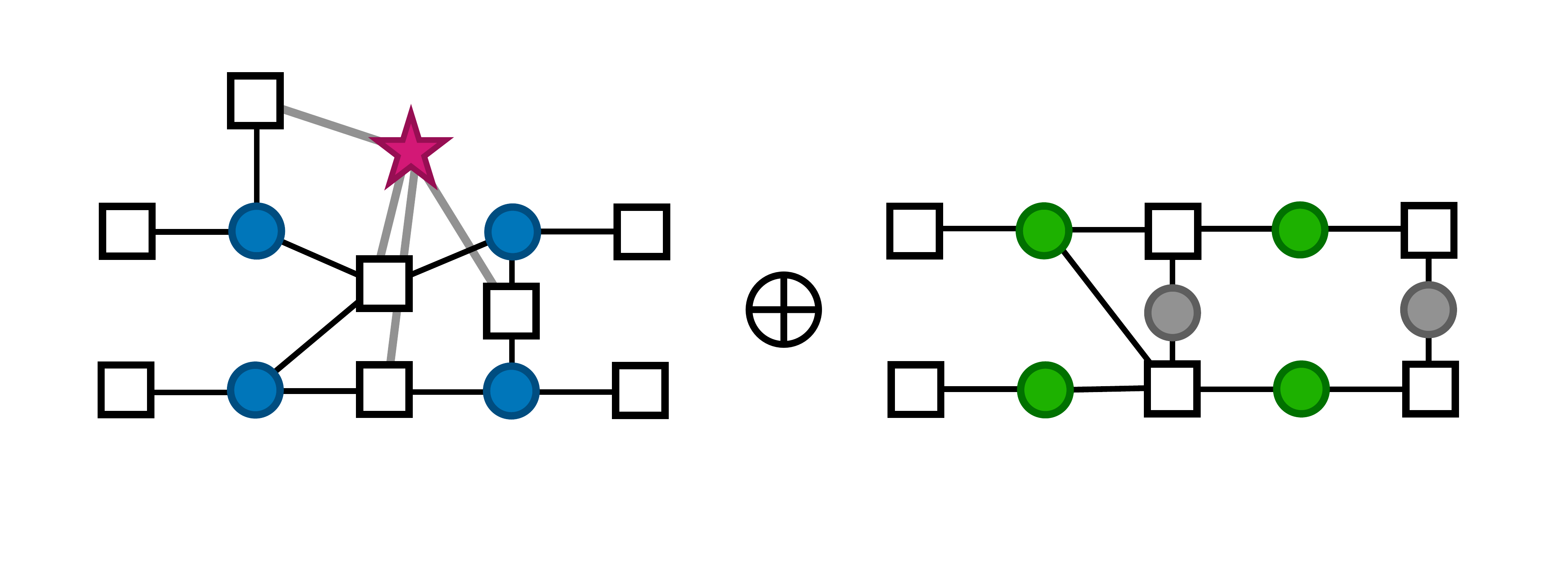} \caption{A schematic depiction of the gauge theory obtained by treating the circuit segment in Fig.~\ref{fig:detector_example} as a full circuit. As described in the main text, squares indicate gauge fields. Blue (green) circles indicate $Z(X)$-type matter fields, gray circles indicate measurement type matter fields, which are each connected by black edges to the gauge fields in the support of their associated Gauss law. The redundancy, which is derived from $G_{\mathcal{C}}$, is depicted by a magenta star, which is connected to the gauge fields in its support via gray edges.}\label{fig:small_full_example} 
\end{figure}

\bigskip

The first step in the construction is to introduce gauge fields, and define Gauss laws and Wilson loops on them. This section will also establish notation for the rest of the work as we switch from ECOs to matter and gauge fields. By and large I have tried to follow the terminology established in Ref.~\cite{Rakovszky_2023_arxiv_physics_ldpc_I}.  

\textit{From spacetime locations to classical spins -- } Recognizing the classical nature in which we decode errors in Clifford circuits, we can now map every spacetime location $(i, t)$ to a pair of classical bits (or spins), eg. $\sigma^t_{i, x}$, $\sigma^t_{i, z}$, where we emphasize that the $x, z$ subscript is simply a label, and $\sigma$ is a classical bit taking values of $0, 1$. The $j$-th measurements in time $t+0.5$ is mapped to a single classical bit $\sigma^{t+0.5}_{j, m}$. The `all-zeros' (or all `down-spins') state, which we henceforth refer to as the vacuum, is identified with an absence of circuit errors. When talking about some $g \in G_{\mathcal{C}}$, we denote by $\sigma[g]$ the product of the spins associated with the support of $g$. For instance, if $g = \eta^t(Z_n) \eta^{t-1}(X_n)$, then $\sigma[g] = \sigma_{n, z}^t \sigma_{n, x}^{t-1}$. We refer to $\sigma$ as \textit{matter fields}. While we philosophically regard these as classical degrees of freedom, it is useful to embed them in a Hilbert space; in that case we write $(\sigma^{X, Z})^t_{i, x/z}$, using the superscript to denote the type of Pauli, in a way completely decoupled from the subscript. 

\textit{Gauge fields -- } Next, we introduce our gauge fields. For every elementary propagator, which may be labelled by $\alpha_n = X_n, Z_n$ and time-step $t+0.5$, we introduce a single bit $\tau_{n, \alpha}^{t+0.5}$. For each measurement slice, which may be labelled by a time $t+0.5$ and measured operator $q$, we introduce a $\tau^{t+0.5}_q$. Again we emphasize that despite the $x, z$ subscript, $\tau$ is a classical bit taking values of $0, 1$. Conventionally, when embedding these degrees of freedom in a Hilbert space, we will think about $\tau$ as Pauli $Z$ and $\sigma$ as Pauli $X$. When talking about some $g \in G_{\mathcal{C}}$, we may notate the associated spin as $\tau[g]$. This is a single spin, in contrast to $\sigma[g]$, which is the product of multiple spins. Similarly, we may embed $\tau$ in a Hilbert space with superscripts $\tau^{X/Z}$ indicating a Pauli matrix. Diagramatically, we will depict $Z$-type matter fields with green circles, $X$-type matter fields with blue circles, and measurement type matter fields with gray circles. All gauge fields will be depicted by squares. This is consistent with the diagrammatic notation already established for the ECOs.

\textit{Gauss laws -- } The matter and gauge fields are tied together by Gauss laws. We introduce a set of local Gauss laws labelled by spacetime locations. The possible spacetime locations have the form $(s, t) = (i, x, t), (j, z, t), (q, m, t)$, where $s$ a tuple comprising an index and label $x, z, m$ describing the nature of the location, and $t$ is a time index. We denote by $(G_{\mathcal{C}}^T)_{(s, t), \cdot}$ the row of the matrix $G_{\mathcal{C}}^T$ associated with the spacetime location $(s, t)$, and we denote by $g \in (G_{\mathcal{C}}^T)_{(s,t), \cdot}$ the ECOs associated with the columns of that row that have a $1$ in them. In this notation, the Gauss laws operators are given by 
\begin{equation}\label{eq:matter_gauge_gauss}
    S_{s}^t = (\sigma^X)_{s} \prod_{ g \in  (G_{\mathcal{C}}^T)_{(s, t), \cdot}} \tau^X[g].
\end{equation}
These Gauss laws will help us relate the circuit errors associated with the matter fields to circuit errors associated with the gauge fields. In building a gauge theory, one often projects onto the symmetric space by only considering states where $S_s^t = +1$. However,
for our purposes it will be easier to think about $S_s^t$ as `allowed moves', i.e. a gauge field configuration labels the same state as another if related by the operator $S_s^t$, as this will allow us to stick to discussing classical configurations of fields. Diagrammatically, we will connect matter fields to all gauge fields in the support of their associated Gauss law via black edges.

\textit{Wilson loops -- } The Wilson loops are observables that live on the gauge fields. They are gauge-invariant, and hence in a gauge theory, correspond to physical observables. Gauge-invariance requires that they commute with all Gauss laws (as they should give the same result for any gauge equivalent configuration). In this prescription, Wilson loops are identified with the redundancies of $G_{\mathcal{C}}$. Since the term `redundancy' is used colloquially in several different ways, we formally define this.

\begin{definition}[Operator redundancy]\label{def:op_redundancy}
    A redundancy $R \subseteq G_{\mathcal{C}}$ is a set of ECOs $g \in R$ such that $\prod_{g \in R} g^{\otimes 2} = I$. Here, we have taken the two-fold tensor product to get rid of any ordering ambiguity, we may informally say that $g \in R$ multiplies to the identity.
\end{definition}

The operator redundancies that show up in $G_{\mathcal{C}}$ may be the single most important concept in this work. Unfortunately, we will eventually encounter several other kinds of (also important) redundancies. When we say `redundancy', we will usually be referring to Def.~\ref{def:op_redundancy} unless specified otherwise. We may arrange all independent redundancies in a matrix $R_{\mathcal{C}}$, where the rows correspond to a redundancy $R$ and the columns correspond to ECOs $g$, with a $1$ if $g \in R$ and $0$ otherwise. Our Wilson loops are then
\begin{equation}
    W_R = \prod_{g \in R} \tau^Z[g].
\end{equation}
It is easy to check that as matrices $R G_{\mathcal{C}} = 0$ \cite{kubica_2018_arXiv_ungauging_QEC}. Diagrammatically, we will depict redundancies as magenta stars, which are connected to the gauge fields in their support by a gray edge.

\bigskip

We have now constructed a proper gauge theory, with gauge fields, Gauss laws and Wilson loops. In its current form, the dynamics (errors) mix matter and gauge fields. Our next step, which will complete the gauging of the spacetime code, is to move the dynamics onto the gauge fields only. We will also specify a set of boundary conditions, although these will undergo much modification later.

\textit{Dynamical gauge fields -- } We will want a theory that only lives on the gauge fields. For $\mathbb{Z}_2$ fields there is a simple disentangling step one can carry out for this purpose (see \cite{Prem_scipost_2019_gauging_swaps, Haegeman_PRX_2015_gauging_states} for details), which amounts to replacing our Gauss laws with 
\begin{equation}\label{eq:gauge_gauss}
    S_s^t \rightarrow \prod_{ g \in  (G_{\mathcal{C}}^T)_{(s, t), \cdot}} \tau^X[g],
\end{equation}
or in other words, simplifying the theory by freezing out the degrees of freedom associated with the matter field.

\textit{Boundary fields -- } For technical reasons that will become clear shortly, we will introduce boundary fields. There are many ways to do this, which we will argue are associated with different kinds of QEC experiments and calculations. For now, we will go with a fairly minimal prescription, which we will refer to as the \textit{open boundary conditions.} This involves introducing additional boundary gauge fields for each of the $2n$ matter fields in the $T$-th time-step, $\tau^{f}_{i, x/z}$, and removing the spacetime locations (and their associated Gauss laws) at time $0$\footnote{If we had included a non-trivial input ${\rm ISG}$ then we would add additional boundary gauge fields before the circuit instead of removing spacetime locations.}. The Gauss laws (and $G_{\mathcal{C}}$ accordingly) associated with final boundary spacetime locations are modified like
\begin{equation}
    \begin{aligned}
        S^{T}_{i, x/z} &\rightarrow (\tau^X)^{f}_{i, x/z} (\tau^X)^{T-0.5}_{i, x/z},
    \end{aligned}
\end{equation}
while the Gauss laws associated with the locations $(i, x/z, 0)$ are removed. Diagrammatically, this corresponds to making sure the temporal boundaries only contain squares. Since gauge fields correspond to gate layers, schematically this looks like a rediscretization of the circuit like in Fig.~\ref{fig:state_convention_II}, where a final operation $C^f$ containing only the identity is introduced. Subsequently, when we speak in generalities, we are generally speaking about properties in the bulk, i.e away from the time boundaries, and we will try to be careful to be clear whenever the discussion starts touching on boundaries. 

At first glance, this procedure seems to introduce new redundancies to $G_{\mathcal{C}}$. For now, we simply assert that we should not add in these new redundancies at this stage. In App.~\ref{app_sec:obc}, we explain how these boundary conditions simply arise via careful treatment of the disentangling step in a way that does not introduce new redundancies. (Only in this step does our construction deviate a little from the prescriptions of Refs.~\cite{kubica_2018_arXiv_ungauging_QEC, Rakovszky_2023_arxiv_physics_ldpc_I}). Ultimately, our choices of boundary are guided by the desire to preserve the elements of fault tolerance, which is natural but not innate to the gauging procedure. We will refer to these as the \textit{input/output boundary fields}.

\bigskip

\begin{figure}[ht!]
\centering 
\includegraphics[width=0.75\linewidth]{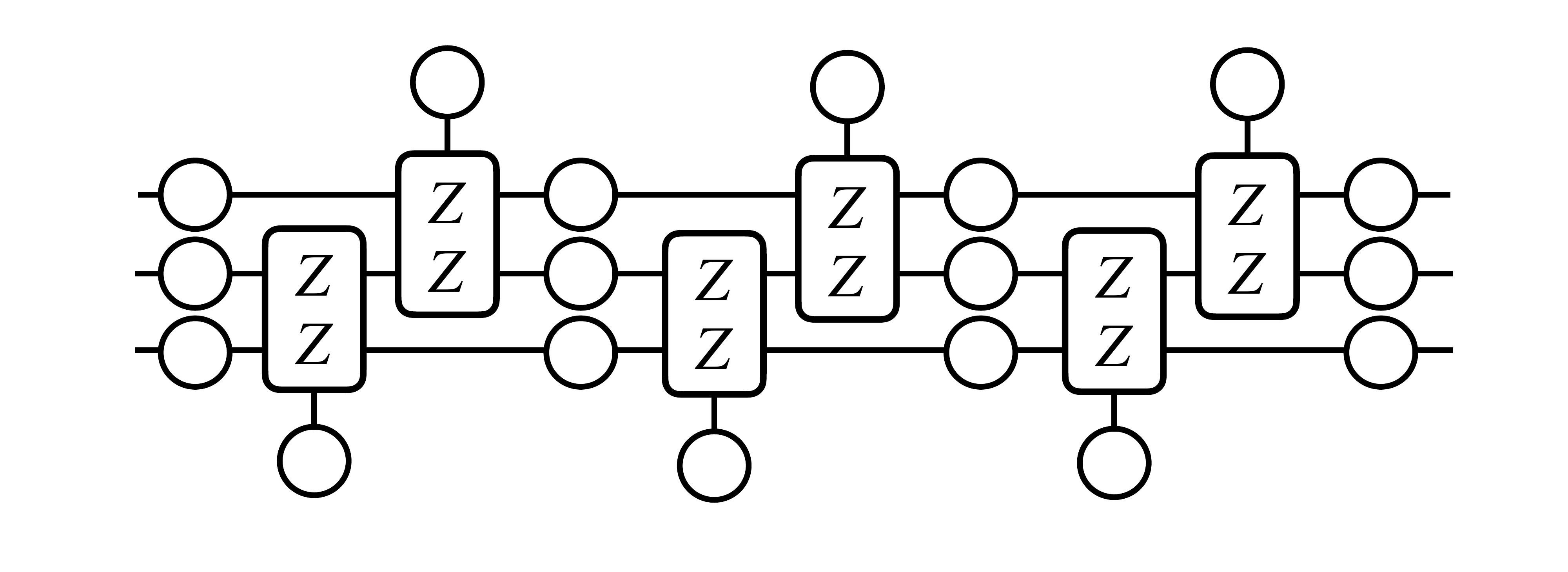} \caption{ Circuit depicting repeated measurements of a three-qubit repetition code, together with spacetime locations discretizing it for a phenomenological noise model. In between each integer layer of spacetime locations, we assume all stabilizers are measured.}\label{fig:repeated_meas_rep_code_circs} 
\end{figure}

\begin{figure}[ht!]
\centering 
\includegraphics[width=0.9\linewidth]{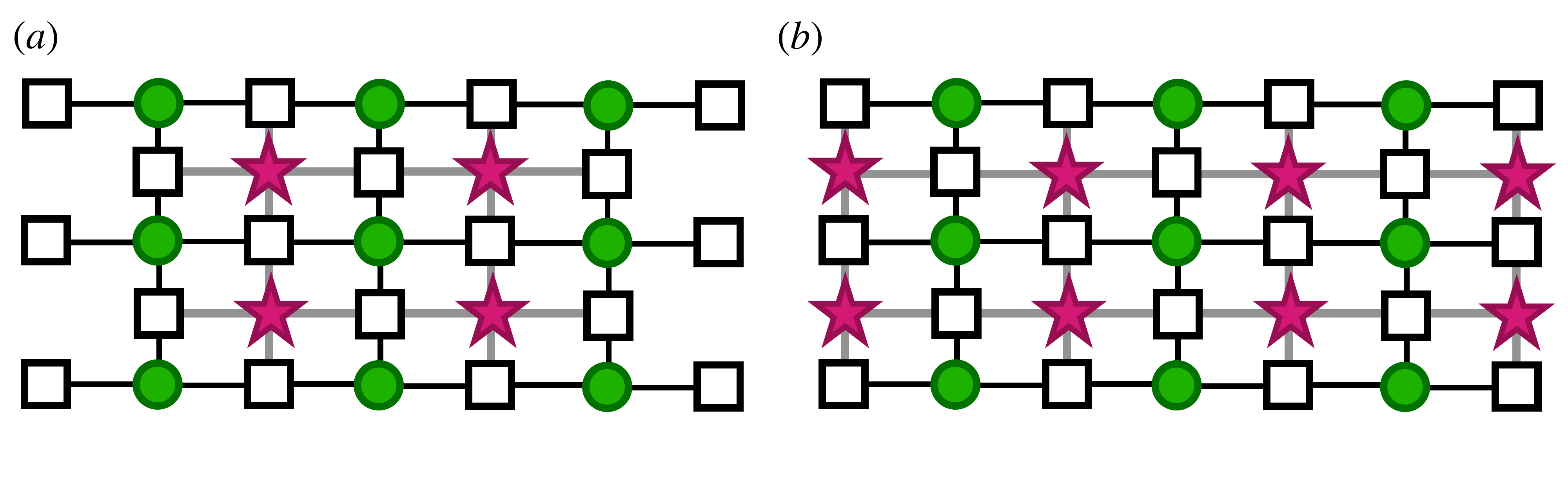} \caption{ The gauge theory associated with the with the repeated measurement repetition code circuit of Fig.~\ref{fig:repeated_meas_rep_code_circs}, with the open boundary conditions described in Sec.~\ref{ssec:construction}. We only depict the $Z$-type matter fields and the connected gauge fields and detectors, since the $X$-type fields are not relevant to the repetition code's ability to protect information. However, one should keep in mind that the $X$-type fields remain present but decoupled from the structure displayed here.
}\label{fig:rep_code_std_and_mem} 
\end{figure}

This completes the construction of the lattice gauge theory obtained by gauging the SSC, which is succinctly described in Eq.~\ref{eq:gauge_complex_2}. In Fig.~\ref{fig:small_full_example}, we depict the gauge theory obtained from the circuit segment in Fig.~\ref{fig:detector_example} as a minimal example containing detectors. In the next few sections, we will justify the construction by examining the correspondence between elements of fault tolerance and the elements of the gauged SSC. As a leading example, we depict the gauge theory obtained from the repeated syndrome measurement of a three-qubit repetition code in Fig.~\ref{fig:rep_code_std_and_mem}a. This is constructed from the circuit depicted in Fig.~\ref{fig:repeated_meas_rep_code_circs}.

\subsection{Elements of fault tolerance}\label{ssec:elements_FT_gauged_SSC}

\subsubsection{Errors as gauge fields}\label{sssec:errors_as_gauge_fields}

We now describe the mapping of gauge fields to errors, i.e. the error identification step, similar to what we did in Eqs.~\ref{eq:error_identification_I}, \ref{eq:error_identification_outcomes_I}. In principle, we should start from the error identification of Sec.~\ref{sec:clifford_circuits_and_gauge_generators} together with the Gauss laws and gauge fixing step to derive consistent error identification rules for the gauge fields. There are subtleties in doing so, which we discuss in Apps.~\ref{app_sec:obc},~\ref{app_sec:spackle_domain_wall} for completeness. Here we will simply assert a particular error identification, and note that it captures all errors in the ungauged theory.

Consider the operation $C^{t>}$, which comprises a unitary $U^{t>}$ and measurements of Paulis $q_1, ..., q_k$. In the absence of error, we have $| \psi^{t+1} \rangle \propto \prod_{j} (I + (-1)^{o(q_j)} q_j) U^{t>} | \psi^t \rangle$. 

\begin{figure}[ht!]
\centering 
\includegraphics[width=0.65\linewidth]{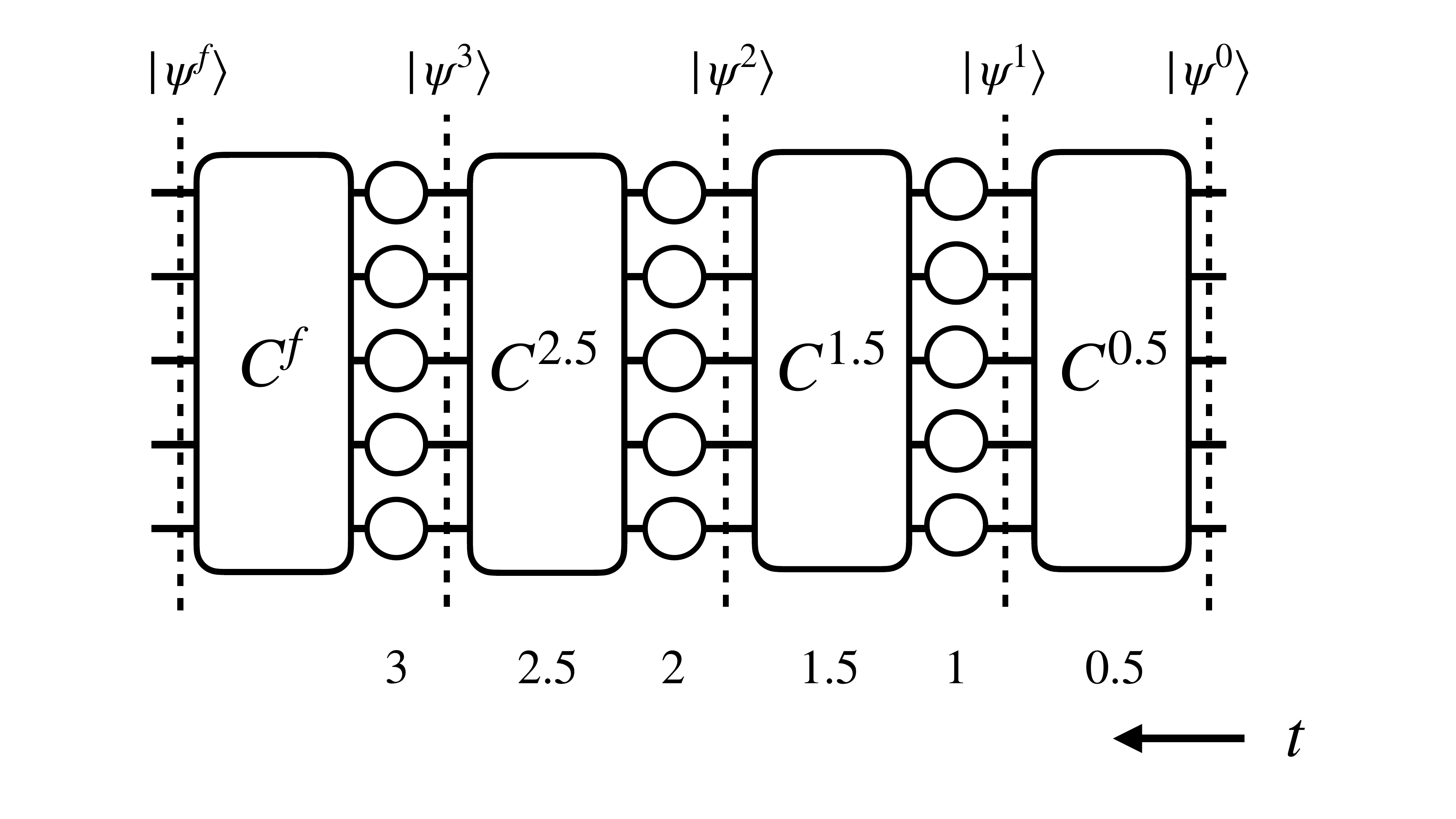} \caption{Schematic depicting our conventions for where we evaluate states and errors in the gauged SSC with OBC. Compared to Fig.~\ref{fig:state_convention}, we have appended a null operation (just the identity) to the end of the circuit, and removed a layer of spacetime locations from the start of the circuit. The reason for this is that we will think about errors as associated with the operation layers rather than the Pauli operators placed on spacetime locations.}\label{fig:state_convention_II} 
\end{figure}

We will identify the gauge configuration with $\tau_{i, x}^t = 1$ (or $\uparrow$) and $0$ everywhere else with a $Z_i$ error occurring just before the unitary part of $C^{t+0.5}$. This results in 
\begin{equation}\label{eq:error_identification_II}
    | \psi^{t+1} \rangle \propto \prod_{j} (I + (-1)^{o(q_j)} q_j) U^{t>} Z_i | \psi^t \rangle,
\end{equation}
and similarly $\tau_{i, z}^t = 1$ indicates an error $X_i$ instead. Note that since we ignore overall phases there is no ordering between the errors associated with $\tau_{i, z}$ and $\tau_{i, x}$. Finally, we identify the error $\tau_{l, m}^t = 1$, with $0$ everywhere else, with a measurement error in $q_l$, such that 
\begin{equation}\label{eq:error_identification_meas_II}
    | \psi^{t+1} \rangle  \propto \prod_{j} (I + (-1)^{o(q_j) + \delta_{j,l}} q_j) U^{t>} | \psi^t \rangle.
\end{equation}
Similar identification holds for $C^f$ (which is identity) and $C^{0.5}$, and the errors from each gauge field combines linearly. We note that we have not lost any information by trimming the spacetime locations at time $0$, since the identified circuit errors span exactly the same set of errors as what we originally identified Eq.~\ref{eq:error_identification_I}. One can think about these gauge fields as associated with gates rather than the idling time between them. Then, the reason we haven't lost any information is because no circuit can tell the difference between an error occurring in the idling time before the first gate, or just as the first gate is carried out.

The following definitions capture the analogous concept to circuit faults and Def.~\ref{def:fault_equiv}.

\begin{definition}[Gauge field fault]\label{def:gauge_fault}
    A fault (on the gauge fields) is non-trivial configuration of gauge fields, and may be associated with a vector $v \in V_1$ or a subset $F \subseteq G_{\mathcal{C}}$ of elementary circuit generators. It is associated with circuit faults as described in Eqs.~\ref{eq:error_identification_II}, \ref{eq:error_identification_meas_II}. 
\end{definition}

\begin{definition}[Gauge field fault equivalence]\label{def:gauge_fault_equiv}
    Two gauge field configurations are equivalent if their associated circuit faults are equivalent via Def.~\ref{def:fault_equiv}.
\end{definition}

To distinguish this notion of equivalence from the pre-existing notion of gauge fields being equivalent under Gauss law moves, we will refer to this as being equivalent as faults. Henceforth, when we say fault we are referring to Def.~\ref{def:gauge_fault} unless stated otherwise. The takeaway of this section is that the gauge fields capture exactly the same set of faults as the original spacetime locations did. In the next section, we will see how these are enforced by the Gauss laws.

Finally, we note that in our earlier discussion, we chose to identify errors in single spacetime locations directly with weight-one circuit errors. This follows the usual treatment of error identification in the discussion of spacetime codes. On the other hand, for a cleaner analogy to gauge theory, one might actually want to go back and re-do the identification of the original spacetime locations with errors in a slightly different way. One could instead try identifying a circuit error with the entire spackle of that error, which spreads the excitations in the matter fields across the circuit from either initial of final time boundary. The advantage of doing so is that it would allow us to draw the analogy between excitations in gauge fields and domain walls more tightly. However, this corrupts the notion of an error weight, as small circuit errors can now look like large weight errors in the spacetime code. We briefly explore this perspective in App.~\ref{app_sec:spackle_domain_wall}.

\subsubsection{Spacetime locations as Gauss laws}\label{sssec:gauss_laws}

Recall for a spacetime location the Gauss law operator takes the form, $S_{i, \alpha}^t = \prod_{ g \in  (G_{\mathcal{C}}^T)_{(i, \alpha, t), \cdot}} \tau^X[g]$. Under the identification of Def.~\ref{def:gauge_fault}, gauge field configurations that are equivalent under Gauss laws are also equivalent as faults.

\begin{prop}[Spacetime locations as Gauss laws]\label{prop:spacetime_gauss_law}
    Gauge configurations are equivalent under Gauss laws if and only if they are also equivalent as faults.
\end{prop}

This is proven in App.~\ref{app_sec:gauss_laws}. The main idea is as follows: We have established that with the stipulated boundary conditions and our error identification maps, gauge configurations are in one-to-one correspondence with spacetime faults; both map to the same set of circuit faults. From Prop.~\ref{prop:fault_equiv}, we know that two spacetime faults lead to equivalent circuit faults if and only if they are related to each other via multiplication by $\langle G_{\mathcal{C}} \rangle$. Hence, to show Prop.~\ref{prop:spacetime_gauss_law}, we show that the Gauss laws map gauge configurations to each other in a way that is isomorphic to $\langle G_{\mathcal{C}} \rangle$.


This can be stated more succinctly as follows: Classes of gauge field faults that are equivalent under Gauss laws in the gauged SSC are in one-to-one correspondence to equivalence classes of effects.

\subsubsection{Detectors as redundancies}\label{sssec:detectors_redundancies}

The final element of the gauged SSC are the Wilson loops, which come from redundancies in $G_{\mathcal{C}}$. The key result is that these are actually detectors. 

\begin{prop}[Detectors are redundancies]\label{prop:detectors_are_redundancies}
    When regarded as groups (see App.~\ref{app_sec:detectors_are_redundancies}), the detectors, $\mathcal{D}$, of a circuit $\mathcal{C}$ are isomorphic to redundancies, $\mathcal{R}$, in $G_{\mathcal{C}}$.
\end{prop}

We give a few comments on the content of this proposition. First, one often already refers to detectors as redundancies. A detector is a set of measurements whose product has fixed parity in the noiseless case -- hence detectors are in fact characterized by redundant measurements. Here, we emphasize that when we say redundancies, we mean redundancies in the ECOs, in the sense of Def.~\ref{def:op_redundancy}. While ECOs characterize measurements, we note that this is not a given, as it required us to have chosen the ECOs to be over-complete in a specific way. As an example, suppose we did not introduce additional ancilla locations and instead used the formulation of either Ref.~\cite{Pesah_arXiv_2025_FT_transformations} or Ref.~\cite{Blackwell_arXiv_2025_distance_floquet}, which only keeps propagators of operators that commute with all measurements. Then choosing a basis for the gauge group/benign errors such that Prop.~\ref{prop:detectors_are_redundancies} holds would obstruct the error identification step and in turn Prop.~\ref{prop:spacetime_gauss_law}. 

Next, we should also note that a version of Prop.~\ref{prop:detectors_are_redundancies} is implicitly realized by the stabilizer flow formalism \cite{McEwen_quantum_2023_time_dynamics}, or in the ZX calculus \cite{Bombin_quantum_2024_unifying_FT}. This suggests that one could think about SSC as a discretization of stabilizer or Pauli flows, which in turn could be regarded as continuum versions of our lattice gauge theory. 

The proof of Prop.~\ref{prop:detectors_are_redundancies} is presented in App.~\ref{app_sec:detectors_are_redundancies}. We note one can quickly intuit Prop.~\ref{prop:detectors_are_redundancies} from Cor.~\ref{cor:gauge_op_ISG} -- the final measurement of a detector must be deterministic (conditioned on all earlier measurements) in a noiseless run of a circuit, which can only happen when we measure an element of an ISG. By Cor.~\ref{cor:gauge_op_ISG}, such an ISG element may be generated by multiplying ECOs from before that time-step, and the measurement slice associated with the final measurement is also equal to this product, so this form a redundancy. However, to show the full statement of Prop.~\ref{prop:detectors_are_redundancies}, we must be a little more careful. We briefly outline the proof strategy here. 

We start by constructing  an explicit map $\phi$ from redundancies to detectors, and show that it is a homomorphism. Each redundancy $R$ may be partitioned into measurement slices $M$ and elementary propagators $P$. The map then takes $\phi(R) = M$, with the measurement slices identified with the measurements in the circuit. To show that $M$ is indeed a detector, we show that multiplying out elements of a redundancy causally (i.e. in time-order) always results in an element of an ISG, such that the final measurement of $M$ is a measurement of an ISG element. Homomorphism follows from the fact that $\phi$ is essentially taking subsets and the group operations of both sides are symmetric difference.

After showing that $\phi$ is a homomorphism, we show in turn that it is injective and surjective. Injectivity can be shown by applying the fact that there are no propagator redundancies (Lem.~\ref{lem:no_prop_redundancies}). Surjectivity is a little more difficult, and we must make recourse to the ancestry formalism of Ref.~\cite{Blackwell_arXiv_2025_distance_floquet}.

Finally, having shown that detectors may be identified with redundancies and hence the Wilson loops of the gauged SSC, we must show that they interact with circuit errors via acting on gauge configurations in the expected way. 

\begin{prop}[Wilson loops acts like detectors]\label{prop:wilson_like_detectors}
    Let $W$ be a Wilson loop, and $| \phi \rangle \in H_{\tau}$ be a gauge configuration. Then $\langle \phi | W | \phi \rangle = -1$ if and only if the circuit fault associated with $| \phi \rangle$ violates the detector associated with $W$.
\end{prop}

\subsection{Topology of the gauged SSC}\label{ssec:topology_gauged_SSC}

We now revisit the interpretation of the gauged SSC as a chain complex, 
\begin{equation}\label{eq:gauge_complex_w_boundaries}
        V_3 \xrightarrow{ \Omega S_{\rm SSC}^T} V_2 \xrightarrow{G_{\mathcal{C}}^{\rm obc}} V_1^{\rm obc} \xrightarrow{R_{\mathcal{C}}} V_0,
\end{equation} 
where $G_{\mathcal{C}}^{\rm obc}, V_1^{\rm obc}$ are $G_{\mathcal{C}}, V_1$ modified by the boundary conditions described in the previous section.

Recall that gauge field configurations and thus circuit faults can be identified with vectors $v \in V_1$. From Prop.~\ref{prop:wilson_like_detectors}, we know that a circuit fault does not trigger any detectors if and only if $v \in {\rm ker} R_{\mathcal{C}}$, and from Prop.~\ref{prop:spacetime_gauss_law}, we know that the circuit faults $v, v'$ are equivalent if and only if $v = v' + w$ for some $w \in {\rm Im} G_{\mathcal{C}}'$. Given Def.~\ref{def:fault_distance}. This tells us that equivalence classes of undetectable errors are given by the first homology class,
\begin{equation}
    H_1 = \frac{{\rm ker} R_{\mathcal{C}}}{{\rm Im} G_{\mathcal{C}}^{\rm obc}}.
\end{equation}
However, as discussed in Sec.~\ref{sssec:circuit_distance}, this does not give us a good definition of a circuit distance, as weight-one errors on the temporal boundaries, corresponding to errors happening before or after the circuit, are all in $H_1$. Note that while we categorically do not want to interpret the gauged SSC as a QEC code, we will colloquially refer to elements of ${\rm ker} R_{\mathcal{C}}$ as the `logicals' of the gauged SSC.

\subsubsection{Gauge fixing and internal distance}\label{sssec:gauge_fix_internal}

Fortunately, there is an easy fix -- we shall simply fix the boundary gauge fields to be $\tau[g] = 0$ for all $g$ on the temporal boundaries, as depicted on the left side of Fig.~\ref{fig:rep_code_gauge_fixed} for the repetition code circuit. This is analogous to imposing $\tau^Z[g]$ as a stabilizer. Analogously thinking about the Gauss law operators as stabilizers, we need to remove one constraint that anticommutes with $\tau^Z[g]$ for each $\tau^Z[g]$. By construction, these are just the Gauss laws on the time-slice $T$ and $1$. Removing these fields leads to the figure on the right side of Fig.~\ref{fig:rep_code_gauge_fixed}. We depict with red and yellow boxes inequivalent undetectable errors, corresponding to a logical $X$ error in the repetition code and a persistent readout error respectively. Since we have not included any input ISG, there can be no detectors supported on the removed gauge fields, all detectors are retained in this gauge fixed theory.

We modify the chain complex accordingly,
\begin{equation}\label{eq:internal_gauge_complex}
        V_3 \xrightarrow{ \Omega S_{\rm SSC}^T} V_2 \xrightarrow{G_{\mathcal{C}}^{\rm int}} V_1^{\rm int}  \xrightarrow{R_{\mathcal{C}}} V_0,
\end{equation} 
where the superscript ${\rm int}$, short for `internal' indicates the above modifications. We note that we don't have to actually modify $V_2$; simply set the relevant columns of $G_{\mathcal{C}}^{\rm int}$ to $0$. Since all Gauss laws and detectors still hold, and the remaining gauge fields still match the internal circuit errors one-to-one, the following theorem follows.

\begin{thm}\label{thm:homology_distance}
The internal distance of a circuit $\mathcal{C}$ is
\begin{equation}
d_{\mathrm{int}}(\mathcal{C})
= \min \left\{ \mathrm{wt}(v) \;:\; v \in \ker R_{\mathcal{C}} \setminus \operatorname{Im} G_{\mathcal{C}}^{\mathrm{int}} \right\},
\end{equation}
where $\mathrm{wt}(v)$ denotes the Hamming weight of $v \in V_1$.
\end{thm}

\begin{figure}[ht!]
\centering 
\includegraphics[width=0.9\linewidth]{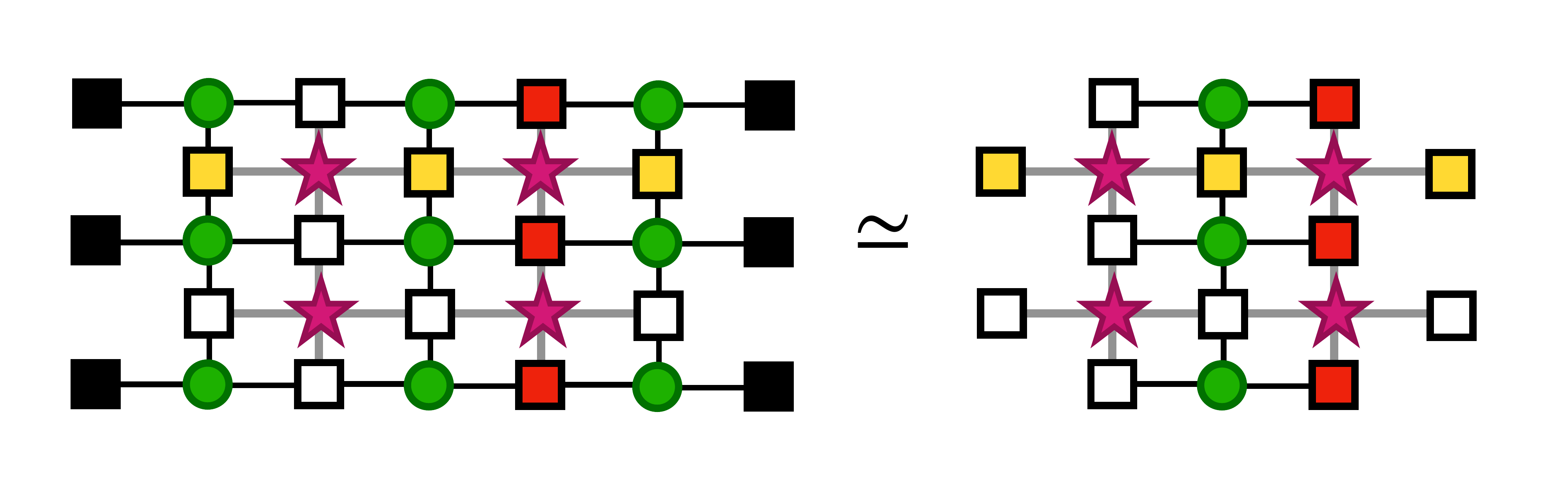} \caption{ The gauge theory associated with the circuit of Fig.~\ref{fig:repeated_meas_rep_code_circs}, with the open boundary fields fixed to be $0$, and the relevant Gauss laws removed. The set of gauge fields colored in red represents a logical error resulting from errors on qubits, and the set of gauge fields colored in yellow represents a logical error resulting from errors in readouts. The internal distance of the circuit in Fig.~\ref{fig:repeated_meas_rep_code_circs} is hence $d_{\rm int}(\mathcal{C}) = 3$.
}\label{fig:rep_code_gauge_fixed} 
\end{figure}

\subsection{Temporal boundaries, local and global redundancies, and QEC experiments}\label{ssec:boundaries}

Up till now, we have not considered the possibility of an input ${\rm ISG}(0)$ in the gauged SSC. We now consider that possibility. The formalism extends in a natural way, and the crucial claims, Prop.~\ref{prop:spacetime_gauss_law}, Prop.~\ref{prop:detectors_are_redundancies}, Thm.~\ref{thm:homology_distance} still hold up to small modifications of the input boundary conditions. However, conceptually separating the input ISG allows us to discuss a different kind of redundancy, Def.~\ref{def:global_redundancy}, that can arise in this situation.

To anchor our discussion, let us first consider the case of a standard memory experiment, comprising the following steps: (1) Prepare a logical eigenstate of the desired QEC code. (2) Carry out repeated syndrome measurements. (3) Measure the output of the circuit, usually transversally, to see if the eigenstate survives. 

Strictly speaking, the SSC for a memory experiment contains no logical qubits. To see this, recall the prescription for the SSC is to put $\eta^0(s)$ for all $s \in {\rm ISG}(0)$ into the gauge group. Recall also that the rank of an ISG can never decrease under measurements and unitaries. In the case of a memory experiment, ${\rm ISG}(0)$ is full rank, so by Cor.~\ref{cor:gauge_op_ISG}, the gauge group of the SSC contains a maximally spanning set of commuting operators in every time-step. This forms a maximally spanning set of operators over all spacetime locations. Hence any operator that commutes with the gauge group must be contained in the gauge group, and so the resulting SSC has no logical operators. To arrive at an SSC with a distance and hence logical operators, we must single out the logical operators and exclude them from the gauge group when constructing the SSC.

What happens in the gauged SSC for such a memory experiment? Specifying an input ISG via ECOs introduces the possibility of a new kind of redundancy, which we will refer to as a \textit{global} redundancy. 
\begin{definition}[Global redundancy]\label{def:global_redundancy}
    A global redundancy $R_{\rm global} \subseteq G_{\mathcal{C}}$ is a redundancy that:
    \begin{enumerate}
        \item contains ECOs associated with every time-step in the circuit, and
        \item cannot be reduced (via the symmetric difference with other redundancies) into a redundancy that does not have support on every time-step.
    \end{enumerate}
    Conversely a local redundancy is any redundancy that is not a global redundancy.
\end{definition}
For physical gauge theories, it is more natural to define a global redundancy as one whose support scales extensively in system size. Since we only explicitly discuss temporal boundaries in this work, it is more useful to capture the analogous definition arising from scaling $T$, which only accounts for one dimension of the system size.

In the memory experiment, elements of the input ISG that are stabilizers only give rise to local redundancies, whereas elements of the input ISG associated with logical operators give rise to global redundancies. This is in fact a general principle: if an operator survives the entire circuit before it is measured, then we probably were not using it for QEC as a detector. Hence, this leads to a natural way to modify our gauging procedure to account for input states, namely: \textit{We will only identify local redundancies as Wilson loops.} Global redundancies on the other hand, can be thought of as ways to specify boundary conditions. 

Let us now compare to the earlier approach of removing logical operators of the input code from the SSC by hand. There are two advantages to instead characterizing these degrees of freedom as global redundancies in the gauged SSC. First, this allows us to look at a circuit without committing to an input code -- the logical degrees of freedom then arise naturally from the dynamics of the circuit, in line with the modern view of dynamical QEC \cite{Hastings_Quantum_2021_floquet_codes, McEwen_quantum_2023_time_dynamics}. Second, by excluding global redundancies from $R_{\mathcal{C}}$, they naturally become `logicals' of the gauged SSC; this is true regardless of whether we view them as ECOs or whether the discretization of the circuit has an odd or even number of time steps. We emphasize that \textit{distinguishing global and local redundancies is a physically-motivated choice that we are making.} It is completely consistent to use global redundancy as a detector if desired, although it is almost never useful in error correction\footnote{One can imagine using the gauged SSC to analyze gadgets, eg. flag gadgets \cite{Chao_PRL_2018_Flag_I, Chao_PRXQ_2020_flag_II, Debroy_PRA_2020_extended_flag}. When examining these sub-circuits alone, certain detectors arise as global redundancies. However, since we never scale these circuits with time, it is not useful to think about these as global redundancies.}. On the other hand, when we consider the application of the gauged SSC to learning theory in Sec.~\ref{ssec:learning_Pauli_noise}, we will find it useful to treat global and local redundancies on the same footing.

Next, we will review and modify the construction of the boundaries of the gauged SSC again to obtain several possible boundary conditions beyond the original open boundary conditions, with each corresponding to a different kind of QEC experiment. We note that Thm.~\ref{thm:homology_distance} holds for all these cases because all these boundary conditions respect error identification, equivalence under Gauss laws, and the relationship between errors and detectors.

\subsubsection{Open boundary conditions and memory experiments}

In light of the above discussion, let us review the construction of open boundary conditions for the gauged SSC again. We do not need to modify the output boundary fields. Now that we have introduced an input ISG, we simply move the input boundary fields one time-step back, with no change in their interpretation as errors happening before the circuit. 

There are two subtleties to address here: First, as discussed in the previous section, the input ISG may now introduce global redundancies. We might introduce a stabilizer that is never measured until the end of the circuit, in correspondence with the logical operators in a memory experiment. We reiterate that in general we will exclude global redundancies from $R_{\mathcal{C}}$. 

Second, one thing we have emphasized is the locality of the errors identified with gauge fields. An input ISG can potentially spoil this. In the bulk of the circuit, errors on gauge fields associated with measurements are local, as they only take a single-bit readout error. However, errors in the stabilizers in an input ISG are only local if we assume that these were directly measured just before the circuit, which is often not the case. In particular, an error in an input stabilizer corresponding to the logical operator of some desired input code may correspond to a very non-local error physically. 

There are several ways one might want to deal with this. First, one can simply discard the gauge field associated with the logical stabilizer. Alternatively, in a memory experiment such a stabilizer is often prepared with local measurements. For instance, with a CSS code, we may prepare a product of all $Z_i$ eigenstates and then measure all $X$-type stabilizers, in order to prepare an eigenstate of all $Z_L$-type logicals. We could include this entire process in the boundary of the gauged SSC -- the result is a gauged SSC very similar to the original open boundary conditions, except that we only have open boundaries of the $Z$-type. Note that we still have to exclude the global redundancies arising in this setup.

Finally, one can leave this open-ended, and say this is left to the task of concatenating (in time) gauged SSCs. In spirit this is what we were imagining doing when we excluded input ISGs in the first place, but we leave the exploration of this possibility to future work.

\subsubsection{Rough boundary conditions and noiseless layers}

Instead of introducing input ISGs as gauge fields, we may instead want to introduce a layer of noiseless state preparation or a layer of noiseless measurements after the circuit. This is useful for making contact with standard models for fault tolerance, which assume a layer of perfect correction at the end \cite{Aliferis_arXiv_2005_concatenation, Gottesman_PRA_1998_FTQC, Gottesman_arXiv_2022_opportunities_FT}. These potentially form detectors on the boundaries, as the first layer of noisy measurements may be compared to the noiseless input stabilizers, and similar for the output. To account for this, we start with the open boundary conditions for the input ISG-free case. We then introduce the relevant detectors on the boundaries as though we had introduced an additional layer of virtual measurements before or after the circuit. These detectors comprise all the checks that together with the virtual measurements would have formed a redundancy, but without the virtual measurement itself. This forms something analogous to a rough boundary in a surface code \cite{Bravyi_arXiv_1998_surface_code}, see Fig.~\ref{fig:rep_code_other_bds}a for an example of the repetition code circuit of Fig.~\ref{fig:repeated_meas_rep_code_circs} with a noiseless layer of measurements after the circuit.

\subsubsection{Periodic boundary conditions and entanglement fidelity experiments}

The memory experiment resulted in minimal change to the boundary conditions we started out with. We now discuss the case of `entanglement fidelity' experiments \footnote{See for instance Remark 2.8 in Ref.~\cite{Blackwell_arXiv_2025_distance_floquet}}. In addition to the circuit of interest $\mathcal{C}$, we prepare a reference system with the same number of qubits that is maximally entangled with the input into $\mathcal{C}$. The reference system and state preparation is assumed to be noiseless. After running $\mathcal{C}$, we perform Bell measurements between the two systems, which are also assumed to be noiseless. This evaluates the entanglement fidelity, which is useful because it provides bounds for other performance metrics \cite{Gilchrist_PRA_2005_distance_measures, Nielsen_arXiv_1996_EF}, and is often analytically tractable to work with \cite{Zheng_PRL_2024_near_optimal}.

The discretization of such a circuit is depicted in Fig.~\ref{fig:EF_expt}. The main circuit is discretized in the usual way. Since no operations are carried out on the ancilla, we introduce just one layer of spacetime locations to them. Finally, we introduce state preparation and measurement ECOs of the form
\begin{equation*}
    \eta^{0}(X_i X_i^{(a)}), \eta^{0}(Z_i Z_i^{(a)}), \eta^{T}(X_i X_i^{(a)}), \eta^{T}(Z_i Z_i^{(a)}),
\end{equation*}
for each qubit $i$ in the circuit, with superscript $(a)$ indicating its ancillary counterpart. Schematically, on the level of abstraction in Fig.~\ref{fig:EF_expt}, this looks the same for both $X$ and $Z$ in the state preparation and measurement parts. 

Next, we can perform a similar gauge fixing step to Sec.~\ref{sssec:gauge_fix_internal}. The state preparation and measurement steps are associated are noiseless, so we fix the associated gauge fields to be $0$. Recall that for each of these fields we fix, we must refactor Gauss laws so that there is remains only one independent Gauss law that is spoiled by the gauge fixing, and we remove this Gauss law from the picture. The result is that the Gauss laws associated with noiseless ancilla locations disappear and the Gauss laws associated with the initial and final layer of matter qubits are merged into one set of Gauss laws, which look like
\begin{equation*}
    S_{i, x/z} = (\tau^X)^{T-0.5}_{i, x/z} (\tau^X)^{0.5}_{i, x/z}.
\end{equation*}
This effectively identifies the initial layer of matter fields with the final layer, giving us \textit{periodic boundary conditions}, see eg. Fig.~\ref{fig:rep_code_other_bds}b. Finally, we add that as long as the circuit $\mathcal{C}$ does not enforce a full rank ISG on an ISG-free input, the periodic boundary conditions will induce a set of global redundancies. These correspond to logical degrees of freedom that are allowed to propagate through the circuit, but should be excluded as detectors (in fact, the violation of these global redundancies directly yields the result of the entanglement fidelity experiment). 

\begin{figure}[ht!]
\centering 
\includegraphics[width=0.5\linewidth]{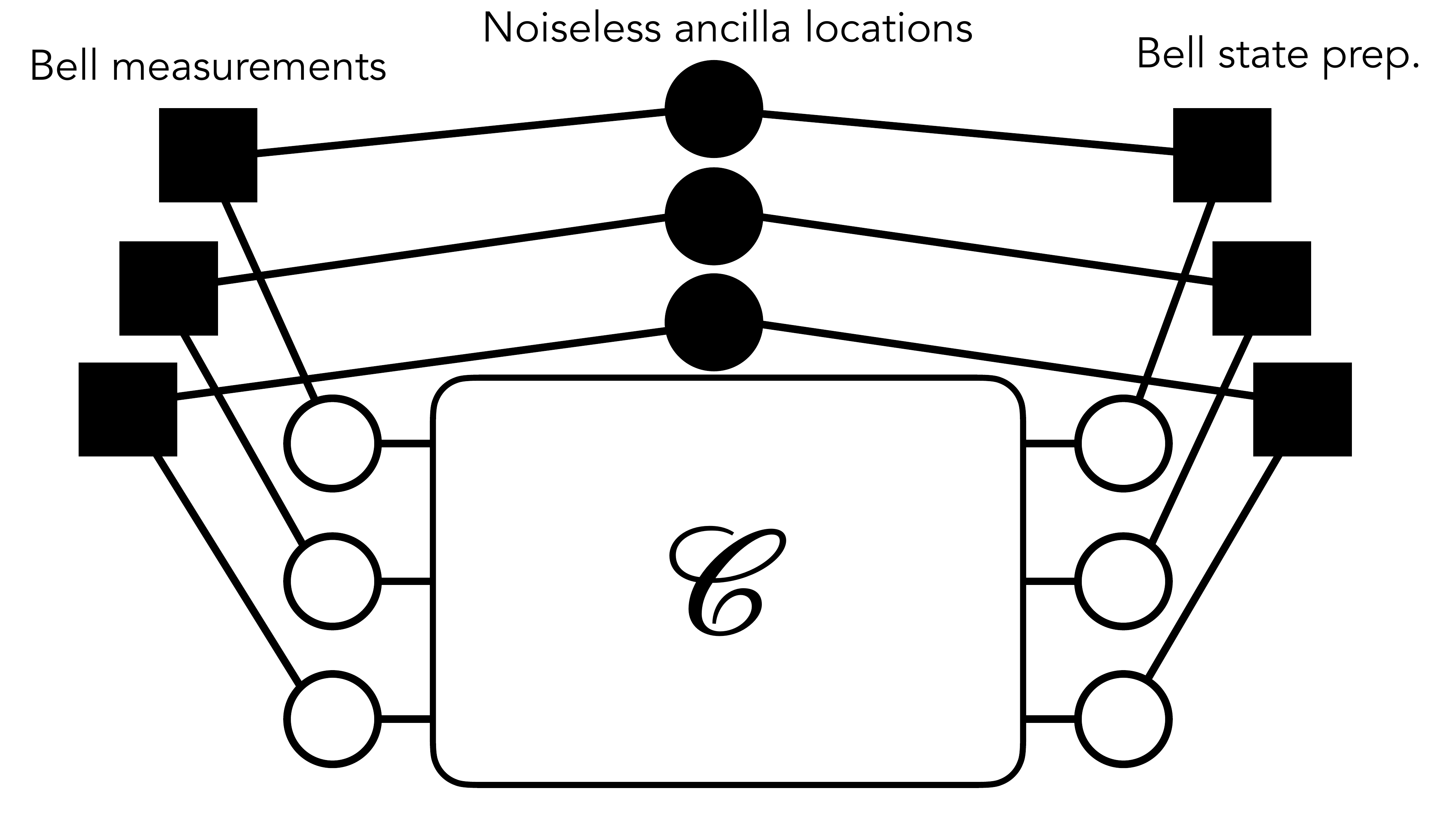} \caption{ 
Schematic depiction of the gauged SSC associated with an entanglement-fidelity experiment. The circuit $\mathcal{C}$ acts only on the system qubits (which is itself discretized and has internal fields and so on), while a noiseless reference system is prepared in Bell pairs with the input and Bell measurements are performed on the output. The ECOs associated with state preparation measurement and preparation have the same form for $X$ and $Z$ type operators, hence the lack of color. Since the reference locations, Bell-state preparation, and Bell measurements are assumed noiseless, their associated gauge fields are fixed to zero. After this gauge fixing, the initial and final matter fields may be identified, producing periodic temporal boundary conditions for the system circuit.
}\label{fig:EF_expt} 
\end{figure}

\subsubsection{Floquet codes}\label{sssec:floquet_codes}

We now discuss Floquet codes \cite{Hastings_Quantum_2021_floquet_codes} as an interesting application of the gauged SSC with periodic boundary conditions. We begin by recalling several key concepts from Ref.~\cite{Blackwell_arXiv_2025_distance_floquet}. In general, a Floquet code is defined by a periodic sequence of measurements. Starting from an arbitrary input state, the Floquet code is said to be in the \textit{steady stage} when the rank of the ISG stops increasing; in other words, when the ISGs of the circuit form the periodic set that will repeat indefinitely as the sequence of measurements is carried out. We restrict to \textit{bounded-inference} Floquet codes, which means the sign of any element in an ISG can be inferred from measurements from only a finite number $\mu$ of earlier time-steps; this is also referred to as the \textit{ancestry} of the element. Ref.~\cite{Blackwell_arXiv_2025_distance_floquet} then defines the code distance of a bounded-inference Floquet code as the minimal weight of a \textit{non-benign} undetectable error inserted in the steady stage. This is justified by their proof that such an error is always equivalent to a logical operator of some ISG within $\mu$ time-steps of the latest-time support of the error. 

 We briefly discuss the concept of a benign error and how it shows up in our work. Recall that a benign error is essentially the same set of operators as our ECOs (although Ref.~\cite{Blackwell_arXiv_2025_distance_floquet} did not include ancillary measurement locations). For readers familiar with the benign error formalism, we first note that all operators generated by the ECOs that are supported only on the integer-time locations are exactly the set of benign errors of a circuit. For a finite-time circuit (which is not considered in Ref.~\cite{Blackwell_arXiv_2025_distance_floquet}), this is not true on the boundaries, but becomes true if we identify time-steps $T$ with $2T$ and so on, which is exactly what the periodic boundary conditions do. 
    
The benign error formalism characterizes measurement errors via pairs of Paulis just before and just after the measurement, such that the earlier error propagates into the later error, and this Pauli anti-commutes with the measurement at the time of measurement (whether this is the Pauli itself $p$ or the conjugated $U^{t>} p (U^{t>})^{\dag}$ depends on convention). The ECOs encode this fact by explicitly making such errors equivalent to measurement errors via the half-integer time locations. As such, the ECOs (modulo measurement dephasers as usual) and benign errors really characterize exactly the same set of errors.

With this blitz review of Ref.~\cite{Blackwell_arXiv_2025_distance_floquet}, we can now discuss how the gauged SSC with periodic boundary conditions can be used to study Floquet codes. We pick some $T \gg \mu$ for convenience. Due to the fact that ISG elements are can be causally formed by ECOs (Cor.~\ref{cor:gauge_op_ISG}) and the time wraps around, every time-step of the SSC with periodic boundary conditions will have the ISG elements of the steady stage, so the gauged SSC indeed characterizes the steady stage of a Floquet code. Since $T \gg \mu$, the ancestry of these ISG elements never see the temporal boundaries of the system. Hence, the SSC with periodic boundary conditions contains all benign errors associated with the steady stage, which then map into Gauss laws in the gauged SSC. A non-benign undetectable error in the steady stage is thus a gauge field configuration that maps to an undetectable error which is not equivalent to the null error. Hence, if we have chosen a sufficiently large $T$, the distance of the gauged SSC with periodic boundary conditions will be the distance of the Floquet code, and upper bounds it otherwise.

\begin{figure}[ht!]
\centering 
\includegraphics[width=0.9\linewidth]{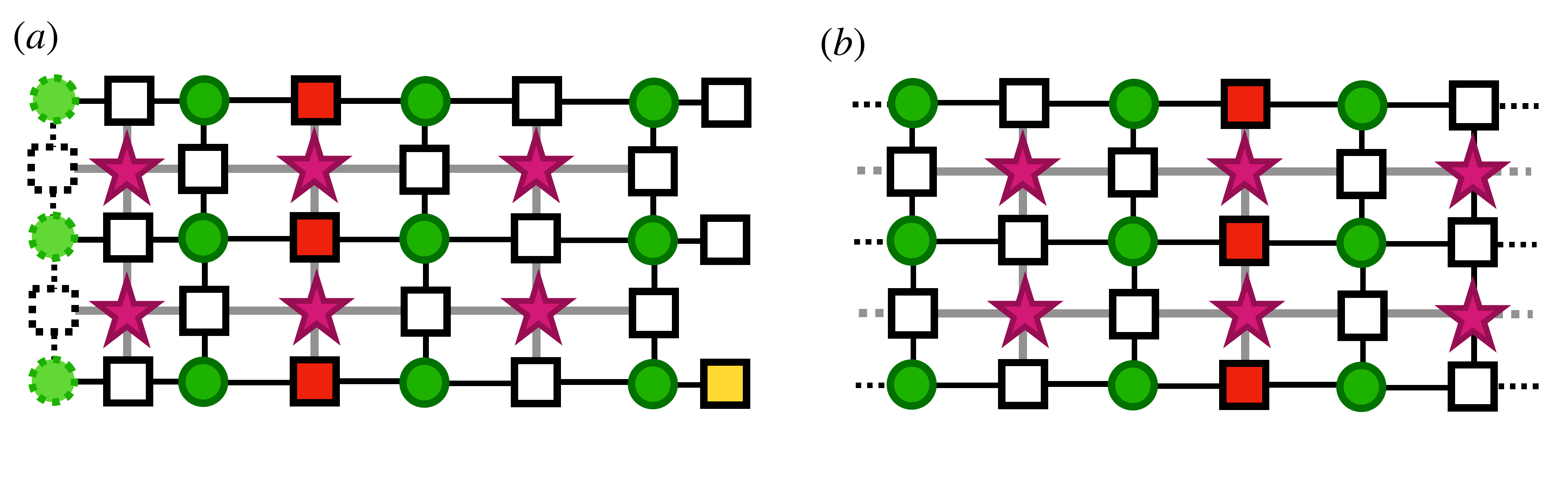} \caption{ The gauge theory associated with the repeated-measurement circuit of Fig.~\ref{fig:repeated_meas_rep_code_circs} under alternative temporal boundary conditions. \textbf{(a)} Rough/output boundary conditions obtained by adding a noiseless layer of virtual measurements after the circuit. Boundary detectors compare the final noisy syndrome data with this noiseless layer. \textbf{(b)} Periodic temporal boundary conditions (see Fig.~\ref{fig:EF_expt}). Colored (red and yellow) gauge field configurations indicate representative undetectable error classes. Note that for periodic boundary conditions, we (choose to) identify global redundancies with logical degrees of freedom rather than detectors.
}\label{fig:rep_code_other_bds} 
\end{figure}

\subsection{Further examples}\label{ssec:examples}

We conclude this section by providing several examples of gauged SSCs. We emphasize that these are all well-known fault-tolerant circuits that we are simply re-interpreting in terms of the gauged SSC.

\begin{figure}[ht!]
\centering 
\includegraphics[width=0.5\linewidth]{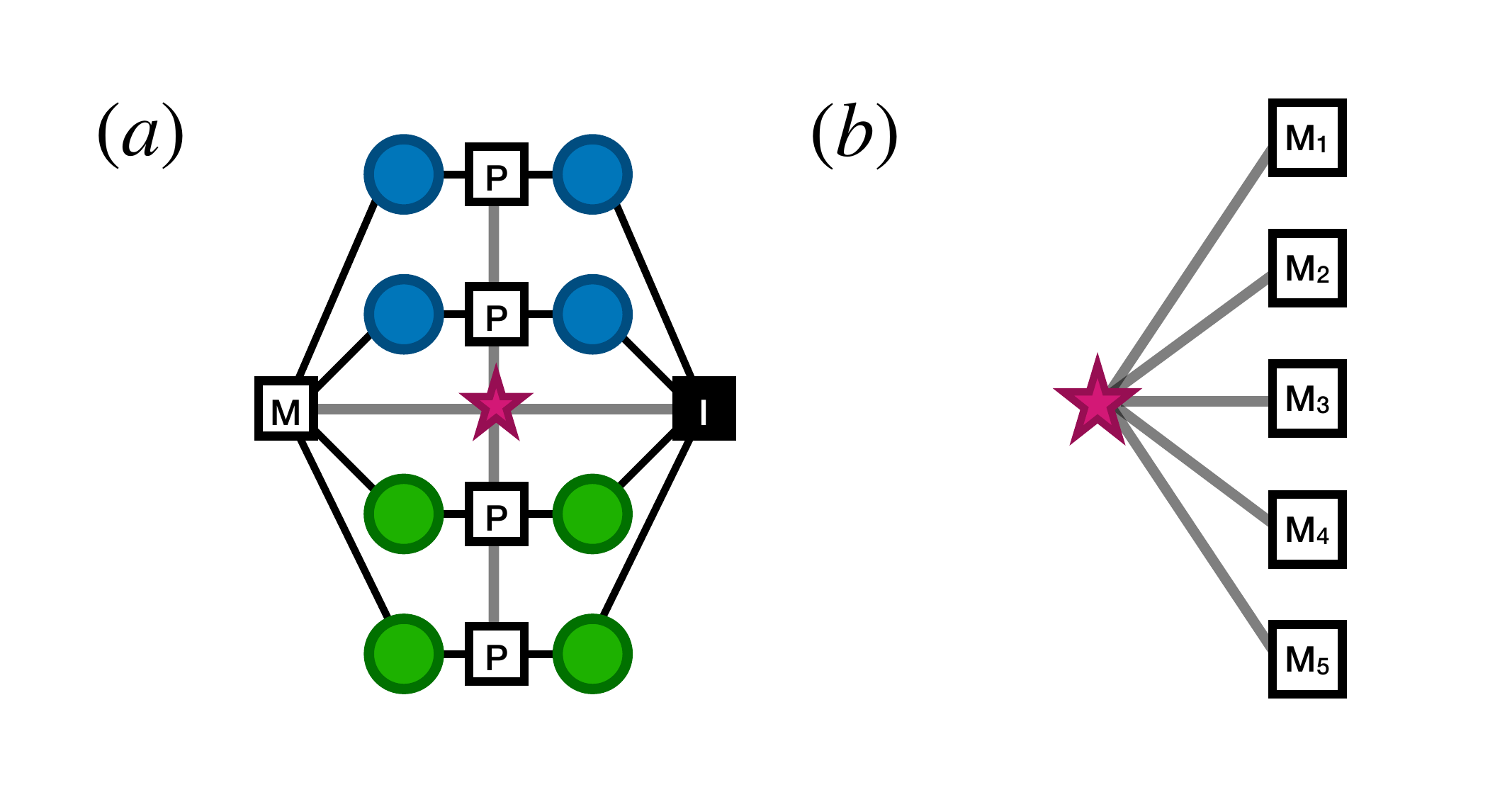} \caption{Local detector structures in the gauged SSC for the data-syndrome code example presented in Sec.~\ref{ssec:examples}. Squares labelled `M' are gauge fields associated with measurement slices whereas those labelled `P' are associated with propagators. \textbf{(a)} Detector formed by measuring input stabilizers, which comprise of the measurement and the propagators in its support. The input stabilizer, labelled `I' is assumed to be noiseless, and so does not participate in the detector. \textbf{(b)} Detector formed from the redundancy comprised of all measurement slices. In both panels, the magenta star marks the Wilson loop associated with the detector.
}\label{fig:data_syndrome_example} 
\end{figure}

\textbf{Data Syndrome codes -- } We start with the example of a data syndrome code \cite{Ashikhmin_IEEE_2020_data_syndrome_codes, Fujiwara_PRA_2014_OG_data_syndrome}. In particular, we consider the five qubit code, with stabilizers $X_i Z_{i+1} Z_{i+2} X_{i+3}$ for $i = 0, 1, 2, 3$ and indices taken modulo $5$  \cite{NC_2019}. In the symplectic representation, the check matrix is:
\begin{equation}\label{eq:H5qb0}
    H^{\rm 5 qb}_0 = \left(
\begin{array}{ccccc|ccccc}
1 & 0 & 0 & 1 & 0 & 0 & 1 & 1 & 0 & 0\\
0 & 1 & 0 & 0 & 1 & 0 & 0 & 1 & 1 & 0\\
1 & 0 & 1 & 0 & 0 & 0 & 0 & 0 & 1 & 1 \\
0 & 1 & 0 & 1 & 0 & 1 & 1 & 0 & 0 & 0 \\
\end{array}
\right)
\end{equation}
When these syndromes are measured noiselessly, the code is able to correct any single qubit error. This is not the case when there is the possibility of measurement errors. In the phenomenological model, this can be fixed by measuring, in addition to the four stabilizers, the linear combination of each stabilizer, which presents the simplest example of a data syndrome code \cite{Fujiwara_PRA_2014_OG_data_syndrome}. One may write the new check matrix as, 
\begin{equation}\label{eq:H5qb1}
    H^{\rm 5 qb}_1 = \left(
\begin{array}{ccccc|ccccc}
1 & 0 & 0 & 1 & 0 & 0 & 1 & 1 & 0 & 0\\
0 & 1 & 0 & 0 & 1 & 0 & 0 & 1 & 1 & 0\\
1 & 0 & 1 & 0 & 0 & 0 & 0 & 0 & 1 & 1 \\
0 & 1 & 0 & 1 & 0 & 1 & 0 & 0 & 0 & 1 \\
0 & 0 & 1 & 0 & 1 & 1 & 1 & 0 & 0 & 0 \\
\end{array}
\right)
\end{equation}
An alternative representation is to regard each of the syndromes as associated with a classical bit. One appends a column to the check matrix for each row of Eq.~\ref{eq:H5qb1}, and then eliminates the support on the first ten columns from the last row to obtain.
\begin{equation}
    H^{\rm 5 qb}_2 = \left(
\begin{array}{ccccc|ccccc|ccccc}
1 & 0 & 0 & 1 & 0 & 0 & 1 & 1 & 0 & 0 & 1 & 0 & 0 & 0 & 0\\
0 & 1 & 0 & 0 & 1 & 0 & 0 & 1 & 1 & 0 & 0 & 1 & 0 & 0 & 0\\
1 & 0 & 1 & 0 & 0 & 0 & 0 & 0 & 1 & 1 & 0 & 0 & 1 & 0 & 0\\
0 & 1 & 0 & 1 & 0 & 1 & 0 & 0 & 0 & 1 & 0 & 0 & 0 & 1 & 0\\
0 & 0 & 0 & 0 & 0 & 0 & 0 & 0 & 0 & 0 & 1 & 1 & 1 & 1 & 1\\
\end{array}
\right) =: \left( \begin{array}{c|c|c}
    H_X & H_Z  & H_D
\end{array} \right),
\end{equation}
which reveals the last block $H_D$ of the matrix as a classical code on the syndromes themselves, hence the name \textit{data syndrome code}. This particular code corrects all single qubit qubit or measurement errors that could occur on a pristine code state being measured.

Let us now construct the circuit describing this data syndrome code in the phenomenological model. The circuit will comprise two time steps $0, 1$, and a single Clifford operation,
\begin{equation*}
    C^{0.5} = (I, \mathcal{M}^{0.5} = \{X_i Z_{i+1} Z_{i+2} X_{i+3}: i = 0, 1, 2, 3, 4, 5\} ),
\end{equation*}
operating on $5$ qubits, with the spatial index taken modulo $5$. Instead of going through ECOs, we directly use the language of gauge and matter fields. These yield the checks,
\begin{gather*}
    \tau^{0.5}_{x, i} \leftrightarrow \sigma^1_{x, i} \sigma^0_{x, i} , \quad \tau^{0.5}_{z, i} \leftrightarrow \sigma^1_{z, i} \sigma^0_{z, i}, \quad
    \tau^{0.5}_{m, i} \leftrightarrow \sigma^1_{x, i} \sigma^1_{z, i+1} \sigma^1_{z, i+2} \sigma^1_{x, i+3}, \quad 
    i = 0, 1, 2, 3, 4
\end{gather*}
which are $X$-type and $Z$-type propagators and measurement slices respectively. In the bulk (i.e. before introducing boundaries), this has a single redundancy, given by 
\begin{equation}\label{eq:5qb_metacheck_red}
    R_4 := \{\tau^{0.5}_{m, i}: i = 0, 1, 2, 3, 4 \}.
\end{equation}
We can now gauge the code and introduce boundaries. For the output boundaries, we introduce the open boundary fields $\tau^{1.5}_{i, x/z}$. For the input boundary, we introduce a rough boundary, corresponding to an independent set of stabilizers, like Eq.~\ref{eq:H5qb0}, which appear as the Gauss law moves
\begin{equation*}
    \tau^{-0.5}_{I, i} \leftrightarrow \sigma^{0}_{i, x} \sigma^{0}_{i+1, z} \sigma^{0}_{i+2, z} \sigma^{0}_{i+3, x}, \quad i = 0, 1, 2, 3,
\end{equation*}
where we use the subscript $I$ for input. In addition to the redundancy corresponding to Eq.~\ref{eq:5qb_metacheck_red}, these introduce $4$ more redundancies in the gauge fields, 
\begin{equation*}
    R_i := \{\tau^{-0.5}_{I, i}, \tau^{0.5}_{m, i}, \tau^{0.5}_{x, i}, \tau^{0.5}_{z, i+1}, \tau^{0.5}_{z, i+2}, \tau^{0.5}_{x, i+3} \}.
\end{equation*}
Finally, we may eliminate the matter fields and gauge fix the noiseless degrees of freedom. We take $\tau^{-0.5}_{I, i}$ to be noiseless, since we are considering the case of an pristine code state subject to data error, and we take $\tau^{1.5}_{x/z, i}$ to be noiseless, since we don't care for errors after the end of the circuit. This also eliminates all the Gauss laws associated with $\sigma^{t}_i$ for integer $t$. In terms of Gauss laws, we are left with the Gauss laws associated with measurement, 
\begin{equation*}
    G^{\rm int}_{\mathcal{C}} = \left(
    \begin{array}{c|c|c}
          H_Z & H_X & \mathbf{0}_{5 \times 5} \\
    \end{array}
    \right)^T,
\end{equation*}
where the first five rows correspond to the $X$-type propagators, the next five correspond to the $Z$-type propagators, and the remaining correspond to the measurement slices. Finally, the redundancy matrix is 
\begin{equation}
    R_{\mathcal{C}} = \left(
    \begin{array}{c|c|c}
          H_X & H_Z & H_D \\
    \end{array}
    \right) = H^{\rm 5qb}_{2},
\end{equation}
where each row corresponds to a column of $G^{\rm int}_{\mathcal{C}}$. We see that the gauged SSC cleanly reproduces the check matrix of a data syndrome code. The detectors of this data syndrome code are depicted in Fig.~\ref{fig:data_syndrome_example}.

\begin{figure}[ht!]
\centering 
\includegraphics[width=0.99\linewidth]{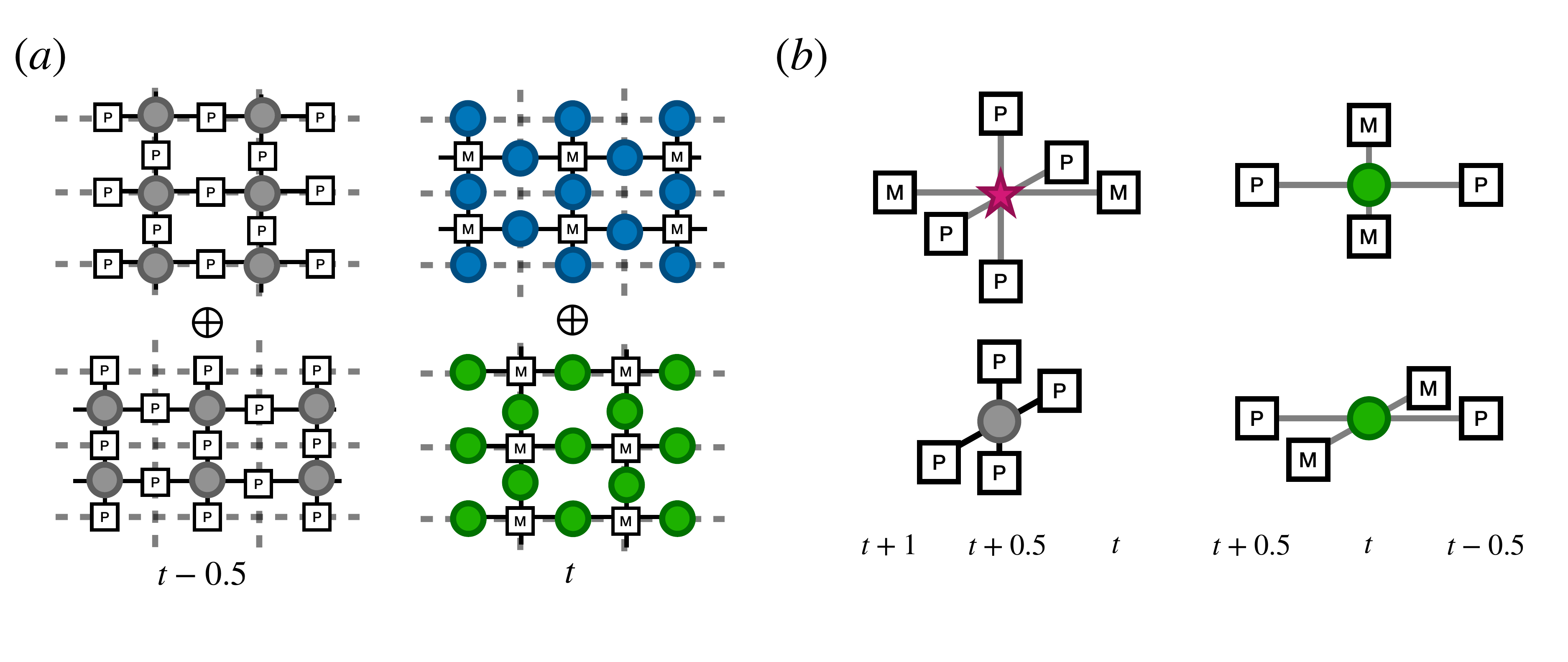} \caption{ Gauged SSC for the repeated measurement of a surface code with phenomenological errors. \textbf{(a)}  Schematic depicting the gauge fields in integer and half-integer time-slices of a circuit performing repeated measurement on a surface code. The gray dashed lines depict the underlying lattice, and the propagators connect the blue (green) circles of adjacent integer time slices. \textbf{(b)} Local structures in the resulting gauge theory. Upper-left: The detectors of the surface code circuit are formed by redundancies between measurement slices on integer time-slices and propagators on half-integer time-slices. This is associated with cubes (vertices) on the 3D lattice when the `M' squares are associated with plaquette $Z$ (vertex $X$) type measurements. Bottom-left: The Gauss laws associated with measurement type spacetime locations. These are associated with faces centered on the half-integer time-steps. Right: These are the Gauss laws associated the $X$-type spacetime locations, which are associated with faces of the 3D lattice centered on integer time-steps. Similar Gauss laws hold for the $Z$-type spacetime locations. 
}\label{fig:toric_example} 
\end{figure}

\textbf{Repeated measurement of CSS codes (phenomenological) -- } We consider a circuit where all stabilizers of some CSS code is measured in every time-step\footnote{The discussion here holds for non-CSS codes as well, but the specific discussion of CSS codes will be a convenient reference for Sec.~\ref{sec:applications}.}. Let the check matrices be $H_X, H_Z$. Say we have $m_X$ and $m_Z$ stabilizers of each type. The measurement slices have the form
\begin{equation*}
    \begin{aligned}
        \eta^{t}\left( \prod_{i = 1}^n X_i^{(H_X)_{ji}} \right), \qquad \eta^{t}\left( \prod_{i = 1}^n Z_i^{(H_Z)_{ji}} \right),
    \end{aligned}
\end{equation*}
for each $j = 1, ..., m_{X(Z)}$ enumerating the number of stabilizers and $t = 1, ..., T$. 

An error $Z_i$ anticommutes with the $j$-th $X$-type stabilizer if $(H_X)_{ji} = 1$, and similarly for $X_i$ with $(H_Z)_{ji} = 1$. Hence, we have elementary propagators of the form
\begin{equation*}
    \begin{aligned}
        \eta^{t+1} \left( Z_i \right) \eta^{t+0.5} \left( \prod_{j = 1}^{m_X} Z_j^{(H_X)_{ji}} \right) \eta^{t} \left( Z_i \right), \qquad \eta^{t+1} \left( X_i \right) \eta^{t+0.5} \left( \prod_{j = 1}^{m_Z} Z_j^{(H_Z)_{ji}}  \right) \eta^{t} \left( X_i \right),
    \end{aligned}
\end{equation*}
for $i = 1, ... , n$. We will impose periodic boundary conditions, corresponding to an entanglement fidelity experiment, so that in the above $t = 0, ..., T-1$, and the $T$-th time-slice identified with $0$. In the symplectic representation, we may write
\begin{equation*}
    G_{\mathcal{C}} = \left( \begin{array}{c|c|c}
    I_T \otimes H_X & 0  & 0 \\
    0 & I_T \otimes H_Z  & 0 \\ 
    R \otimes I_n & 0 & I_T \otimes H_Z^T \\
    0 & R \otimes I_n & I_T \otimes H_X^T \\
\end{array} \right),
\end{equation*}
where the first two rows are the measurement slices, and the last two rows are the elementary propagators, with $R$ being a $T \times n$ matrix defined by $R_{ij} = \delta_{i, j} + \delta_{i+1, j}$, with the delta functions defined mod $T$. Having picked periodic boundary conditions, we have the Gauss laws given by $G_{\mathcal{C}}^T$. Assuming the measured stabilizers are independent, the redundancies are given by adjacent measurement slices and their propagators, so 
\begin{equation*}
    R_{\mathcal{C}} = \left( \begin{array}{c|c|c|c}
    I_{m_X} \otimes R & 0  & I_T \otimes H_X & 0 \\
    0 & I_{m_Z} \otimes R & 0 & I_T \otimes H_Z 
\end{array} \right),
\end{equation*}
where we have excluded the global redundancies of the form $(0 | 0 | I_T \otimes L_X | I_T \otimes L_Z)$, where $L_X, L_Z$ are logical operators of the code. 

The matrices $G_{\mathcal{C}}, R_{\mathcal{C}}$ contain exactly the same blocks that turn up in the fault complexes used to describe measurement-based analogs of repeated measurement \cite{Hillmann_arXiv_2024_fault_complexes}; we comment on this in Sec.~\ref{ssec:SSC_foliated_computation}.  The matrix $R_{\mathcal{C}}$ has the explicit form of a data syndrome code, which encodes the idea that repeated measurement is the simplest data syndrome code \cite{Fujiwara_PRA_2014_OG_data_syndrome, Ashikhmin_IEEE_2020_data_syndrome_codes}.

\textbf{The surface code (phenomenological) -- } We now specialize to repeated measurements of the surface code \cite{Bravyi_arXiv_1998_surface_code}. Consider the gauged SSC in the bulk. The gauge fields associated with single layers of the circuit are depicted in Fig.~\ref{fig:toric_example}a. In integer layers, we have matter fields of $X$-type and $Z$-type occupying the edges of the underlying lattice. We say this is split into the $X$-type and $Z$-type lattice. On the $Z$-type ($X$-type) lattice, we have measurement gauge fields $\tau_{m, \cdot}$ on the plaquettes (vertices), representing the measurement of plaquette (vertex) operators. The lattice on the half-integer locations can be obtained by swapping gauge fields with matter fields, and $X$ and $Z$ lattices, as shown on the left side of Fig.~\ref{fig:toric_example}a. The gauge fields on the half-integer layers are connected to the matter fields on the integer locations. Putting this all together, we can extract the Wilson loops and Gauss laws as in Fig.~\ref{fig:toric_example}b. We see that the gauged SSC for repeated measurements of a surface code yields two copies of a 3D surface code, albeit interpreted as a gauge theory rather than a quantum code \cite{Wen_2017_book}. This is consistent with the well-known statistical mechanical model of this process \cite{Dennis_JMP_2002_topological_quantum_memory}. 

\begin{figure}[ht!]
\centering 
\includegraphics[width=0.99\linewidth]{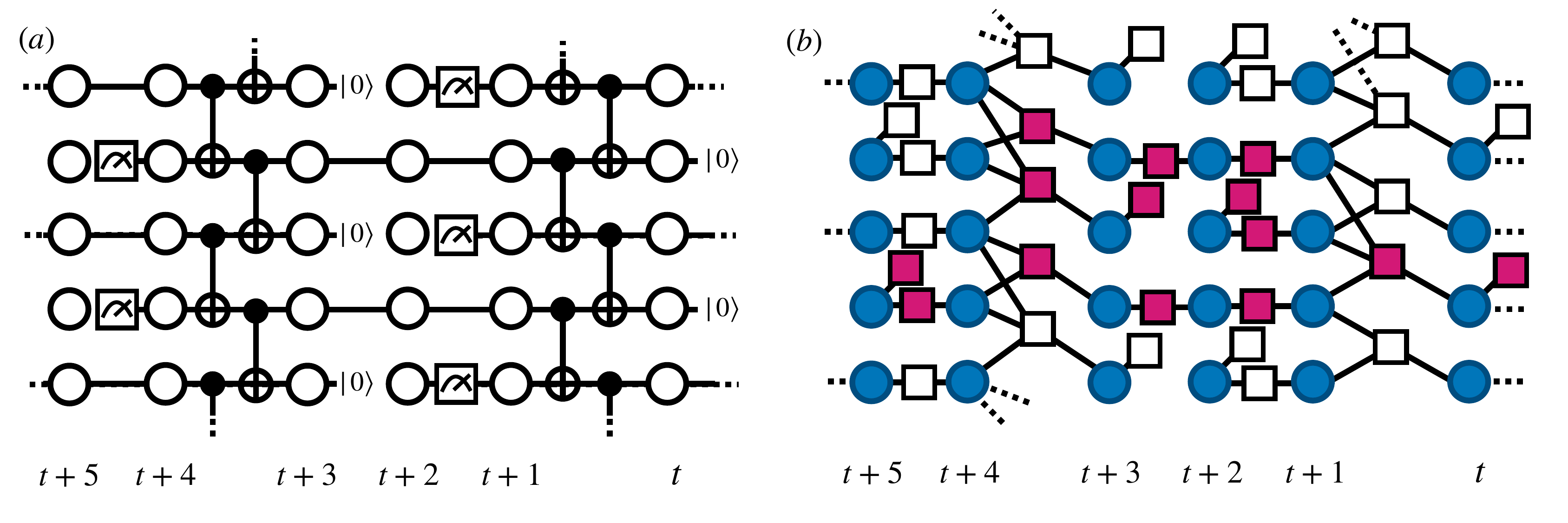} \caption{The gauged SSC for the walking repetition code circuit introduced in Ref.~\cite{McEwen_quantum_2023_time_dynamics}. In particular, a unit of this circuit, comprising three time steps, is the `stepping down' circuit depicted in App.~B of Ref.~\cite{McEwen_quantum_2023_time_dynamics}. \textbf{(a)} Circuit and discretization of two sets of measurements of the stepping down circuit. In each we prepare half the qubits in the $| 0 \rangle$ state and measure the rest in the $Z$ eigenbasis. The unitary gates are a ladder of cNOTs. We discretize this circuit by grouping all the cNOTs together, and separating measurements and state preparations into their individual time-slices. Dotted lines depict where this circuit would connect to other parts of the circuit. \textbf{(b)} The gauged SSC for this circuit in the $Z$ sector. The blue circles are the $Z$-type matter fields, whereas the squares are the gauge fields. We highlight in magenta all the gauge fields participating in a particular local Wilson loop of this circuit. This can be verified visually by noting that every matter field is connected to an even number of magenta squares, indicating a redundancy in the ECOs. This provides a simple example where the gauged SSC captures detectors beyond the repeated-measurement paradigm.
}\label{fig:walking_circuit_example} 
\end{figure}

\textbf{Repetition redux: Stepping circuits -- } As a final example, we revisit a circuit that performs syndrome extraction on the repetition code in a way that is distinct from repeated measurement. In particular, we consider the `stepping down' circuits introduced in Ref.~\cite{McEwen_quantum_2023_time_dynamics}, where the ancilla and data qubits are swapped in every time-step. In Fig.~\ref{fig:walking_circuit_example}a, we depict the circuit for two units of the stepping down circuit, and in Fig.~\ref{fig:walking_circuit_example}b gauged SSC associated with the stepping circuit for the repetition code. In this circuit, the detectors are no longer just products of the same check measured at consecutive time-steps. Nevertheless, they are still captured in the gauged SSC as redundancies, or Wilson loops. An example of such a detector is depicted in Fig.~\ref{fig:walking_circuit_example}b, and comprises all the magenta-highlighted gauge fields. This illustrates that the gauged SSC applies beyond the repeated-measurement setting, and naturally accommodates more general spacetime detector geometries. Extracting the exact geometrical structure of various circuits to compare to known gauge theories could be an interesting endeavor, which we leave for future work.





\section{Applications of the gauged SSC beyond QEC}\label{sec:applications}

We now turn to applications of the gauged SSC beyond fault distances. In Sec.~\ref{ssec:SSC_foliated_computation}, we study the relationship between the spacetime code and foliation \cite{Bolt_PRL_2016_foliated_codes, Brown_PRR_2020_universal_MBQC}, which takes a particularly simple form in the gauged SSC, and allows us to make contact with the gauge theory of MBQC. In Sec.~\ref{ssec:mixed_state_order}, I use the connection between foliation and the spacetime code to comment on the connection between the gauged SSC and mixed state topological order \cite{Negari_arXiv_2025_FT_mixed_state_phases, Sang_PRL_2025_markov_length, Sohal_PRXQ_2025_mixed_state_subsystem}. In particular, we will expand on the notion of a particular set of mixed state phases as a classical memory, and show that it can also be described by the gauge theory of a Clifford circuit. Finally, in Sec.~\ref{ssec:learning_Pauli_noise}, I draw connections to learning theory. Rather surprisingly, we find that the theory of learning Pauli noise \cite{Chen_PRXQ_self_consistent_learning} can be thought of as providing a parent theory for the gauged SSC, in the sense that for any circuit, the detectors, which are the observables, correspond exactly to learnable degrees of freedom, which are the observables in the theory of learning Pauli noise. This provides a little bit of a learning theoretic foundation for why, as we have claimed throughout this work, one might say that detectors are the true observables of a QEC circuit, and also provides some insight into why the spacetime code recently found use in several problems in learning theory \cite{zheng_arXiv_2026_logical_noise, xiao_arXiv_2026_insitu_benchmarking}.

\subsection{The gauged SSC is a foliated computation}\label{ssec:SSC_foliated_computation}

\subsubsection{A first foliated computation}\label{sssec:first_foliation}

We now turn to the idea of foliation. In contrast to Refs.~\cite{Bolt_PRL_2016_foliated_codes, Brown_PRR_2020_universal_MBQC} we focus on circuits that contain both mid-circuit measurements and unitary gates, and we will argue that the gauged SSC \textit{contains} a \textit{foliated computation}.  It is easiest to begin with a concrete example; let us re-consider the gauged SSC depicted in Figs.~\ref{fig:small_full_example}, which is derived from the circuit in Fig.~\ref{fig:detector_example} with open boundary conditions.

\begin{figure}[ht!]
\centering 
\includegraphics[width=0.5\linewidth]{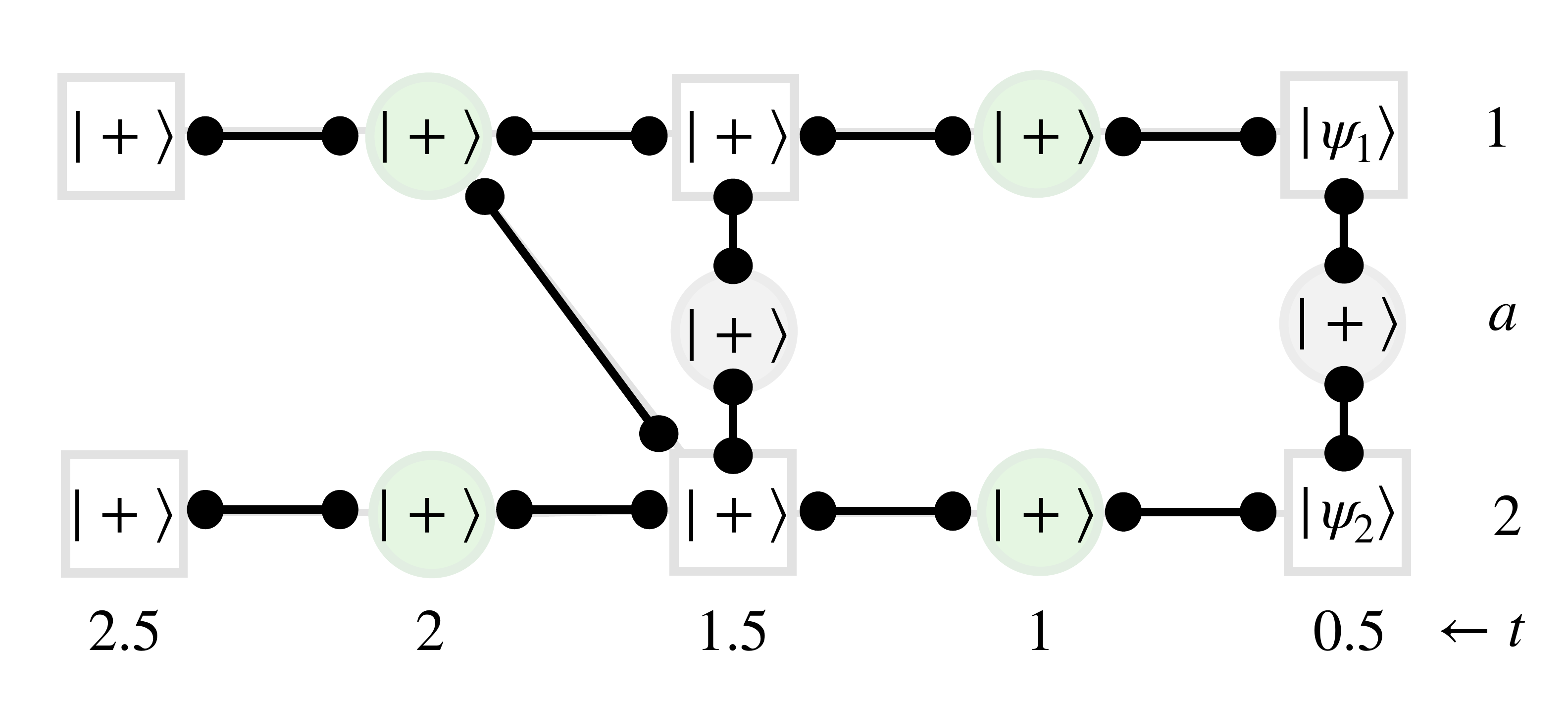} \caption{The resource state for the foliated computation constructed from the circuit in Fig.~\ref{fig:detector_example}. The $X$-type part of the gauged SSC (detectors excluded) is displayed in the background. Regarding this as a graph, we introduce qubits in the $| + \rangle$ state for every node except the input boundary. To obtain the resource state, we carry out $cZ$ gates along every edge. Horizontal and vertical axes serve as spatial and temporal labels for the qubits. }\label{fig:small_foliation_example} 
\end{figure}

From this gauged SSC, we will generate a resource state as illustrated in Fig.~\ref{fig:small_foliation_example}. We refer to this construction as a \textit{foliated computation}. We introduce a qubit for each gauge and each matter field in the $X$-type part of the gauged SSC. Every qubit is intitialized in the $| + \rangle $ state, except the ones associated with the input gauge fields, which can be initialized in arbitrary states $| \psi_1 \rangle, | \psi_2 \rangle$ (corresponding to what one would input into the circuit Fig.~\ref{fig:detector_example}. Then, we perform $cZ$ gates along all the edges in the $X$-type part of the gauged SSC to prepare a \textit{resource state}, which is very nearly a graph state \cite{Raussendorf_PRA_2005_cluster_states}. Finally, performing performing $X$ basis measurements on every layer of qubits except the last one carries out a \textit{foliated computation}. We claim that this carries out the same circuit up to a known software correction, provided we keep the outcomes of the measurements. We call this a \textit{foliated computation}. It is immediate to see that in the absence of unitaries or measurements, this reduces to the ``type-I foliated'' wire of Ref.~\cite{Brown_PRR_2020_universal_MBQC}. We will denote Paulis on the resource state using a time superscript and a spatial subscript, eg. $X^t_i$; see the axes in Fig.~\ref{fig:small_foliation_example}.

To show that the foliated computation is equivalent to the original circuit, we check the logical action, the imposed stabilizers, and the detector constraint. First, we consider the states $| \psi_1 \rangle, | \psi_2 \rangle$, corresponding to the input into the original circuit. It is sufficient to track $X_L^{(1)} = X_1 X_2, X_L^{(2)} = X_2, Z_L^{(1)} = Z_1, Z_L^{(2)} = Z_1 Z_2$ (choosing some complete set of logical operators with the appropriate commutation relations). The first measurement of $Z_1 Z_2$, suppose this has outcome $o(Z_1 Z_2)$, we lose the logical values of $X_L^{1}$, and fix the $Z_L^{(2)} \rightarrow (-1)^{o(Z_1 Z_2)} Z_1 Z_2$ as a stabilizer, but leave $X_L^{(1)}, Z_L^{(1)}$ alone. The next operation, which is $cX_{2 \rightarrow 1}$  transforms the stabilizer like $(-1)^{o(Z_1 Z_2)} Z_1 Z_2 \rightarrow (-1)^{o(Z_1 Z_2)} Z_1$, and the logicals transform like $X_L^{(1)} = X_1 X_2 \rightarrow X_2, Z_L^{(1)} = Z_1 \rightarrow Z_1 Z_2$. The final move is to measure $Z_1$ Since this is a stabilizer, the outcome is deterministic and equal to $o(Z_1 Z_2)$. Nothing changes. Hence, we see that the effect of this circuit is:
\begin{equation*}
    \begin{aligned}
        X_L &= X_1 X_2 \rightarrow X_2, \\
        Z_L &= Z_1 \rightarrow Z_1 Z_2, \\
        {\rm impose: } & \mathcal{S} = \langle (-1)^{o(Z_1 Z_2)} Z_1 \rangle,
    \end{aligned}
\end{equation*}
subject to the constraint that $o(Z_1 Z_2) = o(Z_1)$.

Logical action: First, in preparing the state (i.e. applying cZs), we transform 
\begin{gather*}
        X_L = X^{0.5}_1 X^{0.5}_2 \rightarrow Z^{1}_1 Z^{1}_2 X^{0.5}_1 X^{0.5}_2, \\
        Z_L = Z_1 \rightarrow Z_1.
\end{gather*}

Next, we observe that the following two operators are stabilizers
\begin{equation*}
    \begin{aligned}
        & X^{2.5}_2 \cdot X^{1.5}_1 X^{1.5}_2 \cdot Z^1_1 Z^1_2 \\
        &  Z^{2.5}_2 X^2_2 \cdot Z^{2.5}_1 X^2_1 \cdot X^1_1 Z^{0.5}_1.
    \end{aligned}
\end{equation*}
This tells us that, up to stabilizers of the foliated computation state, we have 
\begin{equation*}
    \begin{aligned}
        X_L &\sim X^{2.5}_2 \cdot X_1^{1.5} X_2^{1.5} X^{0.5}_1 X^{0.5}_2 \rightarrow (-1)^{o(X_1^{1.5}) + o(X_2^{1.5}) + o(X_1^{0.5}) + o(X_2^{0.5})} X^{2.5}_2\\
        Z_L &\sim Z^2_2 X^2_2 \cdot Z^2_1 X^2_1 \cdot X^1_1 \rightarrow (-1)^{o(X^2_2)+o(X^2_1) + o(X^1_1)} Z^2_2 Z^2_1,
    \end{aligned}
\end{equation*}
where the arrow indicates the transformation of the operator after the single-qubit $X$-measurements, so this procedure has accomplished exactly the logical action of the circuit, up to teleporting the state from one location to another.

Next, we should check that the measurement outcomes of the original circuit are available. To see this, note that $X^{0.5}_a Z^{0.5}_1 Z^{0.5}_2$ is a stabilizer of the resource state, so measuring $X^{0.5}_a$ is equivalent to measuring $Z^{0.5}_1 Z^{0.5}_2$ as a stabilizer. Unfortunately, when we created the resource state, we had to shuffle the ordering of unitaries and $X$ vs $Z$-type measurements a little, so at first glance, the second ancilla appears to be measuring $Z_1 Z_2$ as well. However, we note that we have 
\begin{equation*}
    Z_1^{2.5} X_1^2 X^{1.5}_a 
\end{equation*}
 as a stabilizer, so in fact we are effectively correctly measuring $Z_1$ on the final state! It is apparent that these things are equivalent in the noiseless case. Later, we see that surprisingly, up to a refactoring of error identification, we can make the correct propagation relations hold in the noisy case too. Having shown that the ancillas do carry out the expected measurements, this also implies that we have imposed the correct stabilizer, and that we have the correct detector. However, we should note that due to the fact that we had to impose some temporal order between $X, Z$ in order to foliate the state, the connectivity of the measurement had to change a little to maintain the same detecting regions. A simple way to avoid this problem is to simply discretize so that we only ever have either unitaries or measurements in a single time-slice. Finally, we should also point out that the fact that these two measurements have a fixed parity can be seen from the resource state stabilizer 
 \begin{equation*}
    X_a^{1.5} X^1_1 X^1_2 X_a^{0.5},
 \end{equation*}
 which tells us that $X_a^{0.5}, X_a^{1.5}$ have fixed parity conditioned on $X^1_1, X^1_2$ which are measured.

\subsubsection{Yet another redundancy in the gauged SSC}\label{sssec:banach-tarski-prop}

Before we get into the form of the foliated computation for more general circuits, we quickly ask: Why does this even work? Earlier, we have argued that the ECOs and in turn the gauged SSC essentially form an alternate representation of the circuit. In foliating the circuit, we have discarded half of the gauged SSC; yet we have found we are still able to reproduce the circuit using the remaining half.

The elementary propagators are independent by construction, and so never form a redundancy in the sense of Def.~\ref{def:op_redundancy}. However, Clifford operations must conserve commutation relations, and specifying $n$ commuting Paulis allows one to uniquely solve for their destabilizers -- these are Paulis which uniquely anti-commute with exactly one of the original set of Paulis. Together, a fully commuting set of Paulis and their destabilizers forms a generating set for the Pauli group \cite{Aaronson_2004_PRA_stab_sim}. In other words, solving for $U^{t>} X_i (U^{t>})^{\dag}$ for every $i = 1, ..., n$ already contains enough information for us to obtain $U^{t>} Z_i (U^{t>})^{\dag}$. With regards to measurements, the $X$-type part contains $Z$-type measurements in their matter field form and vice versa. Hence, in gauged SSCs where the $X$-type and $Z$-type parts separate, either part already implies the other. Henceforth we refer to these as \textit{CSS-like} circuits.
 
Finally, we note that we can make this connection between the $X$ and $Z$ parts of a CSS-like circuit explicit in the ECOs by a simple basis change. We reparameterize $G_{\mathcal{C}}$ as follows: We keep measurement slices the same, and all $X$-type elementary propagators. However, we replace the propagators of $Z$ operators with their backwards versions:
\begin{equation}\label{eq:Z_back_prop}
    (g_{{\rm prop}}^Z)_n^{t>} := \eta^{t+1}(Z_n ) \eta^{t+0.5} \left( \prod_{i \in A} Z_i \right) \eta^{t}((U^{t>})^{\dag} Z_n U^{t>}),
\end{equation}
where $A$ is the set of measurements in time $t+0.5$ that anti-commute with $Z_n$. This set of ECOs preserve the absence of propagator redundancies, we call it $G_{\mathcal{C}}^{\rm cl}$, and is related to the previous parameterization via a change of basis. The superscript ${\rm cl}$ is short for cluster, for reasons that will become clear shortly. The error identification is modified slightly -- we now identify the gauge field $\tau^{t+0.5}_{z, n}$ with the error $(U^{t>})^{\dag} Z_n U^{t>}$ occurring before the gates in that layer. In this basis, \textit{gauging $G_{\mathcal{C}}^{\rm cl}$ for a CSS-type circuit results in the direct sum of two gauge theories related to each other in the bulk by swapping gauge and matter fields.} However, we note that when gauging an unfortunate consequence is that mapping $X$ errors to gauge fields no longer conserves the error weight. 


\subsubsection{From SSC to resource state}

We now generalize the example in Sec.~\ref{sssec:first_foliation} to arbitrary CSS-like circuits. The construction of the \textit{foliated computation} is straightforward: For such a circuit, we construct the gauged SSC, but keep only the $X$-type part of it\footnote{We could work with the $Z$-type part, but this allows us to keep to conventions in the literature.}. Treating it as a graph $\Gamma^X = (V, E)$, with fields as nodes, we introduce a qubit in the $| + \rangle$ state for every node $v \in V$ except the input boundary nodes, which are allowed to be in arbitrary states. We then carry out $cZ$s along each edge $e \in E$. We refer to this as the \textit{resource state}. Note that this would be exactly a cluster state if we had also started with the input boundary nodes in the $| + \rangle$ state \cite{Raussendorf_PRA_2005_cluster_states}. The \textit{foliated computation} then refers to measuring every node except the output boundary fields in the $X$ basis, also referred to as a `measurement pattern' \cite{Pesah_arXiv_2025_FT_transformations, Brown_PRR_2020_universal_MBQC}.

In the noiseless case, the foliated computation is exactly equivalent to the circuit we started with, and represents a measurement-based compilation of the circuit. In other words, for the same input state, the foliated computation carries out exactly the same computation as the circuit, and extracts measurement results with the same statistics. This is true up to a Pauli frame, i.e. applying signs to different Pauli operators based on the outcomes of the various $X$-basis measurements. .


To show that this is true, we will explicitly derive the action of the foliated computation in App.~\ref{app_sec:deferred_proofs_applications} The derivation proceeds by showing that for each pair of layers of qubits, the correct action is carried out on a set of operators which can be identified with the Paulis of the circuit. One can quickly check that for repeated measurement syndrome extraction circuits, this reproduces the foliated codes of Ref.~\cite{Bolt_PRL_2016_foliated_codes}. More generally, for measurement only circuits, this produces the same compilation to MBQC as Ref.~\cite{Brown_PRR_2020_universal_MBQC}, justifying the name `foliated computation'\footnote{Foliated codes and computations may also include a final layer of measurements or an initial layer of state preparation and so on; our description is easily generalized to the various boundary conditions discussed in Sec.~\ref{ssec:boundaries}. We take it for granted that this is obvious.}. The main intuition here is that the horizontal (along the time axis) connectivity of the foliated computation ensures that the $X$-type operators of the circuit propagate in the right way, which is sufficient, due to their relationship as stabilizers and destabilizers, to also ensure that the $Z$-type operators propagate in the same way. 

We end this section with a few remarks. First, there are many well-established ways to compile circuits, in particular Clifford circuits, into MBQCs, so the foliated computation does not offer anything new in that regard. For instance, the seminal work Ref.~\cite{Raussendorf_PRL_2007_FTQC_2D} simulates circuit wires with 1D graph states, and uses an intermediate layer of nodes to mediate correlations between two wires. For Clifford gates, $X$-type measurements are used to propagate information and $Z$-type measurements are used to cut off correlations in the right patterns to carry out the Clifford gate.

Second, the scheme in Ref.~\cite{Brown_PRR_2020_universal_MBQC} is universal and also takes care of the compilation of non-CSS type circuits. While the compilation of CSS-type circuits is particularly elegant from the gauged SSC, it is not too difficult to generalize to non-CSS-like circuits. However, I have not explored any means that is so simple as the preceding one that I would dare claim the gauged SSC for a non-CSS-like circuit `is a foliated computation' in the same way. On the other hand, taking non-Clifford-ness into account is a little more complicated. In MBQC, non-Clifford-ness requires adaptivity, which is not (at the moment) naturally built into the gauged SSC.

The foliated computation is one possible compilation of a Clifford circuit into an MBQC. Ref.~\cite{Pesah_arXiv_2025_FT_transformations} has also proven another such mapping for any measurement-based quantum computation (MBQC), including foliated codes. For a given MBQC, they begin by writing down the full circuit that prepares the state and then carries out the requisite measurements. They then show that one can get from the SSC obtained from this circuit to the SSC that one gets from a logically equivalent circuit via a series of `fault-tolerant moves', which roughly speaking deform the SSC while preserving its distance.

In the next section, we will argue that foliated computations inherit, to some extent, the elements of fault tolerance from the gauged SSC. In contrast to Ref.~\cite{Pesah_arXiv_2025_FT_transformations}, ours is not a fully distance-preserving map, as it will not account for errors that may arise in preparing the resource state. However, by giving up a true distance-preserving mapping, we obtain a much shorter route to foliated computation that is much more physically evocative, which makes it easier for us to characterize the topology \cite{Hillmann_arXiv_2024_fault_complexes, Newman_Quantum_2020_generating_FT_cluster_states} and eventually the inherent mixed state order \cite{Negari_arXiv_2025_FT_mixed_state_phases} and classical memory in the resource state.

\subsubsection{Elements of fault tolerance}

We have now mapped a circuit to a foliated computation and argued that it carries out the correct action in the noiseless case. We now examine how the other elements of fault tolerance carry through. We assume that the resource state preparation is noiseless, so that the only errors that can occur are readout errors when measuring the qubits of the resource state. These are equivalently characterized as $Z$ errors on the resource state, since the qubits are measured in the $X$ basis. The following general statements apply to all bulk fields, i.e. other than the input/output boundary fields, which should be dealt with separately. 

\textit{Error identification -- } In defining the resource state, we have identified all gauge and matter fields in the $X$-type part of the gauged SSC with qubits of the resource state. In the discussion after Eq.~\eqref{eq:Z_back_prop}, we have pointed out that the $X$-type matter fields here can be identified with $Z$-type gauge fields (which we recall are now identified with errors of the form $(U^{t>})^{\dag} Z_n U^{t>}$). Hence, every qubit of the resource state can be identified with a gauge field. It is tedious but straightforward to check that this is consistent in the sense that a readout error on that qubit in fact yields the same error (has the same bare effect) that the gauge field is identified with. The input boundary is identified with errors in the same way, whereas we do not identify the output boundary is not identified with any error, since they are not measured. This captures all gauge fields in the gauged SSC except the output boundary fields. We can either regard this as a gauge-fixed version of the gauged SSC or identify output errors with errors on the final state after the foliated computation.

\textit{Gauss laws -- } Similar to how we can identify each qubit (other than the input boundary) of the resource state with a gauge field, we can alternatively identify each of them with a matter field. Each matter field is associated with a Gauss law, and so the Gauss laws of the gauged SSC map to stabilizers of the form of 
\begin{equation*}
    X_v \prod_{v' \in \mathcal{N}(v)} Z_{v'},
\end{equation*}
where $v$ is the vertex, which when regarded as a matter field, is associated with the Gauss law, and $\mathcal{N}(v)$ are its neighbors when the resource state is regarded as a graph. These are of course the standard stabilizers of the cluster state. The output boundary fields, which are regarded as $Z$-type matter fields, are identified with Gauss laws in the same way. The input boundary fields are not associated with Gauss laws. Similarly to error identification, the Gauss laws on boundaries can be fixed by either regard this as a gauge-fixed version of the gauged SSC or appending additional qubits to the input layer for the sake of consistency.

\textit{Detectors -- } The detectors of the gauged SSC map to pure $X$-stabilizers of the resource state. This is consistent with the usual identification of detector cells in cluster states, see for instance Refs.~\cite{Bergamaschi_ITCS_2025_single_shot, Newman_Quantum_2020_generating_FT_cluster_states} for the discussion of such stabilizers as `syndromes of syndromes'. 

Hence, the elements of fault tolerance are preserved in the sense that, up to some fiddling with the boundaries, there is a one-to-one mapping to these same quantities if supplied with a noiselessly prepared resource state. This tells us that the fault tolerance of the circuit and the resource state are only qualitatively related (albeit in a quantifiable way), but allows us to make contact with some concepts in mixed state topological order.

\subsubsection{A gauge theory in common}

We have now associated foliation with the gauged SSC, endowing one form of MBQC with a gauge theory, which we have represented as a chain complex in Eq.~\ref{eq:gauge_complex_2}. It turns out the foliation of the gauged SSC can be understood through the lens of yet another chain complex: 
\begin{equation}\label{eq:newman_complex}
    C_3 \xrightarrow{R_{\mathcal{C}}^Z} C_2 \xrightarrow{G_{\mathcal{C}}^X} C_1 \xrightarrow{R_{\mathcal{C}}^X} C_0,
\end{equation}
where $C_2$ are the $X$-type matter fields, which can be identified with $Z$-type gauge fields, $C_1$ are the $X$-type gauge fields, which can be identified with the $Z$-type matter fields. $G_{\mathcal{C}}^X$ is $G_{\mathcal{C}}$ with the $X$-type operators only, which, following the discussion of Sec.~\ref{sssec:banach-tarski-prop}, is the transpose of $G_{\mathcal{C}}^Z$ up to the boundary fields, and $R_{\mathcal{C}}^{X/Z}$ are the redundancies in $X, Z$, such that $R_{\mathcal{C}} = R_{\mathcal{C}}^X \oplus R_{\mathcal{C}}^Z$. 

Wonderfully, Eq.~\eqref{eq:newman_complex} is exactly the chain complex used to describe fault-tolerant cluster states in Ref.~\cite{Newman_Quantum_2020_generating_FT_cluster_states}, up to the choice of inclusion of global redundancies\footnote{A graph state created from a product of $| + \rangle$ states has the global redundancies as an $X$-type stabilizer, but a foliated computation, which leaves the possibility of some unspecified input state, may not.}! Furthermore, gauging repeated measurement circuits yield exactly the fault complexes with the product structure of Ref.~\cite{Hillmann_arXiv_2024_fault_complexes}, similar to our discussion in Sec.~\ref{ssec:examples}. Hence, the gauged SSC provides a simple unifying way to view spacetime codes and foliated computation, exactly analogous to the connection between the toric code and cluster states \cite{Raussendorf_NJP_2007_topological_cluster_states}.

In this sense, the gauged SSC can be viewed as a gauge theory of MBQC of Clifford circuits in arbitrary dimensions. Such a gauge theory has been developed in 1D for arbitrary single-qubit unitaries in Ref.~\cite{Wong_2024_quantum_gauge_theory_MBQC}. We comment on two aspects of the gauge theory presented in Ref.~\cite{Wong_2024_quantum_gauge_theory_MBQC}. First, their gauge theory utilizes periodic boundary conditions. Our work sheds light on the operational interpretation of periodic boundary conditions (which correspond to entanglement fidelity experiments, see Sec.~\ref{ssec:boundaries}), and allows us to understand what other kinds of boundary conditions would look like. Second, as a consequence of periodic boundary conditions, their gauge theory has global Wilson loops corresponding to the global redundancies we observe under periodic boundary conditions. Since they work in 1D, local detectors are not possible without destructively measuring the content of the computation, but our work shows that richer gauge-invariant structures arise beyond 1D. All that said, it seems that the gauged SSC could be a good starting point to extend the program of Ref.~\cite{Wong_2024_quantum_gauge_theory_MBQC} beyond 1D MBQC with periodic boundary conditions.

\subsection{Mixed state order}\label{ssec:mixed_state_order}

We now observe that the foliation of a circuit via the gauged SSC captures as a special case the mapping between circuit and noisy resource state developed in Ref.~\cite{Negari_arXiv_2025_FT_mixed_state_phases}. As a consequence, this opens up the study of the intersection between fault tolerance and mixed state phases for circuits that may utilize schemes beyond repeated measurements \cite{Fujiwara_PRA_2014_OG_data_syndrome, Ashikhmin_IEEE_2020_data_syndrome_codes, Chao_PRL_2018_Flag_I, Chao_PRXQ_2020_flag_II}, or circuit level models that include unitary gates, such as the stepping circuits \cite{McEwen_quantum_2023_time_dynamics} depicted in Fig.~\ref{fig:walking_circuit_example}. We now revisit the protocol introduced in Ref.~\cite{Negari_arXiv_2025_FT_mixed_state_phases}, and ask how the structure of mixed state topological order shows up in the gauged SSC.

Ref.~\cite{Negari_arXiv_2025_FT_mixed_state_phases} focuses on CSS codes that are repeatedly measured. These circuits, discussed in Sec.~\ref{ssec:examples}, are mapped to a resource state, identically to the discussion in the previous section. Errors are mapped to bit-flip errors on the resource state. To obtain a mixed state over error configurations, one may apply a complete bit-flip (apply $I$ with probability $1/2$ and $X$ with probability $1/2$) channel to every qubit except the qubits associated with the output boundary fields. Ref.~\cite{Negari_arXiv_2025_FT_mixed_state_phases} investigates the properties of this mixed state, and characterizes its topological order in terms of (1) the encoding of a classical memory, (2) the spacetime Markov length \cite{Sang_PRL_2025_markov_length}. Their work proved that the associated mixed state phase transition under a dephasing channel, heralded by a divergence in spacetime Markov length, corresponded exactly to the phenomenological threshold of the circuit. 

We will take (2) as a given, and instead focus on (1), which may be given a simple interpretation in terms of the gauged SSC. We start by explaining how quantum information propagates through the foliated code, refining a little the discussion of Ref.~\cite{Negari_arXiv_2025_FT_mixed_state_phases}. We then show that the classical memory associated with the mixed state has a natural interpretation in terms of the gauged SSC.

Since the nodes of the graph state corresponds to the fields corresponding to $X$-type, we can give them a time label $-0.5, 0, 0.5, ..., T, T+0.5$ identically to the original fields, where $-0.5, T+0.5$ are the boundary fields, and a qubit label. Note this is a slightly different convention from Fig.~\ref{fig:small_foliation_example}, where we have introduced an additional input layer (instead of trimming). In this convention, the final and initial layers $-0.5, T+0.5$ will not contain any qubits associated with measurements. There are two kinds of qubits, corresponding to the physical qubits of the original circuit and the measurements on the original circuit. We label the first kind of qubit $(w, i)$, with $w$ for `wire', and $i = 1, ..., n$. In repeated measurement circuits, the measurements are all of $X$ or $Z$ type stabilizers; we say there are respectively $n_X, n_Z$ stabilizers. Measurements of $X$ ($Z$) type stabilizers are in the half integer (integer) time steps, which we label $(S^X (S^Z), j)$ for $j = 1, ..., n_X (n_Z)$. To distinguish the graph state from the gauged SSC, we will denote Paulis $p^{(\cdot)}_{\cdot}$ with superscript being the time label and subscript being the qubit label. 

The detectors of the repeated measurement circuit, as studied in Sec.~\ref{ssec:examples}, are of the form,
\begin{equation}
    \begin{aligned}
        \tau^{t+0.5}_{j, x} \tau^{t-0.5}_{j, x} \prod_{i \in {\rm Sup } S_j^X} \tau_{i, x}^{t} &\qquad \leftrightarrow \qquad X^{(t+0.5)}_{(sX, j)} X^{(t-0.5)}_{(sX, j)} \prod_{i \in {\rm Sup} S^X_j} X^{(t)}_{(w, i)}, \\
        \tau^{t+0.5}_{j, z} \tau^{t-0.5}_{j, z} \prod_{i \in {\rm Sup } S_j^Z} \tau_{i, z}^{t} &\qquad \leftrightarrow \qquad X^{(t)}_{(sZ, j)} X^{(t-1)}_{(sZ, j)} \prod_{i \in {\rm Sup} S^Z_j} X^{(t-0.5)}_{(w, i)}, 
    \end{aligned}
\end{equation}
where the LHS are the operators corresponding to detectors on the gauge fields, and the RHS are stabilizers of the foliated computation that we identify with detectors. Notice that the time index is off-set for the $Z$-type gauge fields. 

Next, we may ask how quantum information is propagated through this resource state. In preparing the foliated computation, logical operators of the input are transformed into logical operators of the resource state, say $L_i$. After carrying out the foliated computation, we should have a logical operator $L_f$ supported on only the output layer. We say that a stabilizer $S_L$ of the resource state is responsible for the temporal consistency of $L_i, L_f$ if we have 
\begin{equation}\label{eq:LBL_equation}
    S_L L_i = L_f B_L,
\end{equation}
where $B_L$ is product of $X$-type operators supported only on layers $t$ where $t < T+0.5$. 

$Z$-type operators are unaffected by the $cZ$ gates used to prepare the resource state. Hence, for a $Z$-type logical $L_Z$ of the input code, we have $L_i^Z = \prod_{i \in {\rm Supp}(L^Z)} Z_{(w,i)}^{(-0.5)}, L_f^Z = \prod_{i \in {\rm Supp}(L^Z)} Z_{(w,i)}^{(T+0.5)}$. The connectivity of the resource state tells us 
\begin{equation}\label{eq:S_L^Z}
    \begin{aligned}
        S_{L^Z} = \prod_{i \in {\rm Supp}(L^Z)} Z_{(w,i)}^{(T+0.5)} \left(  \prod_{t = 0}^{T}  X_{w, i}^{(t)} \right)  Z_{(w, i)}^{(-0.5)} 
    \end{aligned}
\end{equation}
is a stabilizer of the resource state. Identifying $B_{L_Z} = \prod_{i \in {\rm Supp}(L^Z)} \prod_{t = 0}^{T}  X_{(w, i)}^{(t)}$ gives us the condition of Eq.~\eqref{eq:LBL_equation}. Hence, $S_{L_Z}$ ensures that the logical information associated with $L_Z$ persists through the foliated computation.

On the other hand, $X$-type operators are non-trivially transformed by the preparation of the foliated computation. In particular, for an $X$-type logical $L_X$ of the input code, we have $L_f^X = \prod_{i \in {\rm Supp}(L^X)} X_{(w,i)}^{(T+0.5)}$. However, 
\begin{equation*}
    L_i^X = \prod_{i \in {\rm Supp}(L^X)} Z_{(w,i)}^{(1)} X_{(w,i)}^{(1)}.
\end{equation*}
As such, we have 
\begin{equation}\label{eq:S_L^X}
    S_{L^X} = \prod_{i \in {\rm Supp}(L^X)} X^{(T+0.5)}_{(w, i)} \left( \prod_{t = 0}^{T}  X_{(w, i)}^{(t+0.5)} \right) Z_{(w, i)}^{(0)} X^{(-0.5)}_{(w, i)},
\end{equation}
which stabilizes the resource state, and satisfies Eq.~\eqref{eq:LBL_equation} with $B_{L^X} = \prod_{i \in {\rm Supp}(L^X)} \prod_{t = 0}^{T}  X_{(w, i)}^{(t+0.5)} $. The stabilizer $S_{L^X}$ should be compared to the ``$LBL'$ stabilizers'' of Ref.~\cite{Negari_arXiv_2025_FT_mixed_state_phases}, see the footnote\footnote{Ref.~\cite{Negari_arXiv_2025_FT_mixed_state_phases} identifies two sets of `$LBL'$ stabilizers' as responsible for the propagation of logical information. For $Z$-type logicals, their identification $S_{L^Z}'$ is the same as Eq.~\eqref{eq:S_L^Z}. However, for the $X$-type logicals, they instead identify $S_{L^X}' = \prod_{i \in {\rm Supp}(L^X)} X^{(T+0.5)}_{(w, i)} \left( \prod_{t = 0}^{T}  X_{(w, i)}^{(t+0.5)} \right) X^{(-0.5)}_{(w, i)}$ as their $LBL'$ stabilizer (notice the absence of $Z$ operators in $S_{L^X}'$ compared to Eq.~\eqref{eq:S_L^X}). These operators satisfy $S_{L^Z}' = L_f^Z B_{L_Z} L_i^Z, S_{L^X}' = L_f^X B_{L_X} L_i^X$, giving rise to the nomenclature. However, crucially, $S_{L^Z}', S_{L^X}'$ are not on equal footing. Assuming the input state isn't specifically prepared in an eigenstate, we have that $S_{L^Z}'$ is a stabilizer of any foliated code, whereas $S_{L^X}'$ is never a stabilizer. Conversely, $S_{L^Z}, S_{L^X}$ are always stabilizers of the foliated code corresponding to temporal propagation of the logical information, even for `classical' codes like the repetition code.}\footnote{As a final technical (and perhaps overly pedantic) note, Ref.~\cite{Negari_arXiv_2025_FT_mixed_state_phases} argues that the specific example of the repetition code is already `classical' in the sense that it lacks the $S_{L^X}'$ stabilizer. We take the view that this is not sufficient, since \textit{any foliated code} lacks the $S_{L^X}'$ stabilizer. The reason is the repetition code is classical is due to the errors it does (or does not) correct. In the absence of errors it is capable of transmitting quantum information. In terms of the gauged SSC/mixed state order of the resourece state, the classicality reveals itself by inspection of the Wilson loops, which are constrained to one type ($X/Z$) of gauge field. For the dephased state, this structure is a spatial one (detectors are constrained to odd or even time-slices), and does not otherwise inherently differentiate `quantum' or `classical' codes.  }. 

As observed in Ref.~\cite{Negari_arXiv_2025_FT_mixed_state_phases}, these objects can be interpreted as a classical memory. We identify the measurement outcomes from carrying out the foliated computation with physical classical bits. The detectors (identified with $X$ operators on the resource state) play the role of classical checks, in that the product of these measurement outcomes always give $+1$. On the other hand, the $B_{L_Z}, B_{L_X}$ operators play the role of classical logical bits, with the product of their measurement outcomes random. A quick counting argument tells us that this does not account for all degrees of freedom. For concreteness, say we have $n$ qubits, with $m$ stabilizers and hence $2 (n - m)$ logical operators. There are $2 n T$ wire qubits that are measured. The stabilizers are repeatedly measured in every time-step except the input, so there are $m T$ measurements associated with stabilizers, however these are always $+1$ (assuming the input state is noiseless). Discounting them, we have $2nT$ physical bits. Each pair of repeated measurements forms a detector, so there are $m(T-1)$ checks, and a total of $2(n-m)$ logical bits. We have $2nT - m(T-1) - 2(n-m) = (2n-m)(T-1) + m$ unaccounted-for degrees of freedom. If desired, these additional degrees of freedom can be explicitly identified\footnote{First, consider a stabilizer $s$. If $s$ is $Z$-type, then we have that $\prod_{i \in {\rm Sup} s}  X_{(w, i)}^{(t)}$, for any $t < T+0.5$, is independent from the identified checks and logical bits, with analogous $X$-type operator on the half-integer time steps. These hence give $m T$ degrees of freedom. For any $X$-type logical $L_X$ of the input code, we also have $\prod_{i \in {\rm Sup} L_X}  X_{(w, i)}^{(t + 0.5)}$, as independent of the previous constraints, and analogously for $Z$-type logicals $L_Z$. These give $2(n-m)(T-1)$ degrees of freedom, with the $-1$ coming from the fact that combining them all yield the previously identified logical bits associated with $B_{L_Z}, B_{L_X}$. These sum up to $(2n-m)(T-1)+m = (2n-m)(T-1) + m$ degrees of freedom, accounting for the remaining bits from the earlier discussion. These are all logical bits in the sense that their associated outcomes are random.}, but their exact form is not important to our discussion.

Finally, we return to the gauged SSC. The characterization of a mixed state phase in terms of a classical memory can be given a simple interpretation in terms of the gauge complex, Eq.~\eqref{eq:gauge_complex_2}. First, we note that a classical memory, or classical error correction code is simply a chain complex with two terms. Accordingly, we have identified the measurements of the foliated computation with gauge fields -- this is the space $V_1$. In the gauge complex, the checks acting on this are $R_{\mathcal{C}}$, which correspond to the detectors. Hence, the classical memory associated with this mixed state is simply the classical error correction code associated with the subcomplex $V_1 \xrightarrow{R_{\mathcal{C}}} V_0$. The procedure of applying maximal bit-flip noise to all qubits up to layer $T$ corresponds to dephasing all the Gauss laws (except those supported only on the output boundary). In the parlance of weak and strong symmetries \cite{Buca_NJP_2012_weak_and_strong}, the resulting mixed state retains operators associated with $R_{\mathcal{C}}$ as strong symmetries whereas those associated with $G_{\mathcal{C}}$ up to layer $T$ are weakened by the dephasing.

\subsection{Pauli Noise Learning}\label{ssec:learning_Pauli_noise}

One arena in which spacetime codes have a found a surprising use is in learning theory \cite{xiao_arXiv_2026_insitu_benchmarking, zheng_arXiv_2026_logical_noise}. In this section, we provide some insight into this, by establishing a connection between the gauged SSC and the pattern transfer graphs used to establish gauge freedoms in learning Pauli noise \cite{Chen_PRXQ_self_consistent_learning}. 

\begin{figure}[ht!]
\centering 
\includegraphics[width=0.7\linewidth]{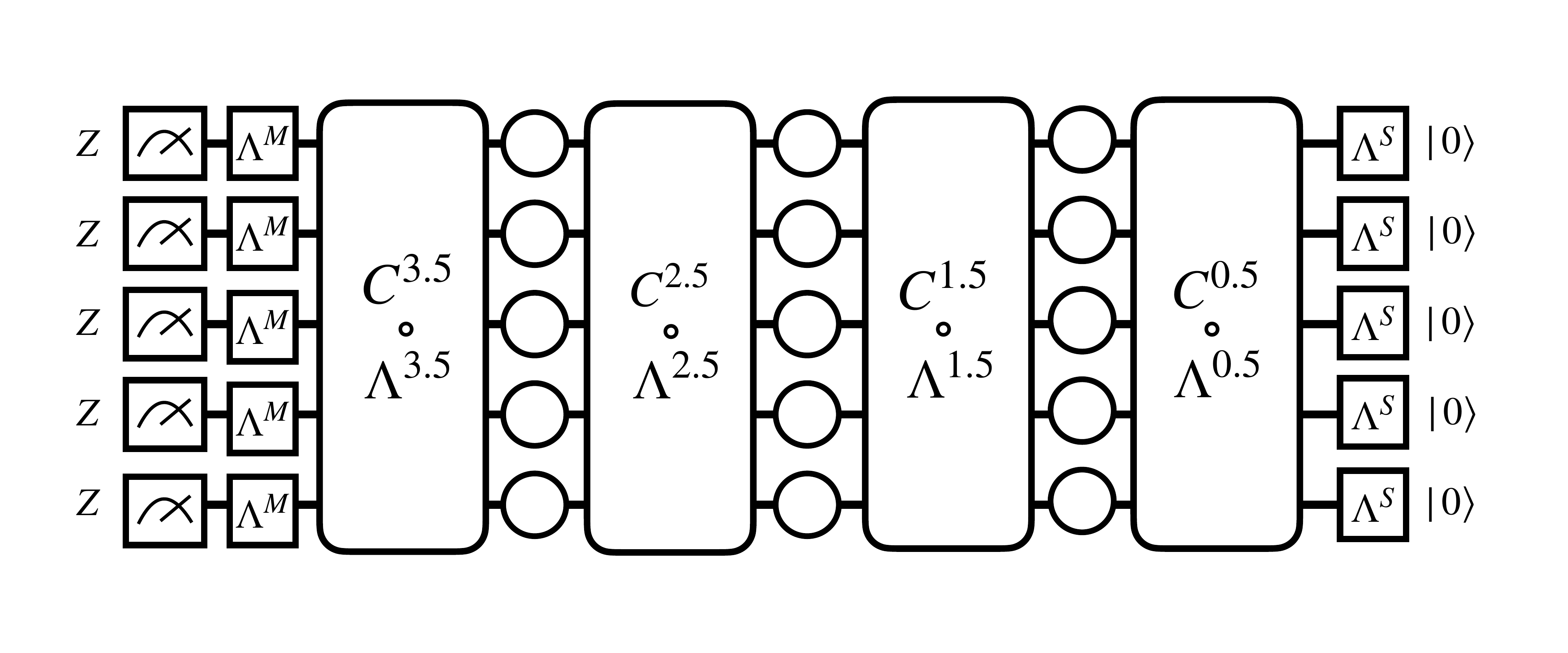} \caption{A typical benchmarking circuit, containing Clifford gates in every time step, as well as preparation of single-qubit $Z$ eigenstates at the start and single-qubit $Z$-measurements at the end. Each layer is assumed to be associated with a Pauli noise channel $\Lambda^{t+0.5}$. The preparation and measurement steps are associated with single qubit Pauli noise channels $\Lambda^{S/M}$. The goal is to estimate as many parameters of these noise channels as possible from the measurement outcomes of the circuit. }\label{fig:general_noise_circuit} 
\end{figure}

\subsubsection{Pauli noise models}

We will start with the usual benchmarking circuit depicted in Fig.~\ref{fig:general_noise_circuit}. We prepare each qubit in a $Z$ eigenstate, carry out some series of Clifford operations $C^{t+0.5}$, restricted to unitaries, and then measure every qubit in the $Z$ basis. Since single qubit basis changes are Clifford operations, this absorbs state preparation and measurement in arbitrary Pauli eigenstates. Such circuits are applied in a range of benchmarking experiments, see Ref.~\cite{Hashim_PRXQ_2025_benchmarking_review} for a review. We assume the following very general noise model: Before every Clifford operation $C^{t+0.5}$ we apply a Pauli noise channel $\Lambda^{t+0.5}$. For Clifford circuits, an arbitrary noise channel can be turned into a Pauli noise channel via twirling \cite{Wallman_PRA_2016_twirling}. We model state preparation and measurement (SPAM) noise with single qubit Pauli noise channels $\Lambda^{S/M}$ occurring just after/before the SPAM operations. In the parlance of Ref.~\cite{Chen_PRXQ_self_consistent_learning}, this is referred to as a \textit{complete} noise ansatz.

To parameterize this noise ansatz, we need to describe Pauli noise channels. A \textit{Pauli noise channel} $\Lambda$ is a distribution ${\rm Pr}(q)$ of Kraus operators over Pauli operators $q \in \overline{P}_n$. However, for our purposes it is useful to rewrite it in terms of Pauli eigenvalues $\lambda_q$, 
\begin{equation*}
    \Lambda(\cdot) = \sum_{q \in \overline{P}_n} {\rm Pr}(q) q (\cdot) q = \frac{1}{2^n} \sum_{q \in \overline{P}_n} \lambda_q q {\rm Tr}(q (\cdot)),
\end{equation*}
where ${\rm Pr}(q)$ and $\lambda_q$ are related by a Fourier transform, $\lambda_q = 4^{-n} \sum_{p \in \overline{P}_n} p (-1)^{\langle q, p \rangle}$, where $\langle \cdot, \cdot \rangle$ is the symplectic inner product. The trace-preserving property fixes $\lambda_I = 1$, so $\Lambda$ is characterized by $4^n - 1$ parameters. The advantage of writing these channels in terms of Pauli eigenvalues is that the channel can be easily composed, eg. two channels $\Lambda^{2}, \Lambda^{1}$ with Pauli eigenvalues $\lambda^{1/2}_q$ compose to give a channel with  Pauli eigenvalues $\lambda^1_q \lambda^2_q$. In our circuit, we take $\Lambda^{t+0.5}$ to be global noise channels, whereas the single qubit state preparation noise $\Lambda^S$ and measurement noise $\Lambda^M$, can WLOG be taken to only have $\lambda^{M/S}_Z$ components.

The Pauli eigenvalues fully parameterize the noise in the circuit. Following Ref.~\cite{Chen_PRXQ_self_consistent_learning}, we  will actually parameterize the noise ansatz via their logarithms; in the circuit of Fig.~\ref{fig:general_noise_circuit}, we have Pauli eigenvalues $\lambda^{\alpha}_{q}$, with $q$ being Paulis and $\alpha = t+0.5$ or $ \alpha = M/S$, we will denote $x^{\alpha}_{\beta} = - \ln \lambda^{\alpha}_{\beta}$, so that composing Pauli eigenvalues becomes taking linear combinations of their logarithms. We briefly consider mid-circuit measurements in Sec.~\ref{sssec:midcircmeas}.

\subsubsection{Learnable degrees of freedom and Wilson loops}

Now let us consider the gauged SSC for the circuit of Fig.~\ref{fig:general_noise_circuit}. We start with the probability distribution $F_k(\mathbf{x})$ of bit-strings $k$ in the measurement outcomes of the circuit for a given set of logarithmic eigenvalues $\mathbf{x}$. This is derived in App.~\ref{app_sec:deferred_proofs_applications} as, 
\begin{equation}\label{eq:outcome_dist_complete}
    \begin{aligned}
        F_k(\mathbf{x}) 
        &= \frac{1}{4^n} \sum_{b, b' \in \mathbb{Z}_2^n} \exp \left( - \sum_{j = 0}^{T-1} x_{ U^{0 \rightarrow j} Z[b'] (U^{0 \rightarrow j})^{\dag}}^{T-j+0.5} -x^M_{b} - x^S_{b'}  \right) (-1)^{k \cdot b}  {\rm Tr}(Z[b]  U^{0 \rightarrow T} Z[b'] (U^{0 \rightarrow T})^{\dag} ), \\
        &= \frac{1}{2^n} \sum_{b, b' \in \mathbb{Z}_2^n} \exp \left( - \sum_{j = 0}^{T-1} x_{ U^{0 \rightarrow j} Z[b'] (U^{0 \rightarrow j})^{\dag}}^{T-j+0.5} -x^M_{b} - x^S_{b'}  \right) (-1)^{k \cdot b} \delta_{Z[b],  U^{0 \rightarrow T} Z[b'] (U^{0 \rightarrow T})^{\dag} },
    \end{aligned}
\end{equation}
where we have denoted $Z[b]:= \otimes_{j = 1}^n Z_j^{b_j}$ for $b \in \mathbb{Z}_2^n$. Due to Lem.~\ref{lem:no_prop_redundancies} and the fact that we only prepare and measure qubits in the first and last time steps, every redundancy in this circuit must be associated with some initial product of operators $Z[b^i]$ that propagate to some final product of operators $Z[b^f]$, i.e. $Z[b^f] = U^{0 \rightarrow T} Z[b^i] (U^{0 \rightarrow T})^{\dag}$. Explicitly, redundancies have the form
\begin{equation}\label{eq:step-wise_label_R}
\begin{aligned}
    R &= \left\{ \eta^0(Z[b^i]),  \eta^t(Z[b^f]), \eta^t(U^{0 \rightarrow t} Z[b^i](U^{0 \rightarrow t})^{\dag})\eta^{t-1}(U^{0 \rightarrow t-1} Z[b^i](U^{0 \rightarrow t-1})^{\dag}): t = 1, ..., T   \right\} \\
    &=:\left\{ \eta^0(R^0),  \eta^t(R^T), \eta^t(R^t)\eta^{t-1}(R^{t-1}): t = 1, ..., T   \right\}.
\end{aligned}
\end{equation}
This allows us to rewrite Eq.~\eqref{eq:outcome_dist_complete} as 
\begin{equation}
    \begin{aligned}
        F_k(\mathbf{x}) 
        &= \frac{1}{2^n} \sum_{R \in \mathcal{R} } \exp \left( - \sum_{j = 0}^{T-1} x^{t+0.5}_{R^t} - x^M_{R^T} - x^S_{R^0}  \right) (-1)^{k \cdot b},
    \end{aligned}
\end{equation}
where $\mathcal{R}$ is the set of all redundancies in the circuit. This shows that the outcome distribution only depends on vectors $\mathbf{x}$ corresponding to the Wilson loops of the circuit. To explicitly estimate these quantities, one simply computes $E[R^T] = E[Z[b^f]]$, for which we have 
\begin{equation}\label{eq:expectation_value_distribution}
    - \ln E \left[ Z[R^T] \right] = \sum_{j = 0}^{T-1} x^{t+0.5}_{R^t} + x^M_{R^T} + x^S_{R^0}.
\end{equation}

This calculation shows us that the Wilson loops fully determine the outcome distribution of the circuit, and hence determine all the noise parameters that may be estimated from this distribution. Denoting the vector on the right-hand side of Eq.~\eqref{eq:expectation_value_distribution} as $\mathbf{x}_{R}$, we see that clearly any linear combination of $\mathbf{x}_{R}$ can be estimated, whereas any vector orthogonal to their span cannot be computed from the output distribution (the expectation value of $Z[b]$ for any $b$ not associated with a redundancy in the above manner is $0$). One is tempted to identify these unlearnable degrees of freedom with the Gauss laws; however, this does not follow  straightforwardly, as we have established Gauss laws to be orthogonal to Wilson loops over $\mathbb{Z}_2$ rather than $\mathbb{R}$.

\subsubsection{Mid-circuit measurements}\label{sssec:midcircmeas}

We now discuss the case of mid-circuit measurements. For simplicity, we consider again the repeated measurement circuit. We will work with a reduced noise model, where we either have an idling or measurement errors. This is not a realistic model for benchmarking purposes \footnote{That is, unless one is really committed to phenomenological noise, the way this work seems to have gotten by this point.}, but illustrative of the connections between detectors and learnability.

To model phenomenological noise, each measurement is assumed to be subject to an identical measurement error $\Lambda^M$ which only has one parameter $\lambda^M$ say, corresponding to a bit-flip error on the ancilla used to measure the stabilizer. In principle, there is also state preparation error on the ancilla, but in the phenomenological model it is easy to see that is exactly equivalent to a measurement error, so we model the ancilla preparation step as perfect. 

Besides that, the idling step is subject to a Pauli error channel $\Lambda^I$, with $\lambda^I_p$ for $n$-qubit Paulis $p$. We will work with a discretization wherein all stabilizers are measured in every time step, and index measurements with the respective stabilizer and time step, $S^{t+0.5}_j$, with $j$ indicating which stabilizer and the superscript indicating which time step.

Redundancies in this circuit are generated by redundancies of the form, 
\begin{equation}
    R = \{\eta^{t+1}(q), \eta^t(q), \eta^{t+1}(q) \eta^t(q)\} =: \{\eta^{t+1}(R^{t+1}), \eta^t(R^t), \eta^{t+1}(R^{t+1}) \eta^t(R^t)\},
\end{equation}
for some stabilizer $q = S_j$. Let us denote the outcome of the first measurement by $m_1$ and the outcome of the second by $m_2$. Without noise, $m_1 = m_2$. Set the POVM operators to be $E_m = (I + (-1)^m q)/2$. Then we have, 
\begin{equation}
    \begin{aligned}
        F_{m_1, m_2} = {\rm Tr} \left( \langle m_1 |_1 \otimes \langle m_2 |_2 \left( \mathcal{M}_2 \circ \Lambda^M_2 \circ \Lambda^I \circ \mathcal{M}_1 \circ \Lambda^M_1 (\rho \otimes | 0 \rangle \langle 0 |_{1} \otimes | 0 \rangle \langle 0 |_2 \right) | m_1 \rangle_1 \otimes | m_2 \rangle_2 \right),
    \end{aligned}
\end{equation}
where the labels $1, 2$ denote ancillas, with $|0 \rangle_i$ being some perfectly prepared reference state, and we model measurements with the map 
\begin{equation}
    \mathcal{M}_i \left( \rho \otimes | a \rangle \langle  a |_i \right) = \sum_{m = 0, 1} E_{m + a} \rho E_{m+a} \otimes | m \rangle \langle m |,
\end{equation}
where $a = 0,1$. This indicates that if the ancilla was in the state $| 0 \rangle$, then measuring an outcome $m$ on the ancilla afterwards would yield the correct POVM $E_m$. Conversely, we mix them up if the ancilla was in the state $| 1 \rangle$. Let's perform the calculations bit by bit. 
\begin{equation}
    \begin{aligned}
        &\langle m_1 |_1 \mathcal{M}_1 \circ \Lambda^M_1 (\rho \otimes | 0 \rangle \langle 0 |_{1} ) | m_1 \rangle_1 \\
        &=\frac{1}{2} \sum_{s_1 = 0, 1} \lambda^M_{s_1} \langle m_1 |_1 \mathcal{M}_1 (\rho \otimes Z_1^{s_1}) | m_1 \rangle_1 \\
        &= \frac{1}{2} \sum_{s_1 = 0, 1} \lambda^M_{s_1} \langle m_1 |_1 \mathcal{M}_1 (\rho \otimes ( | 0 \rangle \langle 0 |_1 + (-1)^{s_1} | 1 \rangle \langle 1 | )) | m_1 \rangle_1 \\
        &= \frac{1}{2} \sum_{a_1, s_1 = 0, 1} (-1)^{a_1 s_1} \lambda^M_{s_1} \langle m_1 |_1 \mathcal{M}_1 (\rho \otimes | a_1 \rangle \langle a_1 |_1 ) | m_1 \rangle_1 \\
        &= \frac{1}{2} \sum_{m, a_1, s_1 = 0, 1} (-1)^{a_1 s_1} \lambda^M_{s_1}  E_{m + a_1} \rho E_{m+a_1} \langle m_1 | m \rangle \langle m | m_1 \rangle_1 \\
        &= \frac{1}{2} \sum_{a_1, s_1 = 0, 1} (-1)^{a_1 s_1} \lambda^M_{s_1}  E_{m_1 + a_1} \rho E_{m_1 +a_1} 
    \end{aligned} 
\end{equation}
Applying $\Lambda^I$ to the operator part,
\begin{equation}
    \begin{aligned}
        \Lambda^I (E_{m_1 + a_1} \rho E_{m_1 +a_1}) = \frac{1}{2^n} \sum_{p \in \overline{\mathcal{P}_n }} \Lambda^I_p  {\rm Tr} \left( p E_{m_1 + a_1} \rho E_{m_1 +a_1} \right) p
    \end{aligned}
\end{equation} 

We can append the second ancilla and apply the second set of measurement and ancilla noise to get
\begin{equation}
    \begin{aligned}
        \langle m_2 |_2 \mathcal{M}_2 \circ \Lambda^M_2 (p \otimes | 0 \rangle \langle 0 |_{2} ) | m_2 \rangle_2 
        = \frac{1}{2} \sum_{a_2, s_2 = 0, 1} (-1)^{a_2 s_2} \lambda^M_{s_2}  E_{m_2 + a_2} p E_{m_2 +a_2} 
    \end{aligned}
\end{equation}

Putting all the expressions together, 
\begin{equation}
    \begin{aligned}
        F_{m_1, m_2} 
        &= \frac{1}{2^{n+2}} \sum_{p \in \overline{\mathcal{P}}_n }  \sum_{a_1, a_2, s_1, s_2 = 0, 1} (-1)^{a_2 s_2 + a_1 s_1} \lambda^M_{s_2} \lambda^M_{s_1}   \lambda^I_p  {\rm Tr} \left( p E_{m_1 + a_1} \rho E_{m_1 +a_1} \right)  {\rm Tr} \left( E_{m_2 + a_2} p E_{m_2 +a_2} \right)
    \end{aligned}
\end{equation}
Now the last term is only non-zero if $p = q$ in which case it is $2^{n-1} (-1)^{m_2 + a_2}$. Hence
\begin{equation}
    \begin{aligned}
        F_{m_1, m_2}
        &= \frac{1}{8}   \sum_{a_1, a_2, s_1, s_2 = 0, 1} (-1)^{a_2 s_2 + a_1 s_1 + m_2 + a_2} \lambda^M_{s_2} \lambda^M_{s_1}   \lambda^I_q  {\rm Tr} \left( q E_{m_1 + a_1} \rho E_{m_1 +a_1} \right) \\
        &=  \frac{1}{8}   \sum_{a_1, a_2, s_1, s_2 = 0, 1} (-1)^{a_2 (1 + s_2) + a_1 (1 + s_1) + m_2 + m_1} \lambda^M_{s_2} \lambda^M_{s_1}   \lambda^I_q  {\rm Tr} \left( E_{m_1 + a_1} \rho E_{m_1 +a_1} \right) \\
        &= \frac{1}{4} \sum_{a_1, s_1 = 0, 1} (-1)^{a_1 (1 + s_1) + m_2 + m_1} \lambda^M_{1} \lambda^M_{s_1}   \lambda^I_q  {\rm Tr} \left( E_{m_1 + a_1} \rho E_{m_1 +a_1} \right)
    \end{aligned}
\end{equation}

Finally, the expectation value of the parity of the two measurements is given by 
\begin{equation}
    \begin{aligned}
        \ln E[(-1)^{m_1 + m_2}] = \ln (F_{0, 0} + F_{1, 1} - F_{0, 1} - F_{1, 0}) = 2 x_1^M + x_q^I,
    \end{aligned}
\end{equation}
so we see that the detector actually yields an estimator for $2 x_1^M + x_q^I$. This is of course read-off from the detector in a similar way to before, with the two measurement slices associated with $\lambda_1^M$ and the propagator (through the identity) associated with $\lambda_q^I$.  

A Clifford circuit in general has redundancies comprising of some set of measurements $m \in R$ and some set of propagators $\eta^{t+1}(U^{t+0.5} q^t (U^{t+0.5})^{\dag}) \eta^{t}(q^t)$, where we combine all propagators associated with the same time-slices, and we set $q = I$ in time-slices for which the propagators have no support. We may perform a more tedious generalization of the above calculation to obtain that the log of the expectation value of the parity across all measurements $m$ yields an estimate of $\sum_{t} x_{q^t}^{t+0.5} + \sum_{m \in M_R} x^m_1$ in a model with a separate global Pauli noise channel in every time-slice and readout errors $x^m_1$ associated with each measurement. Conversely, averaging across all possible input into the circuit, expectation values across measurements that do not correspond to any detector in the above manner will yield random results. Hence, detectors in arbitrary Clifford circuits in fact characterize the learnable degrees of freedom.

\subsubsection{From gauged SSC to Pattern Transfer Graph}

A complete theory for the learnable and unlearnable degrees of freedom in Pauli noise channels has been established in Ref.~\cite{Chen_PRXQ_self_consistent_learning}. Here, we will only describe the formalism on a high level, without pretending to any formality. Suppose we are given a set of gates $\{ U \}$, and each gate comes with an associated Pauli noise channel $\Lambda_U$, such that the actual gates carried out are $U \Lambda_U( \cdot ) U^{\dag}$. Assuming single-qubit gates are free, we wish to characterize the parameters of $\Lambda_U$ for each $U$ in the gate set. To do so, one constructs the \textit{Pattern transfer graph} (PTG), which has vertices labelled by bit-strings, which represent the support of Pauli operators. For each Pauli operator $p$ and each gate $U$, we draw a (directed) edge from the bit-string with $1$'s only in the support of $p$ to the bit-string with $1$'s only in the support of $U p U^{\dag}$\footnote{When single-qubit gates are not free, we instead have a vertex for each Pauli operator. The edges then connect vertices which can be transformed into each other by gates in the gate set.}. Each edge can hence be said to represent a noise parameter, since it is associated with a gate and a Pauli operator. In the case of circuits of the form Fig.~\ref{fig:general_noise_circuit}, an additional vertex, called a root node, is introduced to capture SPAM errors, and an edge is drawn from this vertex to each other vertex to represent state preparation noise supported on the bit-string of the target vertex, in the opposite direction for the analogous measurement noise parameter. The extended construction for mid-circuit measurements is proposed in Ref.~\cite{Zhang_PRXQ_2025_midcircuit_learning}. For such a gate set, the learnable degrees of freedom in this gate set are in exact correspondence to rooted cycles (cycles containing the root node) in this graph, whereas the unlearnable degrees of freedom correspond to cuts and unrooted cycles.

The PTG characterizes Pauli noise associated with a gate set, whereas our discussion has used redundancies to characterize Pauli noise in a circuit. How are these two related? Let us consider a circuit, of the form Fig.~\ref{fig:general_noise_circuit}, constructed from some fixed gate set, and suppose this circuit has some redundancy $R$. Let ${\rm pt}(p)$ be the pattern of a Pauli operator $p$, i.e. the bit-string with $1$'s where $p$ has support. Since the circuit is constructed from a fixed gate set, we have that ${\rm pt}(R^t) \rightarrow {\rm pt}(R^{t+1})$ (see Eq.~\eqref{eq:step-wise_label_R}) is an edge in the graph. In this direction, we can thus map the propagators of the form $\eta^{t+1}(R^{t+1}) \eta^t(R^t)$ to edges in the PTG, and the initialization and preparation operators $\eta^T(R^T), \eta^0(R^0)$ to edges in and out of the root node. This immediately yields expressions of the form Eq.~\eqref{eq:expectation_value_distribution} associated to each detector, telling us what the associated learnable degree of freedom is. We comment that this mapping only works in one direction. It may be cleaner to think about vertices as associated with Paulis rather than patterns, and while this discussion only applies to circuits of the form Fig.~\ref{fig:general_noise_circuit}, we expect a similar analogy to hold for circuits with mid-circuit measurements. The upshot of this discussion is that we can think about redundancies of a circuit, and thus detectors in fault tolerance, as descended from the learnable degrees of freedom associated with the Pauli noise of some gate set. We hope this compelling analogy between the observables in either theory provides some learning theoretic insights into the modern detector-focused approach to fault tolerance.

\section{Discussion and Conclusion}

In this work, we have tried to develop a concrete bridge between fault tolerance and lattice gauge theory by showing that the spacetime code, when appropriately gauged, yields a classical gauge theory that faithfully encodes many elements of fault tolerance. In particular, by gauging a set of elementary circuit operators used to generate the spacetime code, we are able to reorganize the notions of circuit action, fault equivalence, and detectors in gauge-theoretic terms: faults become gauge configurations, equivalence becomes Gauss laws, and detectors become gauge-invariant observables, or Wilson loops. We have also tried to clarify several conceptual points that may be less transparent in the subsystem spacetime code, including the precise role of redundancies and the relationship between local and global redundancies, as well as the operational interpretation of boundary conditions. With the formalism in hand, we have then pointed out connections to foliation, mixed state order, and learning theory.  

Finally, we comment on possible future directions. While this work has focused on mapping circuits to gauge theories, the ability to go in the opposite direction would likely be a much more practically fruitful direction, as it would allow us to use gauge theories with desirable properties to construct new and interesting fault tolerant protocols. Here, we have only examined the most basic elements of gauge theory. It would be interesting to see how more exotic elements of gauge theory may show up in the fault tolerance context -- for instance, here we have restricted ourselves to flat configurations, where all our Gauss law operators always take the value $+1$. What would charges that flip these signs correspond to in circuits? We have also glibly referred to this gauge theory as a phase of matter. While we have proven a sense of distance inherited from a fault-tolerant protocol, this is only suggestive of stability, as we have not examined any phase diagrams, or examined the quantitative impact of noise or temperature. Finally, in this work we have restricted ourselves to Clifford circuits only. How can we extend the construction to non-Clifford circuits? The connection to foliation hints at one of the primary difficulties -- our work has mapped circuits to a static gauge theory, in the sense that we are fixing a set of spatial operators. However, to carry out non-Clifford computation in MBQC requires adaptivity \cite{Raussendorf_AnnPhys_2006_one_way_QC}, suggesting that  one may not be able to fix such a set of operators in advance. One possible starting point is the 1D gauge theory of cluster state computation developed in Ref.~\cite{Wong_2024_quantum_gauge_theory_MBQC}, which also treats non-Clifford computation. We leave all these directions to future work.



\section*{Acknowledgements}

I would like to acknowledge R. Sohal for teaching me about gauge theories and mixed state order, and a lot of the physics required to complete this work. Many of the key ideas in this work have been refined through constant discussions with him. I would also like to acknowledge X. Fu for teaching me about spacetime codes, Q. Xu for pointing out the compelling analogy between fault complexes and outcome codes, M. Teo for helpful discussions of graph homology and also pointing me to Ref.~\cite{kubica_2018_arXiv_ungauging_QEC}, which kickstarted this project. I also acknowledge A. Pocklington, G. Zheng, and S. Chen for many helpful discussions about gauge theory, quantum error correction, and learning theory respectively, and in particular thank A. Pocklington for reading drafts of this manuscript. I would like to acknowledge a helpful talk by R. Sahay that helped to clarify the role of local and global redundancies, as well as a talk by S. Lee that inspired the connection to Pauli noise learning.

\break

\appendix

\section{Some Useful Concepts}\label{app_sec:some_useful_concepts}

\subsection*{Chain complexes}

This section provides a quick description of chain complexes relevant to error correction on qubits \cite{Kitaev_2003_ann_phys_toric_double, Breuckmann_PRXQ_LDPC_review}. A chain complex $C$ of length $n$ is a collection of $n+1$ vector spaces $C_i$ (which we will always take to be over $\mathbb{Z}_2$) and linear maps $\partial_i: C_i \rightarrow C_{i-1}$ (referred to as boundary maps), written 
\begin{equation*}
    C = (C_n \xrightarrow{\partial_n} C_{n-1} ... \xrightarrow{\partial_1} C_1  \xrightarrow{\partial_0} C_0),
\end{equation*}
such that $\partial_i \partial_{i+1} = 0$, i.e. boundaries have no boundary. Elements of $C_i$ are referred to as $i$-chains. An $i$-cycle is an element of ${\rm ker} C_i$, whereas an $i$-boundary is an element of ${\rm Im} C_i$. The $i$-th homology is $H_i(C) = {\rm ker} \partial_i / {\rm Im} \partial_{i+1}$, with elements of $H_i(C)$ referred to as homology classes.

In this work, we do not use any particularly deep properties of chain complexes, merely relying on them to provide a succinct description of the physics common to error correction and gauge theory. 

\subsection*{Symplectic and Weyl representation of Paulis}

This section provides a quick review of the binary symplectic representation of Pauli operators \cite{NC_2019}. Let $x = (a, b)$, $W_x = i^{a \cdot b} \bigotimes_{j=1}^n X^{a_j} Z^{b_j}$. These give all Paulis up to a phase, and $W_x$ are also referred to as Weyl operators. We refer to the $x$ as the symplectic representation of the Pauli $\propto W_x$. 

There are several convenient facts about this representation. First, two vectors $x, y \in \mathbb{Z}_2^n \times \mathbb{Z}_2^n$, we have $W_{x + y} \propto W_x W_y$, reflecting the isomorphism between $\overline{\mathcal{P}}_n$ and $\mathbb{Z}_2^n \times \mathbb{Z}_2^n$. 

Second, this representation encodes the commutation relation between two Paulis via the symplectic inner product. In particular,  two Weyl operators $W_x, W_y$, associated with $x = (x^1, x^2), y = (y^1, y^2)$ commute if and only if 
\begin{equation*}
    \begin{aligned}
        \begin{pmatrix}
            x^1 & x^2
        \end{pmatrix} \begin{pmatrix}
            0 & I \\ I & 0
        \end{pmatrix} \begin{pmatrix}
            y^1 \\ y^2
        \end{pmatrix} = 0,
    \end{aligned}
\end{equation*}
and anticommute if they are $1$. The sandwiched matrix is referred to as the symplectic two-form, denoted 
\begin{equation*}
    \Omega = \begin{pmatrix}
            0 & I \\ I & 0
        \end{pmatrix}.
\end{equation*}

Finally, the Weyl operators $W_{x_1}, ..., W_{x_m}$ are independent if and only if the vectors $x_1, ..., x_m$ themselves are also independent.

\section{Deferred proofs from Sec.~\ref{sec:clifford_circuits_and_gauge_generators}}\label{app_sec:deferred_proofs}

\subsection*{Proofs from Sec.~\ref{ssec:ISGs}}\label{app_ssec:proofs_ISGs}

\subsubsection*{Proof of Lem.~\ref{lem:ISG_update_prop}}

\begin{proof}
    In the following, all equalities only hold up to $\pm$ signs, i.e. we are implicitly dealing with the phaseless Pauli group.
    
    
    Consider an operator $p \in {\rm ISG}(t+1)$. Suppose  $p \in \mathcal{M}^{t>}$, then $\eta^t(p)$ is a measurement slice, and $\eta^t(p) \in G_{\mathcal{C}}$. 
    
    Conversely, suppose $p \notin \mathcal{M}^{t>}$. From Prop.~\ref{prop:ISG_update} or otherwise, $p$ must commute with all measurements in $\mathcal{M}^{t>}$. From Lem.~\ref{lem:centralizer_from_propagator}, this implies that 
    \begin{equation*}
        O_1 := \eta^{t+1}(p) \eta^{t}((U^{t>})^{\dag} p U^{t>}),
    \end{equation*} 
    is a product of elementary propagators. 
    
    Next, from Prop.~\ref{prop:ISG_update}, $p$ must be a product of $U^{t>} s_i (U^{t>})^{\dag}$, with $s_i$ generators of ${\rm ISG}(t)$, as well as some set of measurement operators $m \in \mathcal{M}^{t>}$. Specifically, we may write 
    \begin{equation*}
        p = \prod_{s_i \in A} \left( U^{t>} s_i (U^{t>})^{\dag} \right) \times \prod_{m \in B} m,
    \end{equation*}
    for some set $A \subseteq {\rm ISG}(t)$ restricted to generators $s_1, ..., s_k$, and $B \subseteq \mathcal{M}^{t>}$.

    Finally, since measurements in a single Clifford operation must mutually commute, by Lem.~\ref{lem:centralizer_from_propagator}, the propagator $\eta^{t+1}(m) \eta^t((U^{t>})^{\dag} m U^{t>})$ is a product of ECOs for all $m \in B$. Multiplying by the measurement slice and over $B$, 
    \begin{equation*}
        O_2 := \prod_{m \in B} \eta^t((U^{t>})^{\dag} m U^{t>})
    \end{equation*}
    is a product of ECOs. Putting it all together, we have 
    \begin{equation*}
    \begin{aligned}
        \eta^{t+1}(p) &= O_1 \times \eta^{t}((U^{t>})^{\dag} p U^{t>}) \\
        &= O_1 \times \prod_{s_i \in A} \eta^{t} (s_i) \times \prod_{m \in B} \eta^{t}\left( (U^{t>})^{\dag} m U^{t>} \right) = O_1 \times \prod_{s_i \in A} \eta^{t} (s_i) \times O_2,
    \end{aligned}
    \end{equation*}
    where we have argued that $O_1, O_2$ are products of ECOs, and $\eta^t(s_i)$ is a product of ECOs by assumption of the lemma.
\end{proof}

\subsubsection*{Proof of Lem.~\ref{lem:alt_ISG_prop}}

\begin{proof}
    We do this by induction. Clearly the statement is trivially true for $t = 0$, since there no time-steps less than $0$. 
    
    Now suppose this is true for some time $s$. Let us consider a gauge operator $p \in \langle G_{\mathcal{C}} \rangle$ supported only on time-slice $s+1$, and generated by elementary circuit operators associated with time steps $t < s + 1$ only. We can divide the generators of $p$ into two sets $A_{<}, A_{s}$, such that 
    \begin{equation*}
        p = \pm \prod_{g' \in A_<} g' \cdot \prod_{g \in A_s} g,
    \end{equation*}
    where $A_<$ comprises elementary circuit operators associated with time-steps less than $s$ and $A_s$ comprises elementary circuit operators associated with time-step $s+0.5$ only.

    Operators in $A_s$ can only be supported on time-steps $s, s+0.5, s+1$, so in order for $p$ to only be supported on time-step $s+1$, we need that $\prod_{g' \in A_<} g'$ is not supported on any time-step earlier than $s$. This implies it is only supported on time-step $s$, so by hypothesis, $\prod_{g' \in A_<} g' \in \overline{\rm ISG}(s)$. Since $\prod_{g \in A_s} g$ must cause this term to disappear, we obtain that 
    \begin{equation*}
    \begin{aligned}
        \prod_{g \in A_s} g &= g^{s>}_{\rm prop}\left( \prod_{g' \in A_<} g' \right) \cdot \left( \mbox{ measurement slices }\right), 
    \end{aligned}
    \end{equation*}
    Since condition (i) tells us that $p$ has no support on $s+0.5$, we obtain from Lem.~\ref{lem:centralizer_from_propagator} that $U^{s>} \prod_{g' \in A_<} g' (U^{s>})^{\dag}$ commutes with all measurements in $\mathcal{M}^{s>}$. From Prop.~\ref{prop:ISG_update}, this implies that it is in $\overline{\rm ISG}(s+1)$. Since measurement slices must also be in $\overline{\rm ISG}(s+1)$, we obtain that $p$ is a product of elements in $\overline{\rm ISG}(s+1)$ and hence is itself in $\overline{\rm ISG}(s+1)$.
\end{proof}

\subsection*{Proofs from Sec.~\ref{ssec:elts_of_FT}}


\subsubsection*{Proof of Prop.~\ref{prop:fault_equiv}}

\begin{proof}
    Suppose two faults $F, F'$ are equivalent. From Lem.~\ref{lem:prop_to_effect}, we multiply $F, F'$ each by some set of ECOs so that they are supported only on half-integer time-slices as well as the final time-slice. By assumption, their support on half-integer time-slices must be identical, since this corresponds to the measurements they flip. Furthermore, by assumption their support on the final time-slice may differ but only up to an element of ${\rm ISG}(T)$. By Cor.~\ref{cor:gauge_op_ISG}, this ISG element can be formed by a set of ECOs. Hence, we can form $F'$ from $F$ (and vice versa) via multiplication by ECOs.
    
    To show the other direction, we need to show that the effect due to a fault $F$ cannot be modified by multiplying it by ECOs. It suffices to note that faults combine linearly\footnote{We have not shown this but it follows directly from the results of Ref.~\cite{Delfosse_arXiv_2023_spacetime_code_clifford}}. Then, this reduces to showing that any ECO has a null effect. First, any measurement slice in time $t$ is already in ${\rm ISG}(t)$, so it has a null effect on the state, or any subsequent measurement on the state. Next, any propagator has the form 
    \begin{equation*}
        \eta^{t+1}(g^{t>}_{\rm prop}(p)) \eta^t(p).
    \end{equation*}
    Following the procedure in Lem.~\ref{lem:prop_to_effect} annihilates this operator, so it also has null effect. Hence, the effect due to some fault $F$ cannot be modified by multiplying by ECOs; the proposition follows.
\end{proof}

\subsubsection*{Proof of Prop.~\ref{prop:detectors_are_stabilizers}}

Instead of proving this from scratch, we will make use of results from Ref.~\cite{Pesah_arXiv_2025_FT_transformations}. First, we divide 
\begin{equation}
    \langle G_{\mathcal{C}} - \{ \mbox{ meas. dephasers} \}\rangle = \mathcal{G}_{\mathcal{C}}^{\rm cmm} \cup \mathcal{G}_{\mathcal{C}}^{\rm acmm},
\end{equation}
into two sets; $\mathcal{G}_{\mathcal{C}}^{\rm cmm}$ comprises operators whose support on integer time-slices commute with all measurement slices, and $\mathcal{G}_{\mathcal{C}}^{\rm acmm}$ the rest of them. By Lem.~\ref{lem:centralizer_from_propagator}, we have that $\mathcal{G}_{\mathcal{C}}^{\rm cmm}$ is supported entirely on integer time-slices, and coincides exactly with the gauge group of the subsystem spacetime code constructed in Ref.~\cite{Pesah_arXiv_2025_FT_transformations}. We can now cite a fact from their work:

\begin{fact}\label{fact:Pesahspackle}
    Let $m_1, ..., m_k$ be measurements at times $t_1, ..., t_k$, indexed in each half-integer location by $i_1,..., i_k$, such that $\{m_1, ..., m_k\}$ form a detector. Then 
    \begin{equation*}
        {\rm spackle}'(m_1, ..., m_k) = {\rm spackle}(m_1, ..., m_k) \times \prod_{j = 1}^k \eta^{t_j}(X_{i_j})
    \end{equation*}
    is in the center of $\mathcal{G}_{\mathcal{C}}^{\rm cmm}$. Furthermore, ${\rm spackle}'(m_1, ..., m_k)$ is formed by propagators and measurement slices.
\end{fact}

We can now proceed to prove Prop.~\ref{prop:detectors_are_stabilizers}. 

\begin{proof}
    Since $\mathcal{G}_{\mathcal{C}}^{\rm cmm} \subseteq \mathcal{G}_{\mathcal{C}}$ and $X_{i_j} \in G_{\mathcal{C}}$ as measurement dephasers, from the first part of Fact.~\ref{fact:Pesahspackle}, we obtain that the spackle of a detector is in $\mathcal{G}_{\mathcal{C}}$. 
    
    Now we need to show that it commutes with operators in $\mathcal{G}_{\mathcal{C}}^{\rm acmm}$. Let us consider one such operator, associated with some time $t + 0.5$, which has the form 
    \begin{equation*}
        g := \eta^{t+1}(U^{t>} X_n (U^{t>})^{\dag}) \eta^{t+0.5} \left( \prod_{j \in A} Z_j \right) \eta^{t}(X_n),
    \end{equation*}
    where $A$ indexes the set of measurements in time $t+0.5$ that anti-commute with $U^{t>} X_n (U^{t>})^{\dag}$. 
    
    We will show that this operator commutes with ${\rm spackle}(m)$ for any $m$ (hence it must commute with the product of many $m$'s). Let $m$ be associated with time $\tau + 0.5$. First, if $t < \tau$, then $g$ does not overlap with ${\rm spackle}(m)$. If $t > \tau$, then the support of ${\rm spackle}(m)$ on $t, t+1$ is of the form $\eta^{t+1}(U^{t>} p (U^{t>})^{\dag}) \eta^t(p)$ for some Pauli $p$, which commutes with $g$. 

    Finally, suppose $t = \tau$. Recall that $\pi^t({\rm spackle}(m)) = m$. If $U^{t>} X_n U^{t>}$ commutes with $m$, then $g$ will not have $Z$ support in the measurement location associated with $m$, and so $g$ will commute with ${\rm spackle}(m)$. Conversely, if $U^{t>} X_n U^{t>}$ anti-commutes with $m$ then $g$ will have $Z$ support in the measurement location associated with $m$, and also anti-commute in that location, so that overall $g$ will commute with ${\rm spackle}(m)$.

    Hence, $g$ commutes with any spackle of any $m$, so it must commute with products of them and hence any detector.
\end{proof}

\subsubsection*{Proof of Prop.~\ref{prop:detectors_and_errors_I}}

We now cite a fact from Ref.~\cite{Delfosse_arXiv_2023_spacetime_code_clifford} (albeit in the language of Ref.~\cite{Pesah_arXiv_2025_FT_transformations}). 

\begin{fact}[The stabilizer spacetime code]\label{fact:stab_SC}
    Let $m_1, ..., m_k$ be measurements forming a detector. Then any spacetime fault supported on integer time-slices and anti-commuting with 
    \begin{equation*}
        {\rm spackle}'(m_1, ..., m_k)
    \end{equation*}
    also violates the detector. 
\end{fact}

The proof of Prop.~\ref{prop:detectors_and_errors_I} now follows as a corollary.

\begin{proof}
    Fact.~\ref{fact:stab_SC} accounts for all data errors, so we only need to consider readout errors (the two sets of errors clearly combine linearly and form a basis for the possible errors). In this case, a readout error corresponds to a $Z$ on the half-integer locations where ${\rm spackle}(m_1, ..., m_k)$ has an $X$, so the claim follows.
\end{proof}

\section{Open boundary conditions}\label{app_sec:obc}

In Sec.~\ref{ssec:construction}, we introduced temporal boundaries to the gauged SSC in a rather ad hoc way. After disentangling the matter fields, we introduced an additional layer of gauge fields to the end of the circuit and shaved off the initial layer of matter fields. At first glance, this might rather arbitrary, but it allowed us to prove our claims about the correspondence between the elements of fault tolerance and gauge theory. Readers satisfied with that can safely skip this section without losing continuity.

On the other hand, one might worry that our procedure breaks the character of the gauge theory in some important way. We will now examine the gauging step a little more carefully to argue that these boundary conditions are in fact consistent and can be given a clean interpretation.

The key thing to note is that while we often say we gauge a theory on matter fields to obtain a theory on gauge fields, especially when it comes to $\mathbb{Z}_2$ theories on qubits, this actually comprises two steps: the actual gauging step, followed by the disentangling step. The gauging step imposes a Gauss law constraint that mixes both matter and gauge fields, and takes the form in Eq.~\ref{eq:matter_gauge_gauss}. The disentangling step performs cNOT operations between gauge and matter fields to transform the constraint into one that only lives on gauge fields, Eq.~\ref{eq:matter_gauge_gauss} -- equivalently, one fixes the state of the matter fields to be a product of $| + \rangle$ states.

For the bulk of the circuit we perform both in one fell swoop to obtain the Gauss law constraints. This does not affect our error identification step, as we still have gauge field configurations that we can identify with each circuit error in the bulk. However, let us now examine more carefully what happens at the boundaries.

For the initial time boundary, after the first gauging step, we have Gauss laws $S^0_{i, \alpha} = (\sigma^X)^0_{i, \alpha} (\tau^X)^{0.5}_{i, \alpha} = +1$. Fixing the matter field and disentangling also disentangles $(\tau^X)^{0.5}_{i, \alpha}$ from the rest of the gauge theory, which spoils our expectation that errors associated with it should be propagable into the bulk. To avoid this unpleasant conclusion, we simply avoid fixing and disentangling the matter field. The thing that we must modify is the error identification step, where we now treat the pair of gauge and matter fields as an effective degree of freedom, and we can identify $|11 \rangle, | 00\rangle$ with the presence of error and $ |10 \rangle, |01 \rangle$ with the absence of error. Combining this into an effective field yields the initial boundary condition we prescribed. The initial layer of gauge fields are never involved in redundancies so this does not change the structure of redundancies.

For the final time boundary, again we only perform the first gauging step, but do not disentangle the matter field, and treat it as a dynamical component of the gauged SSC. This is equivalent to introducing an additional layer of gauge fields and performing the full gauging and disentangling procedure. Since this final layer of gauge fields is actually a matter field in disguise, they are never involved in redundancies, so again our treatment of the final time boundary does not introduce new redundancies.

In other words, our prescription for boundary conditions is fully consistent with gauging the original theory and disentangling in the bulk but leaving the initial and final layers of matter fields dynamical. Furthermore, it does not introduce any new redundancies and so does not fundamentally change the structure of the theory.

Boundary conditions can also be given an operational interpretation, which can lead to more pleasing boundary conditions, we examine several cases when we look at examples.

\section{Spackle as error and the domain wall picture}\label{app_sec:spackle_domain_wall}

In Sec.~\ref{sec:circuits}, we outlined the standard way to identify errors with spacetime locations. In this short appendix, we point out that there is another possible way to do the error identification step. Instead of identifying single Pauli errors with single space time locations, we may choose to identify Pauli errors with their entire spackle. For instance, with a single-qubit identity wire with $T$ time slices, we may identify an $X$ error in the $t$-th time-step with the operator $\prod_{t' = t+1}^T \eta^{t'}(X)$. Such an operator would commute with all ECOs except $\eta^{t}(Z) \eta^{t+1}(Z)$, so it identifies this error with a particular gauge operator. This has several advantages: 
\begin{enumerate}
    \item Operationally, the spackles encode commutation relations with operators at every single time-step, so that errors and measurements at different time-steps already have the correct commutation relations without having to further propagate the errors. As such, one can read off from the spackle which detectors are triggered. Furthermore, the final time-slice of the spackle tells you how the final state is changed, so spackles contain just as much information about the elements of fault tolerance as the single Pauli errors.
    (Note that the object identified on the RHS in Fact.~\ref{fact:Pesahspackle} contains information about measurement flips in the half-integer time-slices, so everything can be directly read off of these.)
    \item This directly recovers the error identification used in the gauged SSC. For the 1D chain, this allows us to interpret the gauge field in the gauged SSC as a domain wall between the support of the spackle and the rest of the spacetime qubits.  
\end{enumerate}
Unfortunately, identifying errors in such a way renders the notion of a distance in the spacetime code not very meaningful, as there is no longer any clear correspondence between the weight of an error in the circuit and the weight of the operator we identify it with. Developing this perspective fully is left to future work. 

\section{The transpose code yields equivalent error configurations}\label{app_sec:gauss_laws}

The goal of this section is to prove Prop.~\ref{prop:spacetime_gauss_law}. 

To do so, we must show that Gauss laws map error configurations to equivalent error configurations, in the sense of Def.~\ref{def:fault_equiv}. Conveneniently, we have already proven Prop.~\ref{prop:fault_equiv}, so all we have to show is the following, for each Gauss law: A Gauss law takes an gauge field configuration, which is associated with an error, to a second gauge field configuration, which is associated with a second error. Then, we simply have to show that the second error and the first error are related by some set of ECOs. Since Gauss laws combine linearly, we just have to show this for each local Gauss law.

\begin{lem}[Gauss law for measurements]\label{app_lem:gauss_law_1}
    Consider the errors associated with two gauge field configurations related by the Gauss law associated with a measurement location $\sigma^{t+0.5}_{m, i}$. These errors are related by ECOs and are hence equivalent.
\end{lem}

\begin{proof}
    Let the measurement $\sigma^{t+0.5}_{m, i}$ be associated with the Pauli $q = \underline{q_1 q_2 ... q_n}$, where we use the underline to denote that $q_i$ is the Pauli in the $i$-th position of $q$. We suppose that $q_i = X$ or $Z$ for simplicity (this easily extends to $Y$ by linearity). Then, denote $q_i^c = X$ if $q_i = Z$ and vice versa. The associated measurement slice is $\eta^{t+1}(q)$ and the Gauss law is $\prod_{i = 1}^n \tau_{q_i^c, i}^{t+0.5}$. Note that from Lem.~\ref{lem:centralizer_from_propagator}, and the fact that measurement slices only exist for $t+1>0$, $\eta^{t}(q)$ is a product of the measurement slice and some set of elementary propagators.

    Finally, $\tau_{q_i^c, i}^{t+0.5}$ is identified with the error $q_i$ occurring just before $C^{t+0.5}$. This is the same as the error $\eta^t(q_i)$. Hence, the Gauss law maps errors to errors in the same way as $\eta^t(q)$, and the claim follows.
\end{proof}

\begin{lem}[Gauss law for spacetime locations]\label{app_lem:gauss_law_2}
    Consider the errors associated with two gauge field configurations related by the Gauss law associated with a spacetime location $\sigma^{t+0.5}_{\alpha, i}$, $\alpha = x, z$. These errors are related by ECOs and are hence equivalent.
\end{lem}

\begin{proof}
    Recall that for open boundary conditions, we have chosen the spacetime locations in the gauged SSC to run from $1$ to $T$ (it is straightforward to extend this claim to $0$ as well). We will consider the case of a spacetime location $\sigma^{t}_{x, i}$ as the $z$ case will follow analogously.
    
    We start by considering the Gauss law in the pure unitary case. Let $A$ be the set of single-qubit Paulis $\alpha_j$ that propagate into $X_i$, i.e.
    \begin{equation*}
        A = \{\alpha_j: X_i \in {\rm supp} U^{t-0.5} \alpha[j]_j (U^{t-0.5})^{\dag}; \alpha[j] \in \{X, Z\} \}.
    \end{equation*}
    The Gauss law is then
    \begin{equation*}
        \tau^{t+0.5}_{x, i} \prod_{j \in A} \tau^{t-0.5}_{\alpha[j], j}.
    \end{equation*}
    This maps errors to errors via 
    \begin{equation*}
        \eta^{t}(Z_i) \prod_{j \in A} \eta^{t-1}(\alpha[j]_j^c),
    \end{equation*}
    where similar to before, we use $\alpha[j]_j^c = X$ if $\alpha[j]_j = Z$ and vice versa. We now examine the terms $\alpha[j]$. Note that
    \begin{equation*}
        X_i \in {\rm supp} U^{t-0.5} \alpha[j]_j (U^{t-0.5})^{\dag} \iff \left[Z_i,  U^{t-0.5} \alpha[j]_j (U^{t-0.5})^{\dag} \right] \neq 0 \iff \left[(U^{t-0.5})^{\dag} Z_i U^{t-0.5}, \alpha[j]_j \right] \neq 0.
    \end{equation*}
    This implies that the Pauli in the $j$-th position of $(U^{t-0.5})^{\dag} Z_i U^{t-0.5}$ is simply $\alpha[j]_j^c$. We conclude that 
    \begin{equation*}
        \eta^{t}(Z_i) \prod_{j \in A} \eta^{t-1}(\alpha[j]_j^c) = \eta^t(Z_i) \eta^{t-1} \left((U^{t-0.5})^{\dag} Z_i U^{t-0.5}\right),
    \end{equation*}
    which, as a propagator, is a product of ECOs, and hence maps errors to equivalent errors. The claim follows for the pure unitary case.

    Next, we consider the case when there may be a measurements in time-slice $t-0.5$. Let 
    \begin{equation*}
        B = \{q_j \in \mathcal{M}^{t+0.5}: X_i \in {\rm supp} q_j\}.
    \end{equation*}
    We observe this is equivalent to 
    \begin{equation*}
        B = \{q_j \in \mathcal{M}^{t+0.5}: [Z_i, q_j] \neq 0 \}.
    \end{equation*}
    
    The Gauss law is then 
    \begin{equation*}
        \tau_{x, i}^{t+0.5} \prod_{l \in B} \tau_{m, l}^{t}  \prod_{j \in A} \tau^{t-0.5}_{\alpha[j], j}.,
    \end{equation*}
    where $\tau_{m, l}^{t}$ refers to the gauge field associated with $q_l$. Each $\tau_{m, l}^t$ is associated with a $Z$ operator on the half-integer time-slice readout location, so this maps errors to errors via
    \begin{equation*}
        \eta^t(Z_i) \prod_{l \in B} \eta^{t-0.5}(Z_l) \eta^{t-1} \left((U^{t-0.5})^{\dag} Z_i U^{t-0.5}\right),
    \end{equation*}
    From the equivalent definition of $B$, we see that this is a propagator, and thus is formed by ECOs, proving our claim for the case with unitaries and measurements.

    We note that the above proof also holds for input and output boundaries. For output boundaries, as it only required us to posit terms at $t-0.5, t-1$. Since we picked out input boundary to start from $0.5$, the arguments above all hold.
\end{proof}

\subsubsection*{Proof of Prop.~\ref{prop:spacetime_gauss_law} -- }

\begin{proof}
    Prop.~\ref{prop:spacetime_gauss_law} follows from Lems.~\ref{app_lem:gauss_law_1}, \ref{app_lem:gauss_law_2} and the linearity of the mapping betwen Gauss laws and ECOs.
\end{proof}

\section{Detectors are redundancies}\label{app_sec:detectors_are_redundancies}

The goal of this appendix is to formally prove Prop.~\ref{prop:detectors_are_redundancies}.

We point out in Lem.~\ref{lem:detector_redundancy_groups} that both detectors and redundancies have a group structure, allowing us to reasonably talk about isomorphism. Note that since this group structure is over $\mathbb{Z}_2$, it is simply a choice whether we want to regard this as a group or vector space. Then, the proof of Prop.~\ref{prop:detectors_are_redundancies} follows by explicit construction of a homomorphism in Prop.~\ref{prop:redundancies_to_detectors}, which we then show in Lem.~\ref{lem:injectivity_phi} and Lem.~\ref{lem:surjectivity_phi} is injective and surjective, leading to the claim.


\subsection*{Detectors and redundancies as groups}

We start by pointing out that detectors and redundancies both have a group structure isomorphic to the product of some number of copies of $\mathbb{Z}_2$. Let $\mathcal{M}$ be the set of all measurements in a circuit. We say that a measurement $m$ can have outcomes $o(m) = 0$ or $o(m) = 1$ as a random variable, depending on the input state of the circuit. Then detectors are subsets $D \subseteq \mathcal{M}$ such that the $\sum_{m \in D} o(m) = O_D$, where addition is taken modulo $2$, and $O_D = 0$ or $O_D = 1$ deterministically, independently of the input state of the circuit, when the circuit is noiseless. 

We can combine detectors to obtain detectors by taking the symmetric difference of the two sets, since the outcome,
\begin{equation*}
    \begin{aligned}
        \sum_{m \in D \oplus D'} o(m)
        &= \sum_{m \in D} o(m) - \sum_{m \in D \cap D'} o(m) + \sum_{m \in D'} o(m) - \sum_{m \in D \cap D'} o(m) \\
        &= \sum_{m \in D} o(m) + \sum_{m \in D'} o(m) \\ 
        &= O_D + O_{D'} =: O_{D \oplus D' }
    \end{aligned}
\end{equation*}
is also deterministic and hence a detector. Clearly, the empty set acts as an identity and each detector is its own inverse. Here, we may also regard $\oplus$ is as addition over $\mathbb{Z}_2$ variables indicating set membership, so this group is isomorphic to a subgroup of $\mathbb{Z}_2^{|\mathcal{M}|}$.
 
Similarly, redundancies in $G_{\mathcal{C}}$ also form a group. A set $R \subseteq G_{\mathcal{C}}$ is a redundancy if $\prod_{p \in R} p^{\otimes 2} = I$. Note the two fold tensor product is taken simply to get rid of potential sign and ordering ambiguities when taking the product, since two fold products of Paulis all commute. The symmetric difference of two redundancies is a redundancy, since 
\begin{equation*}
\begin{aligned}
    \prod_{p \in R \oplus R'} p^{\otimes 2} 
    &= \prod_{p_1 \in R } p_1^{\otimes 2} \cdot \prod_{p_2 \in R \cap R' } p_1^{\otimes 2} \cdot \prod_{p_3 \in R'} p_3^{\otimes 2} \cdot \prod_{p_4 \in R \cap R' } p_4^{\otimes 2} \\
    &= \prod_{p_1 \in R } p_1^{\otimes 2} \cdot \prod_{p_3 \in R'} p_3^{\otimes 2} \\
    &= I \cdot I = I.
\end{aligned}    
\end{equation*}
Again, the identity is given by the empty set $\varphi \subseteq G_{\mathcal{C}}$, and each redundancy is its own inverse. This group is isomorphic to a subgroup of $\mathbb{Z}_2^{|G_{\mathcal{C}}|}$. We note that $\mathbb{Z}_2^{|G_{\mathcal{C}}|}$ is strictly larger than $\mathbb{Z}_2^{|\mathcal{M}|}$.

For easy reference, we formally state this:
\begin{lem}[Detectors and redundancies form groups]\label{lem:detector_redundancy_groups}
    The set of detectors and the set of redundancies both form groups, with the group operation being the symmetric difference. Note this implies they are both $\mathbb{Z}_2$ linear spaces.
\end{lem}

Both groups are finite, so it suffices to show that we have injective homomorphisms from $\mathcal{D}$ to $\mathcal{R}$ and vice versa.

\subsection*{Mapping redundancies to detectors}

The goal is to show the following:
\begin{prop}[Redundancies to detectors]\label{prop:redundancies_to_detectors}
    Let $R$ be a redundancy, and let $\phi(R)$ be the set of measurements associated with measurement slices in $R$. Then $\phi(R)$ is a detector.
\end{prop}

We start with an obvious but important fact:
\begin{lem}[No propagator redundancies]\label{lem:no_prop_redundancies}
    There are no propagator redundancies.
\end{lem}

Let $R \subseteq G_{\mathcal{C}}$ be a redundancy. From Lem.~\ref{lem:no_prop_redundancies}, if $R$ is non-empty, then there must exist a non-empty subset $M_R \subseteq R$ comprising measurement slices, where $M_R \simeq \phi(R)$. We can of course identify $M_R$ with the measurements themselves. We will show that this is a detector.

This will take a bit of work. We record two easy lemmas that will help in the proof.

\begin{lem}\label{lem:no_meas_dephasers}
    $R$ does not contain measurement dephasers.
\end{lem}

\begin{proof}
    For every half-integer spacetime location, we have no operator (besides the measurement dephaser) in $G_{\mathcal{C}}$ with $X$-type support. The claim follows.
\end{proof}

\begin{lem}[Terminating a redundancy]\label{lem:terminating_redundancy}
    Let $t_s$ be the earliest time-slice for which any given operator in $R$ has some support. Then, the set of operators in $R$ supported on $t_s$ must comprise at least some measurement slice associated with time $t_s - 0.5$.

    Similarly, let $t_f$ be the latest time-slice for which any given operator in $R$ has some support. Then, the set of operators in $R$ supported on $t_f$ must comprise at least some measurement slice associated with time $t_f - 0.5$.
\end{lem}

\begin{proof}
    The ECOs supported on $t_s$ are either propagators associated with time $t_s - 0.5$ or $t_s + 0.5$, or measurements associated with $t_s - 0.5$. Since propagators supported on $t_s - 0.5$ also have support on $t_s - 1.5$, they cannot be contained in $R$.     Next, in order for $R$ to be a redundancy, the support on $t_s$ of the propagators associated with time $t_s + 0.5$, together with the measurements associated with time $t_s - 0.5$, must themselves multiply to $\pm I$. Since the support of said propagators on $t_s$ are necessarily independent, $R$ must also contain at least one measurement associated with time $t_s - 0.5$. A similar argument holds for $t_f$.
\end{proof}

In other words, redundancies must begin and end in measurements. 

We now establish some notation for the next couple of lemmas. Let $t_1 + 0.5 < ... < t_k + 0.5$ be the set of times associated with measurements in $M_R$. From Lem.~\ref{lem:terminating_redundancy}, we know that $t_1 +1 = t_s$ and $t_k + 1 = t_f$. We will denote by $q_j \in \mathcal{P}_n$ the product (arbitrarily fixing the order) of Paulis measured in time $t_j+0.5$. This is effectively a measurement itself, and we can denote by $o(q_j) = 0,1$ the outcome of this measurement. Subsequently, with this understanding, we will simply refer to $q_j$'s as measurements. Finally, we denote $G^{t>}_R \subseteq R$ the set of propagators in $R$ associated time-step $t+0.5$.

\begin{lem}[Base case]\label{lem:meas_to_ISG_base_case}
    Let $t'$ be an integer time such that $t_1 < t' \leq t_2 + 1$. Then 
    \begin{equation*}
    s_{t_1 \rightarrow t'} := (-1)^{o(q_1)} U^{t_1 + 1 \rightarrow t'} q_1 (U^{t_1 + 1\rightarrow t'})^{\dag} \in {\rm ISG}(t')    
    \end{equation*}
    Furthermore, $\eta^{t'}(s_{t_1 \rightarrow t'})$ is, up to a sign, equal to $\eta^{t_1+1}(q_1) \cdot \prod_{t = t_1}^{t'-1} \prod_{g \in G^{t>}_R} g$, i.e. a product of all ECOs in $R$ associated with times $t'-0.5$ or earlier other than possibly the measurement at $t_2 + 0.5$.
\end{lem}

\begin{proof}

We will prove this inductively.

Base case, i.e. $t' = t_1 + 1$: Since $q_1$ is measured at $t_1 + 0.5$, by Prop.~\ref{prop:ISG_update}, we have $(-1)^{o(q_1)} q_1 \in {\rm ISG}(t_1 + 1)$, where $o(q_1)$ is the sum of the outcomes from each measurement comprising $q_1$. 

Next, we prove the inductive step. Suppose  $s_{t_1 \rightarrow t'} := (-1)^{o(q_1)} U^{t_1 + 1 \rightarrow t'} q_1 (U^{t_1 + 1 \rightarrow t'})^{\dag} \in {\rm ISG}(t')$ and $\eta^{t'}(s_{t'})$ is, up to a sign, equal to $\eta^{t_1}(q_1) \cdot \prod_{t = t_1}^{t'-1} \prod_{g \in G^{t>}_R} g$ for some $t' < t_2 + 1$.

First, since $t' < t_2 + 1$, $R$ does not contain a measurement slice supported on $t'$. Hence, in order for the product of elements in $R$ to have no support on time $t'$, we require that $\prod_{g \in G^{t'>}_R} g = \pm g_{\rm prop}^{t'>}(s_{t'})$. Furthermore, by Lem.~\ref{lem:centralizer_from_propagator}, for the product of elements in $R$ to have no support on the half-integer time-slice $t' + 0.5$, we require that 
\begin{equation*}
    s_{t_1 \rightarrow t'+1} = U^{t'+0.5} s_{t_1 \rightarrow t'} (U^{t'+0.5})^{\dag} = \pm \prod_{g \in G^{t'+0.5}_R} g \cdot \eta^{t'}(s_{t'})
\end{equation*} 
commutes with all measurements in the full circuit at time $t' + 0.5$. By Prop.~\ref{prop:ISG_update}, this also implies that $s_{t_1 \rightarrow t'+1} \in {\rm ISG}(t'+1)$. 

For the case $t' = t_2$, we note that multiplying an element of ${\rm ISG}(t' + 1)$ by measurements carried out in $t' + 0.5$ yields an element of ${\rm ISG}(t' + 1)$ up to a sign.  

The claim follows.

\end{proof}



Next, we should show that this holds for all the measurement times in $R$. The proof proceeds very similarly to Lem.~\ref{lem:meas_to_ISG_base_case}.

\begin{lem}[Redundancies must propagate ISG elements I]\label{lem:meas_to_ISG_inductive_case}
    \begin{equation*}
    s_{t_j + 1 \rightarrow t'} := (-1)^{\sum_{l = 1}^j o(q_l)} U^{t_{j-1} + 1 \rightarrow t'} q_j U^{t_{j-1} + 1 \rightarrow t_{j} + 1} q_{j-1} \cdot ... \cdot U^{t_1 \rightarrow t_2} q_1 (U^{t_1 + 1 \rightarrow t_j + 1})^{\dag} \in {\rm ISG}(t')
    \end{equation*}
    for all $t_j = t_2, ..., t_{k-1}$ and $t_j < t' \leq t_{j+1} + 1$. 
    \begin{equation*}
        \eta^{t'}(s_{t_j + 1 \rightarrow t'}) = \pm \prod_{l=1}^{j} \eta^{t_l}(q_l) \cdot \prod_{t = t_1}^{t'-1} \prod_{g \in G^{t>}} g,
    \end{equation*}
    i.e. a product of all ECOs in $R$ associated with times $t'+0.5$ or earlier, except for possibly measurements in $R$ at $t'+0.5$.
\end{lem}

\begin{proof}
    We will prove this inductively. The base case is given by Lem.~\ref{lem:meas_to_ISG_base_case}.

    Next, we suppose the claim holds for some $t_j < t_{k-2}$. First, we note that by Prop.~\ref{prop:ISG_update},
    \begin{equation*}
            s_{t_{j+1} + 1 \rightarrow t_{j+1} + 1} = (-1)^{o(q_{j+1})} q_{j+1} s_{t_{j} \rightarrow t_{j+1}+1} \in {\rm ISG}(t_{j+1} + 1),
    \end{equation*}
    and 
    \begin{equation*}
        \eta^{t_{j+1} + 1}(s_{t_{j+1} + 1 \rightarrow t_{j+1} + 1} ) = \eta^{t_{j+1} + 1}(q_{j+1} ) \cdot \eta^{t_{j+1} + 1}(s_{t_j+1 \rightarrow t_{j+1} + 1}) = \pm \prod_{l=1}^{j+1} \eta^{t_l}(q_l) \cdot \prod_{t = t_1}^{t_{j+1}} \prod_{g \in G^{t>}_R} g
    \end{equation*}
    follows because $\eta^{t_{j+1} + 1}(q_{j+1} ) \in R$ by assumption.

    We further suppose the claim holds for $s_{t_{j+1} + 1 \rightarrow t'}$. Then this implies that it holds for $s_{t_{j+1} + 1 \rightarrow t'+1}$ for $t' \leq t_{j+2} + 1$ by the same reasoning as the inductive step in Lem.~\ref{lem:meas_to_ISG_base_case}.

    The full claim follows.
\end{proof}

We remark that any given (non-terminal) measurement in a detector may anti-commute with other measurements in a circuit, say in the next time-step. However, in order for this measurement to be a non-terminal part of a detector, there must exist an ISG element which also anti-commutes with the same measurement, such that their product can propagate freely. As per Cor.~\ref{cor:gauge_op_ISG}, every ISG element is encoded by a product in $G_{\mathcal{C}}$. The fact that, rather counter-intuitively, the products of measurements (up to propagators) commute through the measurements in every time-step of the circuit contained in bulk of the detector is a reflection of this.

Finally, we can prove Prop.~\ref{prop:redundancies_to_detectors}.

\begin{proof}
    From Lem.~\ref{lem:meas_to_ISG_inductive_case}, we have that 
    \begin{equation*}
    s_{t_{k-1} + 1 \rightarrow t_k+1} := (-1)^{\sum_{l = 1}^{k-1} o(q_l)} U^{t_{k-1} + 1 \rightarrow t_k + 1} q_{k-1} U^{t_{j-1} \rightarrow t_{j}} q_{j-1} \cdot ... \cdot U^{t_1 \rightarrow t_2} q_1 (U^{t_1 \rightarrow t_j})^{\dag} \in {\rm ISG}(t_k+1),
    \end{equation*}
    and furthermore that $\eta^{t_k+1}(s_{t_{k-1} + 1 \rightarrow t_k+1})$ is the product of every element of $R$ except the final measurement $q_k$. 

    By Prop.~\ref{prop:ISG_update}, we have 
    \begin{equation*}
    (-1)^{o(q_k)} q_k \cdot s_{t_{k-1} + 1 \rightarrow t_k+1} \in {\rm ISG}(t_k+1).
    \end{equation*}
    However, since $R$ is a redundancy, we must have that $(-1)^{o(q_k)} q_k \cdot s_{t_{k-1} + 1 \rightarrow t_k+1} \propto \pm I$. Since it is in an ISG, it must be equal to $I$. This implies that the sign coming from the outcomes, which is $(-1)^{\sum_{l = 1}^k o(q_l)}$ must be deterministic. Hence, the measurements $q_1, ..., q_k$ must form a detector. This is exactly the set of measurements given in $\phi(R)$, so we conclude that $\phi(R)$ is a detector.
\end{proof}

We have shown that the map $\phi$ indeed maps redundancies to detectors. For completeness, we point out that it is a homomorphism.
\begin{lem}[Homomorphism property]\label{lem:phi_homomorphism}
    $\phi$ is a homomorphism.
\end{lem}
\begin{proof}
    A redundancy $R_i$ can be split into two disjoint sets of operators, $R_i = M_i \cup P_i$, where $M$ are measurement slices and $P$ are propagators. Then $\phi(R_1 \oplus R_2) = \phi(M_1 \oplus M_2 \cup P_1 \oplus P_2) = M_1 \oplus M_2 = \phi(R_1) \oplus \phi(R_2)$, identified with measurements in the circuit, so $\phi$ preserves the group operation. Hence it is a homomorphism.
\end{proof}

\subsection*{The mapping is an isomorphism}

To promote $\phi$ to an isomorphism, we prove that it is injective and surjective.

\begin{lem}[Injectivity]\label{lem:injectivity_phi}
    $\phi$ is injective.
\end{lem}
\begin{proof}
    Let us consider two redundancies $R_1, R_2$ containing the same set of measurements, i.e. $\phi(R_1) = \phi(R_2)$. Then they must differ by a set of propagators $P = R_1 \oplus R_2$. Since the symmetric difference of two redundancies must be a redundancy by Lem~\ref{lem:detector_redundancy_groups}, and a non-empty set of propagators alone cannot be a redundancy by Lem.~\ref{lem:no_prop_redundancies}, $P$ must be empty. Hence $R_1 = R_2$, and the claim follows.
\end{proof}

To show surjectivity, we start with a lemma.
\begin{lem}\label{lem:squishing_first_measurement}
    Let $D = \{m_1, ..., m_k\}$ be a detector. Then, $(-1)^{o(m_1)} U^{t_1 + 1 \rightarrow t_2 + 1} m_1 (U^{t_1 + 1 \rightarrow t_2 + 1})^{\dag} \in {\rm ISG}(t_2 + 1)$.
\end{lem}

\begin{proof}
    It suffices to show that in order for $D$ to be a detector, there cannot exist $m'$ at some time $t'$ where $t_1 < t' \leq t_2$ that anti-commutes with $U^{t_1 + 1 \rightarrow t'+1} m_1 (U^{t_1 + 1 \rightarrow t'+1})^{\dag}$. 
    
    Recall our convention that $m_1, ..., m_k$ are labelled time-sequentially with $t_1 < t_2 < ... < t_k$. We observe that in order for $D$ to be a detector, ${\rm ISG}(t_2 + 1)$ must have an element of the form $(-1)^{o(m_1)} s$, where $s$ is some Pauli operator that does not depend on any other measurement in the circuit. If there existed no such operator, $D$ cannot be a detector, as it must include at least one other measurement before $m_2$ for the total parity to be deterministic.  

    Suppose there exists some $m'$ that anti-commutes with $U^{t_1 + 1 \rightarrow t'+1} m_1 (U^{t_1 + 1 \rightarrow t'+1})^{\dag}$. WLOG, let $m'$ be the first such anti-commuting measurement, so that $(-1)^{o(m_1)} U^{t_1 + 1 \rightarrow t'} m_1 (U^{t_1 + 1 \rightarrow t'})^{\dag} \in {\rm ISG}(t')$. We may write down a basis for ${\rm ISG}(t')$ with this operator as the only element depending on $o(m_1)$.
    
    By Prop.~\ref{prop:ISG_update}, if there exists no element $q \in {\rm ISG}(t' + 1)$ such that $q  U^{t_1 + 1 \rightarrow t' + 1} m_1 (U^{t_1 + 1 \rightarrow t' + 1})^{\dag}$ commutes with $m'$, then ${\rm ISG}(t' + 1)$ can have no dependence on $o(m_1)$. Hence, there must exist such an element. However, since all ISG elements in the circuit must depend on some prior measurement (since ${\rm ISG}(0) = \{\}$ by assumption), any element of ${\rm ISG}(t'+1)$ that depends on $o(m_1)$ must also depend on the outcome of such prior measurements, meaning $D$ cannot be a detector. Hence, by contradiction, we cannot have any such $m'$ anti-commuting with $U^{t_1 + 1 \rightarrow t'+1} m_1 (U^{t_1 + 1 \rightarrow t'+1})^{\dag}$. The claim follows.
\end{proof}

\begin{cor}\label{cor:squishing_first_measurement}
    Let $D = \{m_1, ..., m_k\}$ be a detector. Then, $\eta^{t_2 + 1}(U^{t_1 + 1 \rightarrow t_2 + 1} m_1 (U^{t_1 + 1 \rightarrow t_2 + 1})^{\dag})$ is a product of ECOs.
\end{cor}

The above allows us to show surjectivity.

\begin{lem}[Surjectivity]\label{lem:surjectivity_phi}
    $\phi$ is surjective.
\end{lem}

\begin{proof}
    We will show this inductively. Base case: Suppose $D = \{m_1, m_2\}$ is a detector. Then it follows from the argument of Lem.~\ref{lem:squishing_first_measurement} that $(-1)^{o(m_1)} U^{t_1 + 1 \rightarrow t_2 + 1} m_1 (U^{t_1 + 1 \rightarrow t_2 + 1})^{\dag}$ is the only element of ${\rm ISG}(t_2 + 1)$ that depends only on $o(m_1)$, hence, for $D$ to be a detector, we must have $m_2 = \pm U^{t_1 + 1 \rightarrow t_2 + 1} m_1 (U^{t_1 + 1 \rightarrow t_2 + 1})^{\dag}$. Then from Cor.~\ref{cor:squishing_first_measurement}, we know that $\eta^{t_2+1}(m_2)$ is the product of the measurement slice $\eta^{t_1 + 1}(m_1)$ and propagators, so that together with the measurement slice associated with $m_2$ we have a redundancy $R$ such that $\phi(R) = D$.

    Now suppose that for any Clifford circuit with a detector comprising $k$ or fewer (but larger than $2$) measurements, we can form a redundancy that maps to this detector under $\phi$. Now, consider a circuit $\mathcal{C}$ containing a detector $D = \{m_1, ..., m_{k+1}\}$. From the argument of Lem.~\ref{lem:squishing_first_measurement}, we know that there are no measurements between $m_1, m_2$ that anticommute with $m_1$. Hence, we may construct an alternative circuit $\mathcal{C}'$ which is identical to $\mathcal{C}$ everywhere except time-step $t_1 + 0.5, t_2 +0.5$. In $t_1 + 0.5$, we remove the measurement of $m_1$, and in time-step $t_2 + 0.5$, we replace the measurement of $m_2$ with measurement of $m_2 U^{t_1 + 1 \rightarrow t_2 + 1} m_1 (U^{t_1 + 1 \rightarrow t_2 + 1})^{\dag}$. Clearly, in this alternative circuit, we will have $D' = \{m_2 U^{t_1 + 1 \rightarrow t_2 + 1} m_1 (U^{t_1 + 1 \rightarrow t_2 + 1})^{\dag}, m_3, ..., m_{k+1}\}$ as a detector of length $k$, so that there exists a $R'$ with $\phi(R') = D'$. Now, to construct an analogous $R$ for the original circuit, we note that in deforming $\mathcal{C}$ to $\mathcal{C}'$ we have not changed any of the propagators and measurement slices after time $t_2$ at all, so we can keep all these operators in $R$ (specifically the operators in $\mathcal{C}$ which are identified with them). The measurement slice in $R'$ associated with $m_2 U^{t_1 + 1 \rightarrow t_2 + 1} m_1 (U^{t_1 + 1 \rightarrow t_2 + 1})^{\dag}$ can be replaced with the ECOs described in Cor.~\ref{cor:squishing_first_measurement} along with the ECO associated with the measurement of $m_2$. Then clearly this new $R$ will remain a redundancy, and we will have $\phi(R) = D$. Hence, the full claim follows.
\end{proof}

\subsection*{Detectors are Wilson loops}

Proof of Prop.~\ref{prop:wilson_like_detectors}.

\begin{proof}
    Let us consider an error in the support of a redundancy, i.e. the error identified with any of the gauge fields in the support of the redundancy. There are two such kinds of errors: measurement errors associated with measurement slices, or propagation errors. The claim is obvious for measurement errors. For propagation errors, first consider a single $X$-type error on the $i$-th qubit at some time $t + 0.5$, which are associated with a single elementary propagator $\eta^{t}(Z_i) \eta^{t+1}(U^{t>} Z_i (U^{t>})^{\dag})$. For now, assume $\eta^{t}(X_i) \eta^{t+1}(U^{t>} X_i (U^{t>})^{\dag})$ is not part of the same redundancy. Now since the unitaries must preserve commutation relations, propagating the $Z_i$ operator through each time-step always obtains operators that anti-commute with the propagated version of the $X_i$ operator.  
    In a redundancy, the propagation of this $Z$ operator must terminate on measurements, which may occur at different time-steps. Nevertheless, this implies that if we propagated all the associated measurements to the end of the circuit, we will obtain an operator that anti-commutes with the $X_i$ error propagated all the way to the end of the circuit. Hence, this error must cause a parity flip in this set of measurements and thus violate the detector. If the redundancy contains the ECO of the form $\eta^{t}(X_i) \eta^{t+1}(U^{t>} X_i (U^{t>})^{\dag})$ then we consider the propagation of the $Y_i$ operator, which similarly anticommutes with $X_i$. A similar reasoning holds for $Z$ errors. Finally, we can take linear combinations of such errors to obtain the claim.
\end{proof}

\section{The action of a foliated computation}\label{app_sec:foliated_computation}

To show that the action of a foliated computation for a CSS-like circuit is correct, it suffices to consider a particular subgraph $\Gamma^{t+0.5}$ of the $X$-type part of the gauged SSC. 

First, we say the Clifford operation at time $t+0.5$ comprises a unitary $U$ and measurements of commuting pure $X$ operators $q^X_1, ..., q^X_{k^X}$ and pure $Z$ operators $q^Z_1, ..., q^Z_{k^Z}$. Some useful additional notation: set $A_i$ to be the index set of qubits in $U X_i U^{\dag}$ that have non-trivial support, and $B^{t+0.5}_i$ the analogous index set for $U Z_i U^{\dag}$. For an index $i$ for which a measurement has non-trivial support we will simply say $i \in q^{X/Z}_j$.

Next, we introduce three layers of nodes (corresponding to qubits) in $\Gamma^{t+0.5}$. 
\begin{itemize}
    \item In layer $t+0.5$, we introduce qubits $(i, t+0.5)$ for each wire in the circuit.  We also introduce qubits $(a_j, t+0.5)$ for each $Z$-type measurement.
    \item In layer $t+1$, we introduce qubits $(i, t+1)$ for each wire in the circuit. We also introduce qubits $(a_j, t+1)$ for each $X$-type measurement.
    \item In layer $t+1.5$, we introduce qubits $(i, t+1.5)$ for each wire in the circuit.
\end{itemize}

The intra-layer edges correspond to:
\begin{itemize}
    \item For $t+0.5$, the qubit $(a_j, t+0.5)$ is connected to the qubit $(i, t+0.5)$ if and only if we have $[q^Z_j, U^{t+0.5} X_i (U^{t+0.5})^{\dag}] \neq 0$, i.e. $q^Z_j$ does not commute with the propagated version of $X_i$.
    \item For $t + 1$, the intra-layer edges are given by connecting the $X$-type-measurement-associated qubit with the wires in the support of the Pauli they measure. 
    \item For $t + 1.5$, we do not introduce any intra-layer edges.
\end{itemize}

Finally, the inter-layer edges are given by:
\begin{itemize}
    \item For layers $t + 0.5$ and $t + 1$, we have edges between qubits $(i, t+0.5)$ and the qubits $(j, t + 1)$ for $j \in A^{t+0.5}_i$.
    \item Between layers $t + 1$ and $t + 1.5$, we have edges between qubits $(i, t+1)$ and $(i, t+1.5)$. 
\end{itemize}
Finally, we foliate this subgraph $\Gamma^{t+0.5}$ by leaving the wire qubits $(i, t+0.5)$ in an arbitrary state and preparing all other qubits in a $ | + \rangle$ state, and then applying $cZ$ gates along all edges. We refer to this resource state as $ | \Gamma^{t+0.5}_b \rangle$.

We will now show that this subgraph carries out the correct Clifford operation associated with the time-slice $t+0.5$. 

First, we observe that the operators
\begin{equation*}
    \begin{aligned}
        \widetilde{X}_i^{t+0.5} &= (-1)^{\zeta[t+0.5]_i^x} X^{t+0.5}_i Z^{t+1}_i, \\
        \widetilde{Z}_i^{t+0.5} &= (-1)^{\zeta[t+0.5]_i^z} Z^{t+0.5}_i,\\
    \end{aligned}
\end{equation*}
where $\zeta[t+0.5]_i^{x/z} = 0,1$ are arbitrarily chosen, are not stabilizers of $| \Gamma^{t+0.5}_b \rangle$, and they commute with all stabilizers. We identify these with the Paulis before application of the Clifford operation, and identify the analogous operators with $t \rightarrow t+1$ with the Paulis after application of the Clifford operation. 

We claim that we apply the Clifford operation by measuring all qubits in layers $t+0.5, t+1, t+1.5$ and all qubits $(a_j, t+1.5)$. We refer to the post measurement state as $| \Gamma^{t+0.5}_f \rangle$. To show this, we must show:
\begin{enumerate}
    \item \textit{The correct operators are measured --} The final state will have stabilizers that are proportionate to $\widetilde{q}_i^{X/Z}$, where the sign includes a term that comes from measuring the associated ancilla.
    \item \textit{Logical degrees of freedom survive --} If $\prod_{j \in A_i} \widetilde{X}^{t+0.5}_j$ stabilizes $| \Gamma^{t+0.5}_b \rangle$ and $\prod_{j \in A_i} X^{t+0.5}_j$ commutes with all $q_l^Z$ then $\prod_{j \in A_i} \widetilde{X}^{t+0.5}_j$ will (up to a sign) stabilize $| \Gamma^{t+0.5}_f \rangle$. Similarly for any possible $Z$-type stabilizer.
    \item \textit{Anti-commuting degrees of freedom are removed -- } If $\prod_{j \in A_i} \widetilde{X}^{t+0.5}_j$ stabilizes $| \Gamma^{t+0.5}_b \rangle$ and $\prod_{j \in A_i} X^{t+0.5}_j$ anti-commutes with some $q_l^Z$ then $\prod_{j \in A_i} \widetilde{X}^{t+0.5}_j$ (with any sign) will not be a stabilizer of $| \Gamma^{t+0.5}_f \rangle$. Similarly for any possible $Z$-type stabilizer.
\end{enumerate}
We note that we do not have to consider combinations of $X, Z$-type operators due to our restriction to Clifford circuits.

\textit{The correct operators are measured --} For $X$-type measurements, we note that 
\begin{equation*}
    X_{a_j}^{t+1} \prod_{i \in q^X_j} Z_i^{ t+1} \sim X_{a_j}^{t+1} \prod_{i \in q^X_j} X_i^{ t+1.5},
\end{equation*}
stabilizes $| \Gamma^{t+0.5}_b \rangle$, so measuring $X_{a_j}^{t+1}$ imposes the stabilizer 
\begin{equation*}
    (-1)^{o(x_{a_j}^{t+1})} \prod_{i \in q^X_j} X_i^{ t+1.5}
\end{equation*}
on $| \Gamma^{t+0.5}_f \rangle$. 

For $Z$-type measurements, the argument is more subtle. First, note that $((i, t+0.5), (j, t+1)) \in E$ if and only if $j \in A_i$. We also have  $((i, t+0.5), (j, t+1)) \in E$ if and only if $i \in {\rm supp} U^{\dag} Z_j U$.

Now $((a_j, t+0.5), (i, t+1)) \in E$ if and only if $[q_j^Z, U X_i U^{\dag}] \neq 0$. Requiring also the CSS property, this is equivalent to $i \in {\rm supp} U^{\dag} q_j^Z U$. So we have stabilizers,
\begin{equation*}
    X^{t+0.5}_{a_j} \prod_{i \in U^{\dag} q_j^Z U} Z_i^{t+0.5}.
\end{equation*}
However, we also have stabilizers, 
\begin{equation*}
    Z_j^{t+1.5} X_j^{t+1} \prod_{i \in U^{\dag} Z_j U} Z_i^{t+0.5} \prod_{l: j \in q_l^X} Z_{a_l}^{t+1}
\end{equation*}
Multiplying through the support of $q_j^Z$, 
\begin{equation*}
    \begin{aligned}
        &\prod_{p \in q_j^Z} \left( Z_p^{t+1.5} X_p^{t+1} \prod_{i \in U^{\dag} Z_p U} Z_i^{t+0.5} \prod_{l: p \in q_l^X} Z_{a_l}^{t+1} \right) \\
        &= \prod_{p \in q_j^Z} Z_p^{t+1.5} X_p^{t+1} \prod_{i \in U^{\dag} q_j^Z U} Z_i^{t+0.5},
    \end{aligned}
\end{equation*}
where the disappearance of the last term comes from the fact that $q_j^Z, q_l^X$ must commute for all $j, l$. Finally, multiplying this by the stabilizer we started with shows that 
\begin{equation}
    X_{a_j}^{t+0.5} \prod_{p \in q_j^Z} X_p^{t+1} Z_p^{t+1.5}
\end{equation}
is a stabilizer of $|\Gamma_b^{t+0.5} \rangle$. Hence, measuring all the qubits in layer $t+0.5, t+1$ in the $X$ basis imposes a stabilizer 
\begin{equation}
    (-1)^{o(x_{a_j}^{t+0.5})} \prod_{i \in q_j^Z} (-1)^{o(x_p^{t+1})} Z_p^{t+1.5}
\end{equation}
on $|\Gamma^{t+0.5}_f \rangle$. Hence, we see that the correct stabilizers are imposed with the ancillas contributing to the signs in the right way.

\textit{Logical degrees of freedom survive --} Now, suppose the initial state is stabilized by some $s^X = \prod_{i \in S} X_i$ such that $[U s^X U^{\dag}, q_j^Z] = 0$ for all $j$. Then after applying the $cZ$ gates, we have 
\begin{equation*}
    \prod_{i \in S} X_i^{t+0.5} \prod_{j \in A_i} Z_j^{t+1}
\end{equation*}
as a stabilizer. The condition $[U s^X U^{\dag}, q_j^Z] = 0$ tells us that 
\begin{equation*}
    \prod_{j \in A_i} Z_j^{t+1} X_j^{t+1.5} 
\end{equation*}
Hence, we have 
\begin{equation*}
    \prod_{i \in S} X_i^{t+0.5} \prod_{j \in A_i} X_j^{t+1.5}
\end{equation*}
as a stabilizer of $| \Gamma_f^{t+0.5} \rangle$, and measuring gives 
\begin{equation*}
    \prod_{i \in S} (-1)^{o(x_i^{t+0.5})} \prod_{j \in A_i} X_j^{t+1.5}
\end{equation*}
as a stabilizer of the final state. Recognizing that we have the correct connectivity for propagating $Z$-type operators as described in Sec.~\ref{sssec:banach-tarski-prop}, an almost identical argument holds for $Z$-type stabilizers.

\textit{Anti-commuting degrees of freedom are removed --} This follows from having established that measurements are carried out correctly, so that on the final remaining layer of qubits we cannot have incompatible stabilizers.

\bigskip

We have shown that the foliated computation construction carries out the correct action on qubits in some layer $t+0.5$ and teleports them to $t+1.5$. To obtain the full foliated computation for the entire circuit, we can concatenate this procedure. We first imagine doing this for layers $0.5, 1, 1.5$. Then we introduce qubits in $t = 2, 2.5$ and carry out the $cZ$ gates. These cZ gates commute with all measurements on layers $0.5$ and $1$ that have already been done, and the subsequent measurements on layers $1.5, 2$ also commute with all cZ gates and $X$ measurements that were done in the original layers. Proceeding in this manner we see that we always have the right state, up to measurement outcomes, on the final half-integer layer, which finally gets us to the output boundary.

\section{Deferred derivations from Sec.~\ref{sec:applications}}\label{app_sec:deferred_proofs_applications}

\subsection*{Outcome distribution for complete noise ansatz}

In this section, we derive Eq.~\ref{eq:outcome_dist_complete}. Let $b \in \mathbb{Z}_2^n$. We denote $Z[b]:= \otimes_{j = 1}^n Z_j^{b_j}$. The initial state may be written
\begin{equation}
    \rho_0 = \frac{1}{2^n} \otimes_{j=1}^n (I_j + Z_j) = \frac{1}{2^n} \sum_{b \in \mathbb{Z}_2^n} Z[b].
\end{equation}
Errors on the state preparation can only (up to a global phase), flip $|0 \rangle's$ to $|1 \rangle's$. Hence only the $Z$-type Pauli eigenvalues in state preparation matter. These can be indexed by $\lambda^S_b := \lambda^S_{Z[b]}$ for some $b \in \mathbb{Z}_2^n$. This gives the effect of $(\Lambda^S)^{\otimes n}$ on the initial state as, 
\begin{equation}
\begin{aligned}
    (\Lambda^S)^{\otimes n}(\rho_0) 
    = \frac{1}{2^n} \sum_{b \in \mathbb{Z}_2^n} \lambda_b^S Z[b].
\end{aligned}
\end{equation}

The measurement outcomes of the circuit can be indexed by a bit string $k \in \mathbb{Z}_2^n$ representing the outcome of each $Z$ basis measurement. The POVM in the noiseless case is 
\begin{equation}
    E_k = \frac{1}{2^n} \sum_{b \in \mathbb{Z}_2^n} (-1)^{k \cdot b} Z[b],
\end{equation}
with the product forming $Z[b]$ picking up a minus sign everywhere $k$ is non-zero. In the noisy case, we also have that only the $Z$-type eigenvalues matter, so that,  
\begin{equation}
    \begin{aligned}
        {\rm Tr} (E_k (\Lambda^M)^{\otimes n}(\cdot ))  =\frac{1}{2^n} \sum_{b \in \mathbb{Z}_2^n} (-1)^{k \cdot b} \lambda^M_{b} {\rm Tr}(Z[b] (\cdot) ).
    \end{aligned}
\end{equation} 

The action of the noisy Clifford unitaries $\widetilde{C}$ may be written,
\begin{equation}
    \begin{aligned}
        \widetilde{C}(\cdot) 
        &= \frac{1}{2^n} \sum_{p_0 \in \overline{\mathcal{P}}_n} \left( \prod_{j = 0}^{T-1} \lambda_{ U^{0 \rightarrow j} p_0 (U^{0 \rightarrow j})^{\dag}}^{T-j+0.5} \right) {\rm Tr} \left( (\cdot) p_0 \right) U^{0 \rightarrow T} p_0 (U^{0 \rightarrow T})^{\dag}.
    \end{aligned}
\end{equation}

These allow us to work out the probability distribution over $k$ for some $\mathbf{x}$ representing the negative log Pauli eigenvalues of the circuit elements is,
\begin{equation}
    \begin{aligned}
        F_k(\mathbf{x}) 
        = \frac{1}{4^n} \sum_{b, b' \in \mathbb{Z}_2^n} \left( \prod_{j = 0}^{T-1} \lambda_{ U^{0 \rightarrow j} Z[b'] (U^{0 \rightarrow j})^{\dag}}^{T-j+0.5} \right) (-1)^{k \cdot b} \lambda^M_{b} \lambda^S_{b'} {\rm Tr}(Z[b]  U^{0 \rightarrow T} Z[b'] (U^{0 \rightarrow T})^{\dag} ) 
    \end{aligned}
\end{equation}

\printbibliography

\end{document}